\documentclass[11pt]{article}

\usepackage{sn-preamble} 
\usepackage{tikz}
\usepackage{verbatim}
\usepackage{cancel}
\usepackage{enumitem}
\usepackage{bm}
\usepackage{caption}
\usepackage{changepage}
\usepackage{microtype}

\newcommand{\ignore}[1]{}

\newcommand{\dtv}{\mathrm{d}_{\mathrm{TV}}}
\newcommand{\dmax}{d_{\max}}

\newcommand{\alg}{\textsf{ALG}}
\newcommand{\DNFLearn}{\textsc{DNF-Learn}}

\newcommand{\genterm}{\textsc{Generate-Long-List-of-Terms}}
\newcommand{\genlistofterms}{\textsc{Generate-List-of-Terms}}
\newcommand{\prune}{\textsc{Prune}}
\newcommand{\expand}{\textsc{Expand}}
\newcommand{\leap}{\textsc{Noise}}
\newcommand{\findfar}{\textsc{Find-Far-Point}}
\newcommand{\exactlearn}{\textsc{Exact-Learn}}

\newcommand{\mix}{\mathrm{mix}}
\newcommand{\quasipoly}{\mathrm{quasipoly}}

\title{
DNF Learning via Locally Mixing Random Walks
}
\author{Josh Alman\thanks{Email: \texttt{josh@cs.columbia.edu}}
\\ \textsl{Columbia University} \and Shivam Nadimpalli \thanks{Email: \texttt{shivamn@mit.edu}} \\ \textsl{MIT} \and 
Shyamal Patel\thanks{Email: \texttt{shyamalpatelb@gmail.com}} \\ \textsl{Columbia University} \and Rocco A. Servedio\thanks{Email: \texttt{ras2105@columbia.edu}}\\ \textsl{Columbia University}}
\date{}

\begin{document}

\pagenumbering{gobble}
\maketitle

\begin{abstract}

We give two results on PAC learning DNF formulas using membership queries in the challenging ``distribution-free'' learning framework, where learning algorithms must succeed for an arbitrary and unknown distribution over $\{0,1\}^n$.

\begin{itemize}

\item [(1)] We first give a quasi-polynomial time ``list-decoding'' algorithm for learning a \emph{single term} of an unknown DNF formula.  More precisely, for any target $s$-term DNF formula $f = T_1 \vee \cdots \vee T_s$ over $\zo^n$ and any unknown distribution ${\cal D}$ over $\zo^n$, our algorithm, which uses membership queries and random examples from ${\cal D}$, runs in $\quasipoly(n,s)$ time and outputs a list $\calL$ of candidate terms such that with high probability some term $T_i$ of $f$ belongs to $\calL$.

\item [(2)] We then use result (1) to give a $\quasipoly(n,s)$-time algorithm, in the distribution-free PAC learning model with membership queries, for learning the class of size-$s$ DNFs in which all terms have the same size. Our algorithm learns using a DNF hypothesis.

\end{itemize}

The key tool used to establish result (1) is a new result on ``locally mixing random walks,'' which, roughly speaking, shows that a random walk on a graph that is covered by a small number of expanders has a non-negligible probability of mixing quickly in a subset of these expanders.
\end{abstract}

\newpage

\setcounter{tocdepth}{2}
\tableofcontents

\ 

\

\newpage

\pagenumbering{arabic}

\section{Introduction}

\subsection{Background}

A \emph{DNF (Disjunctive Normal Form) formula} is an OR of ANDs of Boolean literals.
The problem of learning an unknown DNF formula was proposed in the original pioneering work of Leslie Valiant introducing the PAC learning model \cite{Valiant:84}, motivated by the fact that DNFs are a natural form of knowledge representation as well as a fundamental type of Boolean function that arises in countless contexts.   
In the ensuing decades since \cite{Valiant:84}, DNF learning has emerged as a touchstone challenge in the field of computational learning theory, as witnessed by the intensive research effort that has aimed at developing efficient DNF learning algorithms across a wide range of different models.

In the original ``distribution-independent PAC learning'' version of the problem that was posed by Valiant in \cite{Valiant:84}, there is an unknown and arbitrary probability distribution ${\cal D}$ over $\zo^n$, and a sample oracle which, each time it is queried, provides an independent pair $(\bx,f(\bx))$ where $\bx \sim {\cal D}$ and $f$ is the unknown $s$-term DNF formula to be learned. 
The goal of the learning algorithm is to efficiently construct an (efficiently evaluatable) hypothesis function $h: \zo^n \to \zo$ such that with probability at least $1-\delta$,\footnote{Standard and simple arguments \cite{HKL+:91} show that it is enough to achieve this with any fixed constant probability such as $0.9$ or $0.1$; hence, in most of our ensuing discussion, we will gloss over the role of the confidence parameter.} the error rate $\Pr_{\bx \sim {\cal D}}[h(\bx) \neq f(\bx)]$ of the hypothesis $h$ is at most $\epsilon.$
A further goal is for the DNF learning algorithm to construct a DNF formula as its hypothesis, i.e.~so-called \emph{proper} DNF learning.

While DNF learning has been intensively studied, the original distribution-independent learning problem described above turns out to be very challenging, and progress on it has been limited; we describe some of this progress below.

\medskip

\noindent {\bf Distribution-independent PAC learning of $s$-term DNF (with and without queries).}
In his original paper \cite{Valiant:84}, Valiant used feature expansion and a simple reduction to learning conjunctions to give an $n^{O(s)}$-time algorithm for learning $s$-term DNF in the distribution-independent PAC model described above.
For small (e.g.~constant or poly-logarithmic) values of $s$ this is still the fastest runtime known, but progress has been made on mildly exponential algorithms which run faster than $n^{O(s)}$ when $s$ is at least some particular polynomial in $n$.  
More precisely, Bshouty \cite{bsh96} gave an an approach based on decision-list learning which PAC learns $s$-term DNF in time $2^{O(\sqrt{n \log s} (\log n)^{3/2})}$; Tarui and Tsukiji \cite{TaruiTsukiji:99} gave a boosting-based approach which PAC learns $s$-term DNF in time $2^{O(\sqrt{n} \log n \log s)}$; and Klivans and Servedio \cite{KlivansServedio:04jcss} used polynomial threshold function representations and linear programming to give a $2^{O(n^{1/3} \log s \log n)}$-time algorithm.

Given the seeming difficulty of distribution-independent DNF learning in the original PAC model of random examples only, it is natural to give more power to the learning algorithm by allowing it to make \emph{membership queries} (i.e.~black box queries to the function that is being learned).  Indeed, Valiant's original paper already posed the problem of distribution-independent PAC learning of DNF formulas using both random examples and membership queries. 
Starting with the work of Blum and Rudich \cite{BlumRudich:95}, a variety of different algorithms have been given for distribution-independently learning $s$-term DNF over $\{0,1\}^n$ using random examples and membership queries in time $\poly(n) \cdot 2^s$ \cite{Bshouty:95,Kushilevitz:97,Bshouty:97,BBB+:00}. (In fact, all of these algorithms work in the slightly more general but closely related \emph{exact learning} model of learning from membership queries and ``equivalence queries.'') While these algorithms can learn $O(\log n)$-term DNF formulas in polynomial time, once the number of terms $s$ is more than polylogarithmic they run in more than quasipolynomial time, and even for modest polynomial values of $s$ they do not run as fast as the algorithms of \cite{bsh96,TaruiTsukiji:99,KlivansServedio:04jcss}.

\medskip

\noindent
{\bf Relaxing the problem:  learning in easier models, and learning restricted subclasses of DNFs.}   An extensive body of work has considered easier variants of the DNF learning problem. 
One popular way to relax the problem is to require that the learner only succeed under the \emph{uniform distribution} on $\zo^n$; this is a significantly easier setting than the distribution-free setting because terms of length $k$ are only satisfied with probability $2^{-k}$, and hence, intuitively, long terms can be ``safely ignored.''
Building on this intuition, Verbeurgt \cite{ver90} used attribute-efficient disjunction learning algorithms to give an $n^{O(\log(s/\eps))}$-time uniform-distribution algorithm for learning $s$-term DNF using only random examples. 
A more recent result of De et al.~\cite{de2014learning} (which we will discuss further in \Cref{sec:technical-overview}) implies that $s$-term DNF formulas can be learned in $n^{O(\log(s/\eps))}$ time under the uniform distribution given only uniform random satisfying assignments of the DNF.

In a celebrated result, Jackson \cite{jac97} showed that his Harmonic Sieve algorithm, which combines ``smooth'' boosting techniques with Fourier-analytic techniques, learns $s$-term DNF formulas  in $\poly(s,n,1/\eps)$ time under the uniform distribution using membership queries.  (The precise polynomial running time was subsequently improved in \cite{BJT:99,KlivansServedio:03ml,Gavinsky:03}.) 
Subsequent work adapted the Harmonic Sieve to perform uniform-distribution learning given \emph{uniform quantum superpositions} of labeled examples \cite{BshoutyJackson:99,JTY:02}, or given labeled examples generated according to a \emph{uniform random walk} over the Boolean hypercube $\{0,1\}^n$ \cite{BMO+:05}. We remark that all of these algorithms make essential use of Fourier techniques which crucially rely on the product structure of the underlying distribution, and do not seem to be useful for distribution-independent learning.  

Finally, many papers (including \cite{Valiant:84,Hancock:91,AizensteinPitt91,Berggren93,KushilevitzRoth:93,AizensteinPitt:95,PR95,ABKKPR98,dommispit99,BBB+:00}) have considered algorithms for learning restricted subclasses of DNF formulas, typically using membership queries.
We remark that these works generally consider rather heavily restricted subclasses of DNF formulas that satisfy some quite specific structure; despite much effort, only limited success has been achieved in learning rich subclasses of DNF formulas in the distribution-independent setting.\footnote{Many additional works relax the problem in two ways, by considering learning restricted subclasses of DNF formulas under restricted distributions, most typically the uniform distribution \cite{PagalloHaussler89,HancockMansour:91,KMP:94,Khardon:94,Verbeurgt:98,SakaiMaruoka:00,Servedio:04iandc,JacksonServedio:06long, JLSW11:dam, Sellie:08,Sellie-2009}.} 

\subsection{Our Results}

In this paper we give two main algorithmic results towards learning DNF formulas in the distribution-independent PAC model using membership queries. 

The first of these results (see \Cref{thm:list-decoding}) achieves a weaker algorithmic goal than PAC learning, but it applies to general $s$-term DNF formulas in the general distribution-independent PAC + MQ setting.
The second of these results (see \Cref{thm:exact-k-learn-intro}) gives a distribution-independent PAC + MQ learning algorithm for quite a broad subclass of $s$-term DNF formulas.  We emphasize that both of these algorithms run in time $\quasipoly(n,s)$, which is much faster than either the $2^{\tilde{O}(n^{1/2} \log s)}$- or $2^{\tilde{O}(n^{1/3} \log s)}$-type  running times of \cite{TaruiTsukiji:99,bsh96,KlivansServedio:04jcss} or the $\poly(n,2^s)$ running times of \cite{BlumRudich:95,Bshouty:95,Kushilevitz:97,Bshouty:97,BBB+:00} for the general distribution-independent PAC learning problem.

\begin{theorem} [List-decoding a single term.]
\label{thm:list-decoding}
Let $f=T_1 \vee \cdots \vee T_s$ be any unknown $s$-term DNF over $\zo^n$ and let $\calD$ be any unknown distribution over $\zo^n$. 
Let $p := \Pr_{\bx \sim \calD}[f(\bx) = 1].$
There is an algorithm {\sc List-Decode-DNF-Term} in the PAC + MQ model that runs in time\footnote{We remark that the ${\frac 1 p}$ running time dependence is inherent for any algorithm, since $1/p$ calls to the sample oracle are required even to find a single satisfying assignment of $f$; note also that if $p \ll \eps$ then it is trivial to learn $f$ to accuracy $\eps$ under ${\cal D}$.}
\[
	\frac{\Theta(1)}{p} + (ns)^{O(\log(ns))}
\] 
and outputs a list $\calL$ of at most $(ns)^{O(\log(ns))}$ terms, such that with probability at least 0.99 some term $T_i$ of $f$ appears in $\calL$.
\end{theorem}

\begin{remark} [An easy decision tree analogue, and an easy uniform-distribution analogue] \label{remark:DT}
We remark that a similar-in-spirit theorem for \emph{decision trees} with at most $s$ leaves is trivial. This is because any decision tree with at most $s$ leaves must contain some leaf at depth $\leq \log s$, so without using any examples at all it is possible to construct a list, consisting of  all $\approx n^{\log s}$ many possible root-to-leaf paths of length $\leq \log s$, which is guaranteed to contain some actual root-to-leaf path that is present in the target decision tree.  However, this argument completely breaks down for DNF formulas, since, for example, it may be the case that the number of terms $s$ is $\poly(n)$ but all terms have length $\Theta(n).$

We further remark that a simple routine used in \cite{de2014learning} gives a quantitatively similar runtime and performance guarantee to \Cref{thm:list-decoding}, in a setting where the algorithm is given \emph{uniform random satisfying assignments of the DNF $f$}. Perhaps surprisingly, \Cref{thm:list-decoding} shows that (using membership queries) a similar result can be achieved even in the much more challenging \emph{distribution-independent} setting. We discuss the (significant) technical challenges that must be overcome to obtain our distribution-independent result in \Cref{sec:technical-overview}.

\end{remark}

\Cref{thm:list-decoding} does not yet give us a full-fledged algorithm for PAC learning DNF, since instead of constructing a high-accuracy hypothesis for the unknown DNF formula $f$ it only constructs a list of candidates one of which is guaranteed to be one of the terms in $f$. It should also be noted that the term $T_i$ which {\sc List-Decode-DNF-Term} finds need not necessarily have
significant probability of being satisfied under the distribution ${\cal D}$.  
On the positive side, to the best of our knowledge \Cref{thm:list-decoding} is the first quasipolynomial time algorithm that ``does anything non-trivial'' towards learning an unknown and arbitrary $s$-term DNF formula in the general distribution-independent PAC + MQ model.  
We are optimistic that \Cref{thm:list-decoding}, and the tools used in its proof, may lead to renewed progress on the DNF learning problem.  

As evidence for this optimism, our second main result shows that \Cref{thm:list-decoding} can be used (along with a range of other ingredients and ideas) to learn a rich subclass of DNF formulas.
We say that an $s$-term DNF formula $f = T_1 \vee \cdots \vee T_s$ is an \emph{exact-DNF} if there is some value $1 \leq k \leq n$ such that each term $T_i$ contains exactly $k$ distinct literals.
Our second main result is a quasipolynomial-time 
learning algorithm for this class:

\begin{theorem} [Learning exact-DNF]
\label{thm:exact-k-learn-intro}
There is a distribution-independent PAC+MQ\footnote{We remark that we believe that our techniques should extend to learning exact-$k$ DNFs in the stronger exact learning model of Angluin.} algorithm which, given any $\eps > 0$, learns any unknown size-$s$ exact-DNF to accuracy $\eps$ using a hypothesis which is a DNF of size $O(s \log(1/\eps))$. The algorithm runs in time 
\[
\exp \left( \log^{O(1)}(ns) \right) \cdot \poly(\eps^{-1}).
\] 
\end{theorem}

We feel that the class of exact-DNF formulas is quite a rich subclass of DNF formulas compared with other subclasses of DNF formulas for which efficient distribution-independent learnability has been established in prior work. 
To the best of our knowledge, prior to \Cref{thm:exact-k-learn-intro}
the fastest known algorithm for learning size-$s$ exact DNF formulas was either the $\poly(n,2^s)$-time algorithms of \cite{BlumRudich:95,Bshouty:95,Kushilevitz:97,Bshouty:97,BBB+:00} 
or the $2^{O(n^{1/3} \log n \log s)}$-time algorithm of \cite{KlivansServedio:04jcss}, depending on the value of $s$.  In contrast, we give a quasipolynomial-time algorithm.
We remark that \Cref{thm:exact-k-learn-intro} is easily seen to also give a quasipolynomial-time algorithm for the class of ``approximately exact'' DNF formulas, i.e.~$s$-term DNF in which all terms have lengths that differ from each other by at most an additive $\polylog(n,s)$.

Finally, we remark that a potentially interesting aspect of our work is that \Cref{thm:list-decoding} and \Cref{thm:exact-k-learn-intro} are proved using algorithmic techniques and technical machinery (in particular, random walks and analysis of expansion and rapid mixing) which are quite different from the ingredients that have been employed in prior DNF learning algorithms.
We further note that there has recently been some algorithmic work \cite{liu2024locally} that analyzes Markov chains before they have fully mixed to their stationary distribution and leverages properties of these Markov chains for various algorithmic applications. While our techniques and analyses are very different from those of \cite{liu2024locally}, our results reinforce the high-level message that there may be unexpected algorithmic applications of studying the behavior of random walks even before they mix to the stationary distribution.

\subsection{Technical Overview} \label{sec:technical-overview}

To motivate our list-decoding algorithm, we begin by recalling an approach due to \cite{de2014learning} for learning $s$-term DNF formulas in a ``relative-error'' sense under a variant of the uniform-distribution learning model.  
In the problem considered in \cite{de2014learning}, the learning algorithm is given independent examples drawn uniformly from $f^{-1}(1)$, where $f$ is an unknown $s$-term DNF over $\zo^n$, and the goal is to construct a hypothesis DNF $h$ such that ${\frac {|h^{-1}(1) \ \triangle f^{-1}(1)|} {|f^{-1}(1)|}}$ is small. 
A key ingredient in the \cite{de2014learning} algorithm for this problem is a routine that 
produces a list $\calL$ of at most $n^{O(\log s)}$ terms such that with high probability some term in ${\cal L}$ is among the terms of $f$ (note that this is quite similar to our list-decoding guarantee).
In the \cite{de2014learning} context this can be done in a simple way by leveraging the following easy observations: (i) if $\bx^1,\dots,\bx^{2 \log n}$ are i.i.d.~uniform satisfying assignments of $f$, then with probability at least $(1/s)^{2 \log n}$ all $2 \log n$ of the $\bx^i$'s satisfy some single fixed term $T_j$ of $f$; and moreover, (ii) if all $2 \log n$ of the $\bx^i$'s satisfy the same term $T_j$ of $f$, then with high probability the \emph{longest} term that is satisfied by all of $\bx^1,\dots,\bx^{2 \log n}$ is precisely $T_j$ (intuitively, this is because any variable that does not occur in $T_j$ is equally likely to be set to 0 or 1 in each $\bx^i$).

The correctness of the \cite{de2014learning} routine relies heavily on the assumption that the examples $\bx^i$ are \emph{uniform random} satisfying assignments of $f$; in the general distribution-free setting that we consider $\calD$  can be arbitrary, so we cannot hope to directly use this simple approach.  
But can anything along these lines be done?
In an effort to adapt the approach to our setting, consider the graph $G[f^{-1}(1)]$ which is the subgraph of the Boolean hypercube $\{0,1\}^n$ that is induced by the satisfying assignments of $f$.
As a jumping-off point, we note that if a suitable \emph{random walk} over that graph (which can be performed using membership queries) had small mixing time, then given an initial satisfying assignment $y \in \{0,1\}^n$ (which can easily be obtained from the distribution ${\cal D}$), running independent random walks starting from $y$ would give us a way to draw samples from the stationary distribution over $G[f^{-1}(1)]$. Since the maximum degree of any vertex in $G[f^{-1}(1)]$ is $O(n)$, intuitively the stationary distribution is not too different from the uniform distribution over satisfying assignments, and indeed this would enable us to  simulate the approach from \cite{de2014learning}.

Of course, the graph $G[f^{-1}(1)]$ need not have small mixing time; even for simple DNF formulas such as $f(x)= (x_1 \land x_2) \lor (\overline{x_1} \land \overline{x_2})$, the graph may not even be connected.  
In this example, though, a random walk starting from a satisfying assignment in $T_1 := x_1 \land x_2$ will rapidly mix to a uniform satisfying example of $T_1$, and likewise for $T_2$.
So at least in this simple example, taking multiple independent random walks starting from an arbitrary point $y \in f^{-1}(1)$ and following the \cite{de2014learning} approach would indeed yield a term of $f$.

\subsubsection{Locally Mixing Random Walks}
Given the previous discussion, it is natural to explore an approach of choosing an initial satisfying assignment $y$ and performing random walks starting from $y$.
We explore this approach by analyzing what we call \emph{locally mixing random walks}. In particular, we are interested in the following general scenario: We are given a large graph $G = (V,E)$, with relatively small maximum degree $\dmax$, whose vertices are covered by sets $A_1, \dots A_s$,  where each induced subgraph $G[A_i]$ is a $\theta$-expander. Roughly speaking, we would like to understand how long it may take a random walk starting from an arbitrary $v \in V(G)$ to mix in $G \left[ \bigcup_{i \in S} A_i \right]$, for some $S \subseteq [s]$.   (As suggested above, in our intended use case the graph $G$ will be $G[f^{-1}(1)]$, the induced subgraph of $\zo^n$ whose vertices are the satisfying assignments of the DNF $f$, and each set $A_i$ will be $T_i^{-1}(1)$, the set of satisfying assignments of the $i$-th term.  Each $G[A_i]$ is a $(n-k_i)$-dimensional hypercube, where $k_i$ is the length of term $T_i$, so each $G[A_i]$ is indeed a good expander, and the maximum degree $\dmax$ of $G$ is $O(n)$, which is not too large.)

To achieve our desired term list-decoding, it would roughly suffice to show that for every $v \in V(G)$ there exists a set $S \subseteq [s]$ such that a random walk starting at $v$ rapidly mixes to the stationary distribution of $G \left[\bigcup_{i \in S} A_i \right]$, i.e.~in at most $n^{O(\log(s))}$ steps.\footnote{Note that while the above statement is about mixing to the stationary distribution, the ratio $\pi(u)/\pi(v)$ of the probabilities of any two vertices $u,v$ under the stationary distribution $\pi$ of $\smash{G \left[ \bigcup_{i \in S} A_i \right]}$ is at most $\dmax$, which we have assumed to be relatively small; so the stationary distribution over $\smash{G \left[ \bigcup_{i \in S} A_i \right]}$ is an adequate proxy for the uniform distribution over $\smash{G \left[ \bigcup_{i \in S} A_i \right]}$ for our purposes.}  Unfortunately, such a statement is not true, as witnessed by the following example:

\begin{example} \label{example:wacky}
Let $G$ be a graph of maximum degree $O(n)$ which is structured as follows: $G$ contains a degree-$n$ expander $G_1$ of size $2^n$ and a vertex-disjoint degree-$2n$ expander $G_2$ of size $2^{2n}.$  $G$ also contains an additional vertex $v$ that is the root of two vertex-disjoint (other than the root) $(n+1)$-regular trees $T_1$ and $T_2$, each of depth $n/\log n$.  The $2^n$ leaves of the first tree are the $2^n$ vertices of $G_1$, and the $2^n$ leaves of the second tree are $2^n$ of the vertices of $G_2$. Note that the graphs $T_1 \cup G_1$ and $T_2 \cup G_2$ are good expanders, and that the vertices of $G$ are covered by the sets $A_1=V(T_1)  = V(T_1) \cup V(G_1)$, $A_2=V(T_2) \cup V(G_2).$  

Now consider a random walk starting from $v$: it has an equal chance of going to either $T_1$ or $T_2$ in its first step, and given that it goes to $T_i$, it has a $1-o(1)$ chance of mixing in $G[A_i]$ within $\poly(n)$ steps.  In fact, given that it initially goes to $T_i$, with probability $1-o(1)$ it will take $2^{\Omega(n)}$ steps before it enters the other tree or the other expander. Thus after, say, $2^{\sqrt{n}}$ steps, a random walk will have probability roughly $1/2$ of being distributed roughly uniformly over $G[A_1]$, and probability roughly $1/2$ of being distributed roughly uniformly over $G[A_2]$. But this does not correspond to the stationary distribution over $G[A_1]$, $G[A_2]$, or $G=G[A_1 \cup A_2].$
\end{example}

Given examples such as the one above, we must lower our requirements; fortunately, we can do this in such a way that meeting those lowered requirements is still useful for a \cite{de2014learning}-type approach.
We only require that for each $v \in V$, there is some event $E_v$ such that a random walk $v:= \bY_0, \dots \bY_\ell$ starting at $v$ mixes rapidly in some $G \left[ \bigcup_{i \in S} A_i \right]$ \emph{conditioned on $E_v(\bY_1, \dots \bY_\ell)$ occuring}, and that moreover $E_v(\bY_1, \dots \bY_\ell)$ occurs with non-negligible probability.\footnote{Returning to \Cref{example:wacky}, we could take the event $E_v$ to be that the first step of the random walk takes it into $T_1$.}
This is made precise in our Local Mixing theorem, which is the main technical result undergirding \Cref{thm:list-decoding}:

\begin{theorem}[Local Mixing, informal statement]
\label{thm:local-mixing-intro}
	Let $G = (V,E)$ be a graph with maximum degree $\dmax$, and let $A_1, ..., A_s$ be such that each $G[A_i]$ is a $\theta$-expander and $\bigcup_i A_i = V$. For any starting vertex $v \in V$, there exists a value $\ell = (s\dmax \log(|V|) \log(1/\eps)/ \theta)^{O(\log(s))}$; an event $E_v$; and an index $j \in [s]$, such that a lazy random walk $v := \bY_0, \bY_1, \dots, \bY_\ell$ in $G$ satisfies 

	\begin{itemize}
	
	\item [$(i)$] $\Pr[E_v(\bY_0, \dots \bY_\ell)] \geq (sd_{\max} \log(|V|) \log(1/\eps)/\theta)^{-O(\log(s))}$; and 
	\item [$(ii)$] the distribution of $\bY_\ell \,|\, E_v(\bY_0, \dots \bY_\ell)$ has total variation distance at most $\eps$ from the uniform distribution over $A_j$.
	
	\end{itemize}
\end{theorem}

We say that the \emph{local mixing time} of a vertex $v \in V$ is precisely the minimum number of steps needed from a random walk starting at $v$ so that an event $E_v$ exists satisfying properties $(i)$ and $(ii)$. (Note that we defined this with respect to being uniform over a set $A_i$ rather than a set $\bigcup_{i \in S} A_i$, as the two are roughly equivalent once we introduce events $E_v$, up to $\dmax$ and $s$ factors.)

Before sketching the proof of \Cref{thm:local-mixing-intro}, we note that the framework of our local mixing problem bears a resemblance to Spielman and Teng's local clustering algorithm \cite{spielman2013local}. At a very high level, \cite{spielman2013local} observe that under the assumptions of \Cref{thm:local-mixing-intro}, if $G$ does not have too many edges connecting any two sets $A_i$ and $A_j$, then a random walk of moderate length starting from a random point in any $A_i$ will mix in $A_i$ with high probability. Stated more quantitatively, they observe that if $S \subseteq V(G)$ is such that $(S, \overline{S})$ forms a $\wt{O}(\lambda^2)$-sparse cut, then a random walk of length $\wt{O}(1/\lambda)$ starting from a random vertex $\bv \in S$ will stay in $S$ with high probability. On the other hand, if $G[S]$ is a $\lambda$-spectral expander
then a walk of length $\wt{O}(1/\lambda)$ should mix over $G[S]$ (cf. \Cref{sec:prelims}).

Our local mixing problem is different in two significant respects: $(a)$ the starting point of our random walk may be chosen adversarially, and $(b)$ we get no guarantee on the number of edges between expanders. At a very high level, our approach works by circumventing these differences so that we can apply the observation of Spielman and Teng.

To start, we discuss how to handle adversarial starting points, i.e.~$(a)$ above. Indeed, we'll argue that it is enough to show that for every $A_i$, \emph{most} points in $A_i$ have small local mixing time. Roughly speaking, this is done by arguing that if $R \subseteq V$ is the set of starting points that have large local mixing time, then there must be many edges leaving the set $R$ (see \Cref{sec:goodsetslocalmixing}). 
Here is a sketch of the argument:  First, note that $|R \cap A_j| \geq |R|/s$ for some $j$ since $A_1,\dots,A_s$ are a cover of the vertices. $R$ cannot include most of the vertices in $A_j$, since most vertices in $A_j$ have small local mixing time. But we can now use the fact that $A_j$ is an expander, to show that there are many edges in $A_j$ which leave $R$. It follows that a random walk starting in $R$ will quickly leave $R$ and arrive at a vertex which has low local mixing time. Continuing the random walk from that vertex, we get a combined walk that mixes to the stationary distribution over some $A_j$, as desired. With some effort, one can use an argument of this sort to show that it suffices to prove that for every $A_i$ we have that an $(s d_{\mathrm{max}} \log(|V|) \log(1/\eps)/\theta)^{-O(\log(s))}$-fraction of points have small local mixing time.

We now turn to handle $(b)$. The proof of \Cref{thm:local-mixing-intro} first shows that it suffices to handle the situation in which $A_1 \sqcup \cdots \sqcup A_s$ are a disjoint partition of the vertices of $V$ (see \Cref{sec:disjointification}). For simplicity, in the following high-level discussion we will consider the case where each $A_i$ has the same size. While this is not without loss of generality --- the general case where expanders may have different sizes is more involved --- this simpler setting will highlight many of the key ideas that we use.

At a high level, our goal will be to find a new set of $\alpha$-expanders $B_1, ..., B_{s'}$, where each $B_i = \bigcup_{j \in S_i} A_j$ for not necessarily disjoint sets
$S_i \subseteq [s]$, that cover the graph $G$. In light of our previous discussion about point $(a)$, it will suffice to show that every $B_i$ has at least a $(s d_{\mathrm{max}} \log(|V|) \log(1/\eps)/\theta)^{-O(\log(s))}$ fraction of points with small local mixing time.

To achieve this, we start by noting that if there are many edges between $A_i$ and $A_j$, then $A_i \cup A_j$ must have large expansion. In this case, we can prove the statement by recursing on 
\[A_1, \dots, A_{i-1}, A_{i+1}, \dots, A_{j-1}, A_{j+1}, \dots, A_s, A_i \cup A_j,\]
 which is a partition of $V$ into $s-1$ rather than $s$ pieces. We terminate when each expander corresponds to a sparse cut and we can apply the argument of \cite{spielman2013local} to all of them. Making this quantitative does yield a bound on the local mixing time, albeit naively we can only hope for a rather weak one: Note that by Cheeger's inequality each $A_i$ is a $O(\theta^2)$-spectral expander. If $(A_i, \overline{A_i})$ is a $\wt{\Omega}(\theta^2)$ sparse cut, then it follows that there exists a $j \in [s]$ such that there are at least $\wt{\Omega}(\theta^2 |A_i|/s)$ edges between $A_i$ and $A_j$. At best, we will only be able to get that $A_i \cup A_j$ has (combinatorial) expansion at least $\wt{\Omega}(\theta^2/ s)$. Thus, in each of the $s$ steps of the recursion, we may square the mixing time bound, which would yield a bound of the form $(s d_{\mathrm{max}} \log(|V|) \log(1/\eps)/\theta)^{2^{O(s)}}$.
 
To get to our desired quasi-polynomial bound, we make two improvements to the above argument. For the first improvement, we note that we are losing considerably by switching between combinatorial and spectral expansion. Indeed, each time we do this, we must square the expansion value. That, said because the graphs we are working with are covered by expanders, we can apply the higher order Cheeger inequality \cite{kwok2013improved, louis2012many, lee2014multiway} to avoid this repeated squaring. This brings the bound down to $(s d_{\mathrm{max}} \log(|V|) \log(1/\eps)/\theta)^{O(s)}$.

 To motivate the second improvement, consider the following example

\begin{example}
Suppose we have a graph $G$ that is covered by four expanders $A_1, A_2, A_3, A_4$ of the same size. There are $\theta^2 |A_1|$ edges from $A_1$ to $A_2$, $\theta^4 |A_1|$ edges from $A_1$ to $A_3$, and $\theta^8 |A_1|$ edges from $A_1$ to $A_4$, and these are  the only edges in the graph between different expanders. Moreover, suppose that $A_1 \cup A_2$ is a $\theta^2$-expander and $A_1 \cup A_2 \cup A_3$ is a $\theta^4$-expander.

Let us trace through the execution of our recursive scheme on this example. We would first merge $A_1$ and $A_2$ and then recurse on $A_1 \cup A_2, A_3, A_4$. We would then form $A_1 \cup A_2 \cup A_3$, and recurse on $A_1 \cup A_2 \cup A_3, A_4$. Finally, we would form $A_1 \cup A_2 \cup A_3 \cup A_4$ and terminate.

A closer analysis, though, shows that we can in fact output the expanders $B_1 = A_1 \cup A_2 \cup A_3, B_2 = A_4$. Indeed, since $(A_3, \overline{A_3})$ forms a $\theta^4$-sparse cut, it follows by \cite{spielman2013local} that most points in $A_3$ (say, at least half of the points) have a small local mixing time. Since all the expanders have the same size, at least $1/6$ of the points in $A_1 \cup A_2 \cup A_3$ have small local mixing time. Thus, there is no need for the last recursive call.
\end{example}

This motivates the following approach to use fewer than $s$ recursive calls. Recall that for every set $A_i$, either $(A_i, \overline{A_i})$ is a $\wt{O}(\theta^2)$-cut, or there exists a $j_i$ such that there are at least $\wt{\Omega}(\theta^2 |A_i|/s)$ edges between $A_i$ and $A_{j_i}$. In the first case, we have that most points in $A_i$ have small local mixing time by Spielman and Teng's observation. In this case, if in any later stage of the recursion we try to merge a set $\bigcup_{\ell \in S} A_\ell$ with $A_i$, then we have that the merged set contains ``enough'' points with a small local mixing time. Now we consider the second case, in which we must handle the sets $A_i$ that do not correspond to a sparse cut; for each such $i$, we merge $A_i$ and $A_{j_i}$ and recurse. (We remark that if $A_{j_1} = A_2$ and $A_{j_2} = A_3$, then we would merge all three of $A_1, A_2, A_3$ for the next iteration. In general, one can consider making a graph with vertices corresponding to the sets $A_1, ..., A_s$, and an edge between $A_i$ and $A_j$ iff there are at least $\theta^2 |A_i|$ edges between the pair. We then merge the connected components in this graph to form the new expanders in the next level of the recursion.\footnote{We note that in the general case, when the expanders can have different sizes, determining how to merge the expanders becomes somewhat more involved and one must look carefully at the sizes of the expanders. 
}) After one iteration of this process, we now have the guarantee that the number of expanders that don't contain a large number of vertices with small local mixing time has been halved. Thus, we only need to recurse $O(\log(s))$ times before terminating, which gives us a bound of $(s d_{\mathrm{max}} \log(|V|) \log(1/\eps)/\theta)^{O(\log(s))}$ as desired.

\subsubsection{``List Decoding'' of size-$s$ DNFs:  Sketch of \Cref{thm:list-decoding}}
\label{sec:list-decoding-DNFs}

We now turn to discuss how locally mixing random walks yield a list-decoding type guarantee for DNFs, i.e.~how \Cref{thm:local-mixing-intro} implies \Cref{thm:list-decoding}.
A proof achieving the claimed ${\frac 1 p} + (ns)^{O(\log(ns))}$ time bound of \Cref{thm:list-decoding} involves some modifications of intermediate technical results  in the proof of \Cref{thm:local-mixing-intro}, and is given in \Cref{sec:speedy-mixing}.
That said, a simple ``black-box'' application of \Cref{thm:local-mixing-intro} already gives a qualitatively similar result which is quantitatively slightly weaker, i.e.~an algorithm running in time ${\frac 1 p} + (ns)^{O(\log(s) \log(n))}$; we give that proof below.

\medskip

\noindent \emph{Proof of \Cref{thm:list-decoding} variant with a ${\frac 1 p} + (ns)^{O(\log(s) \log(n))}$ time bound.}
Recall that in the context of \Cref{thm:list-decoding}, the graph $G$ is the induced subgraph of the Boolean hypercube graph $\zo^n$ 
that is induced by the vertex set $f^{-1}(1)$.  
We have that $\dmax=O(n)$, and each $A_i$ is a subcube corresponding to the satisfying assignments of term $T_i$ in $f$.  
Recalling the well-known fact that an $r$-dimensional hypercube is a $1/r$-expander, we have that each $G[A_i]$ is a $\theta$-expander for $\theta=1/n.$
Taking $\eps=1/n$, it follows from \Cref{thm:local-mixing-intro} that for every $v \in f^{-1}(1)$ there exists a value $\ell = (ns)^{O(\log(s))}$, an event $E_v$, and a $j \in [s]$ such that $(i)$ a random walk $\bY_0, \dots, \bY_\ell$ satisfies $\Pr[E_v(\bY_0, \dots \bY_\ell)] \geq (ns)^{-O(\log s)}$ and $(ii)$ $\bY_\ell | E_v(\bY_0,\dots,\bY_\ell)$ has total variation distance $\frac{1}{n}$ from the uniform distribution over satisfying assignments of the term $T_j$.

Now, let $\mu$ denote the distribution of $\bY_\ell | E_v(\bY_0,\dots,\bY_\ell)$. We note that for any $i \in [n]$ such that neither $x_i$ nor $\overline{x}_i$ occurs in term $T_j$, we have that $\Pr_{\bx \sim \mu}[\bx_i = b] \geq 1/4$ for both values of $b \in \{0,1\}$. On the other hand, if $x_i$ or $\overline{x}_i$ is a literal in $T_j$ (say it is $x_i$), then  $\Pr_{\bx \sim \mu}[\bx_i = 1] \geq 1 - \frac{1}{n}$. It easily follows from this that the largest term satisfied by a set of $\Theta(\log n)$ many draws from $\mu$ will, with $1-o(1)$ probability, be precisely the term $T_j$, as in \cite{de2014learning}.

So, it suffices for us to get $\Theta(\log(n))$ draws that satisfy $E_v(\bY_0,\dots,\bY_\ell)$. We note that the probability that $\Theta(\log(n))$ independent random walks of length $\ell$ starting from $v$ all satisfy $E_v,$ is $(ns)^{-O(\log(s) \log(n))}$. Thus, we build our list $\calL$ as follows:  for each value $i=1,\dots,(ns)^{O(\log s)}$ (think of this as a guess for the correct value of $\ell$), for 
\[
(ns)^{O(\log(s) \log (n))}
\text{~many repetitions},
\]
we draw $\Theta(\log n)$ many $n$-bit strings which are the endpoints of independent length-$i$ random walks starting from $v$, and we take the largest term satisfied by all $\Theta(\log n)$ of those strings.
With overall probability $1-o(1)>0.99$, one of the terms in our list ${\cal L}$ will be precisely $T_j$.
\qed

\subsubsection{Learning Exact DNFs: Overview of the Proof of \Cref{thm:exact-k-learn-intro}}

We turn to giving an overview of the proof of \Cref{thm:exact-k-learn-intro}.
Throughout our discussion, we will assume that the length $k$ of each of the terms is known; this is without loss of generality since we can simply try each possible value $k=1,\dots,n.$

At a high level, our goal is to give a procedure which constructs a list $\calL$ of quasi-polynomial length that contains some term which (i) is satisfied by a non-trivial fraction (roughly $1/s$) of the weight that $\calD$ puts on positive assignments of $f$, and (ii) is almost never satisfied by draws from $\calD$ that are negative assignments of $f$ (see the notion of a ``weak term learner'' in \Cref{subsec:weak-term-learners}).   With such a procedure, as discussed in \Cref{subsec:weak-term-learners} it is not difficult to construct a high-accuracy hypothesis which is a DNF of size $O(s \log(1/\eps))$.

To construct such a list, our approach is roughly as follows.  Suppose that we have collected some set $\calL$ of at most $\quasipoly(n,s)$ many candidate terms (the initial such set is obtained from our list-decoding procedure).  We do a brute force sweep over all (at most quasipolynomially many) candidate terms that differ from a term in $\calL$ on at most $\polylog(ns)$ literals, to ``expand'' $\calL$ to include all nearby terms.  Then, using queries, we can ``prune'' terms from this expanded list by removing any term $T$ which does not ``behave like'' an implicant of $f$ on random assignments satisfying $T$ (i.e.~any term $T$ for which uniformly sampling points in $T^{-1}(1)$ reveals a point on which $f=0$).  It can be shown (see \Cref{subsec:expand-and-prune}) that this pruning procedure ensures that the set of terms that survive it can have at most quasi-polynomial size.  Taking $\calL$ to now be the set of terms that survived the pruning, we repeat.

Since the size of $\calL$ is bounded, this greedy approach of expanding $\calL$ must eventually get stuck. When this happens, one possibility is that our list $\calL$ already contains a term in $f$ which $\calD$ satisfies with sufficiently high probability; if this is the case, then we are done already. If this is not the case, then it must be the case that every term in $f$ that is not included in $\calL$ must have distance at least $\polylog(ns)$ from all of the terms in $\calL$.  Let $T_1$ be such a ``far'' term that has significant weight under $\calD$. It now follows that a random assignment $y$ satisfying such a term $T_1$ is $\polylog(ns)$ far from satisfying any term in $\calL$. We will call such a point in $\{0,1\}^n$ a \emph{far point}. (Note that for general DNFs, i.e. those in which terms do not have the same size, we can have a ``far term'' $T_1$ with no corresponding ``far point.'' In particular, if $T_1 = x_1 \land x_2 \land \dots \land x_{n/2}$ and $\calL = \{ \overline{x_1} \}$, then any satisfying assignment of $T_1$ is very close to satisfying $\overline{x_1}$, namely one just needs to flip the first bit. However, $T_1$ is very far from the term $\overline{x_1}$ as it includes many additional literals.)

It turns out (see \Cref{lem:far-point-learn}) that finding a far point will be sufficient for us to learn a new term from $f$ and thereby make progress. That said, finding a far point will take some work (see \Cref{sec:finding-far-points}). The idea of our approach to find far points is roughly as follows:  Let $y$ be a satisfying assignment of a ``far'' term $T_1$ that does not appear in the list $\calL$ (if such a point is hard to find, then it means that $\calL$ already contains terms that enable us to form a high accuracy hypothesis DNF, i.e.~we were in the first case mentioned above).
We begin by applying some random noise to the point $y$ (see \Cref{alg:warmup-estimator}). Oversimplifying a bit, this can be thought of as yielding a point $y'$ that with high probability will still satisfy $T_1$ but is far from satisfying any term $T \in \calL$ that has little overlap with $T_1$. Since all remaining terms have high overlap with $T_1$, it follows that there exists some literal $\ell$ that appears in a constant fraction of terms in $\calL$ that $y'$ is close to satisfying. We don't know, however, whether this literal $\ell$ is in $T_1$ or not. So, we make a guess as to whether or not it is in $T_1$, and branch on both outcomes of the guess. This is useful for the following reason:  if $\ell$ is indeed in $T_1$, then we have made progress towards learning $T_1$. Otherwise, we can flip the $i$th bit of $y'$, and this increases our distance from a constant fraction of terms in $\calL$, which constitutes significant progress towards finding a far point.

In line with the intuition that we are making progress in both cases, it can be shown that this guessing procedure has quasi-polynomial running time.  At the end of the process, we will have generated a far point, which enables us to find a new term in $f$. In summary, by repeatedly expanding $\calL$ and generating far points as described above, we can produce a list $\calL$ of quasi-polynomial length that contains some term in $f$ which $\calD$ satisfies with non-trivial probability. As mentioned earlier, with such a list of terms in hand it is not difficult to construct a high-accuracy hypothesis which is a DNF of size $O(s \log(1/\eps))$.

\part{Locally Mixing Random Walks}
\label{part:mix}


\section{Preliminaries on Expanders and Random Walks}
\label{sec:prelims}

\subsection{Graphs and Expansion} 
\label{subsec:prelims-random-walks}

Throughout this paper, the underlying graph $G = (V,E)$ that we consider is an unweighted undirected graph which may have multi-edges and self-loops (in fact, we will assume that $G$ does have many self-loops; see \Cref{ass:nice} below).  Given a set $S \subseteq V$, we define its volume as
\[
	\vol_{G}(S) := \sum_{v \in S} \deg(v),
\]
where each self-loop at a vertex contributes one to its degree.
We write $\dmax$ to denote $\max_{v \in V} \deg(v),$ the maximum degree of any vertex in $G$. 

We define the \emph{one-sided conductance of a set} $S \subseteq V$ to be
\[
	\Psi_{G}(S) := \frac{|E(S,\overline{S})|}{\vol_{G}(S)}.
\]
The \emph{conductance of the set} $S$ is taken to be
\[
	\Phi_{G}(S) := \max \cbra{\Psi_{G}(S), \Psi_{G}(\overline{S})}. 
\]
Finally, the \emph{conductance of the graph} $G$ is given by 
\[
	\Phi(G) := \min_{\emptyset \neq S \subsetneq [n]} \Phi_{G}(S).
\]
We say that $G$ is a \emph{$\theta$-expander} if $\Phi(G) \geq \theta$.

We sometimes omit subscripts and simply write $\Vol(S), \Psi(S), \Phi(S)$ when it is clear from context that the relevant graph is $G$.

The quantities defined above are intimately connected to the behavior of \emph{random walks} on $G$. Throughout this paper, by a ``random walk'' we mean the standard random walk on $G$, which moves along a uniformly chosen random edge out of the current vertex at each step.  We make the following useful assumption:

\begin{assumption} \label{ass:nice}
We assume that each vertex $v$ in the graph $G$ has precisely as many self-loops as edges to other vertices. 
We further assume that the graph $G$ has no isolated vertices (later, in \Cref{rem:no-isolated-vertices}, we will see that this is without loss of generality for our main result, \Cref{thm:local-mixing}).
We say that a graph $G$ with these properties is \emph{nice}.
\end{assumption}

Thanks to this assumption, a standard random walk on the nice graph $G$ corresponds to a lazy random walk (which stays put at each time step with probability $1/2$) in the graph $G'$ that is obtained from $G$ by removing all self-loops.  

We will frequently work with induced subgraphs; for $T \subseteq V$, we write $G[T]$ to denote the subgraph of $G$ that is induced by $T$. 
We let $\deg_T(v)$ denote the degree of $v$ in $G[T]$. Similarly, for $S \subseteq T$ we define 
\[
\Vol_T(S) := \sum_{v \in S} \deg_T(v), \quad
\Psi_T(S) :=  \frac{|E(S ,\overline{S})|}{\vol_{T}(S)}, \quad
\Phi_T(S) := \max \cbra{\Psi_{T}(S), \Psi_{T}(\overline{S})}. 
\]

\begin{remark}
We will frequently consider random walks on $G$ conditioned on staying in a set $S \subseteq V$. Such a conditioned random walk is easily seen to correspond to a standard random walk on $G[S]$.  (Note that this no longer corresponds exactly to the usual notion of a ``lazy'' random walk on $G[S]'$, the graph obtained from $G[S]$ by removing self-loops, because $G[S]$ may not be nice since the number of self-loops of a vertex $v$ in $G[S]$ may be larger than the number of edges from $v$ to other vertices in $G[S]$.)
\end{remark}

\paragraph{Big-Oh Notation:} Throughout the paper we use $\wt{O}(f)$ notation to suppress both $(\log f)^{O(1)}$ as well as $(\log|V(G)|)^{O(1)}$ factors.  We remark that while our results about random walks hold for arbitrary graphs $G=(V,E)$ as described above, in our DNF learning applications

\begin{itemize}

\item the graph $G$ will have vertex set $V = f^{-1}(1) \subseteq \zo^n$ where $f$ is an $s$-term DNF, and 

\item $G$ will be the subgraph of the $n$-dimensional Boolean hypercube that is induced by $V$ (augmented with self-loops as described in \Cref{ass:nice}), so there is an edge between two vertices $u,v \in V$ iff $f(u)=f(v)=1$ and the Hamming distance between $u$ and $v$ is 1.  

\end{itemize}
In this context $\log |V(G)|$ may be as large as $\Theta(n)$, so $\wt{O}(f)$ notation may suppress factors that are as large as $\poly(n)$.

\subsection{Total Variation Distance}

We write $(\Omega, \calF)$ to denote a probability space over a finite set $\Omega$, i.e.~${\cal F}$ is a collection of possible events (subsets of $\Omega$). We recall the familiar definition of total variation distance:

\begin{definition}[TV distance]
\label{def:tv-definition}
	Let $p$ and $q$ be probability distributions over $(\Omega, \mathcal{F})$. The \emph{total variation distance} between $p$ and $q$ is defined as
	\[
		\dtv(p, q) := \sup_{E \in \calF} \abs{p(E) - q(E)}.
	\]
For finite $\Omega$, we equivalently have
	\[
		\dtv(p, q) = \frac{1}{2} \sum_{\omega\in\Omega} \abs{p(\omega) - q(\omega)}. 
	\]
\end{definition}

We will frequently use the following simple result, which puts a bound on how TV distance can change under conditioning: 

\begin{lemma}
\label{lem:tv-condition}
	Suppose $p$ and $q$ are probability distributions over $(\Omega, \calF)$, and let $E\in\calF$. Then 
	\[
		\dtv(p|E, q|E) \leq \frac{2\cdot \dtv(p,q)}{p(E)}. 
	\]	
\end{lemma}

\begin{proof}
	Fix an event $F \in \mathcal{F}$. Note that by~\Cref{def:tv-definition}, we have
	\begin{equation} \label{eq:tv-def-cond-1}
		|p(F \land E) - q(F \land E) | \leq \dtv(p,q)
	\end{equation}
	and
	\begin{equation} \label{eq:tv-def-app}
		| p(E) - q(E)| \leq \dtv(p,q). 
	\end{equation}
	It follows from~\Cref{eq:tv-def-cond-1} that 
	\[
		\left| p(F|E) - \frac{q(F \land E)}{p(E)} \right| = \left| \frac{p(F \land E)}{p(E)} - \frac{q(F \land E)}{p(E)} \right| \leq \frac{\dtv(p,q)}{p(E)},  
	\]
	and also 
	\begin{align*}
		\left| \frac{q(F \land E)}{p(E)} - q(F|E) \right| 
		& = \left| \frac{q(F \land E)}{p(E)} - \frac{q(F \land E)}{q(E)} \right|\\
		& = \left| \frac{q(E) q(F \land E) - p(E) q(F \land E)}{p(E) q(E)} \right| \\
		& \leq \frac{\dtv(p,q) q(F \land E)}{p(E) q(E)} \\
		& \leq \frac{\dtv(p,q)}{p(E)}.
	\end{align*}
	It then follows from the triangle inequality that 
	\[
		\abs{p(F|E) - q(F|E)} \leq \frac{2\cdot \dtv(p,q)}{p(E)},
	\]
	completing the proof. 
\end{proof}

We will also need the following result:

\begin{lemma}
\label{lem:tv-refine}
Let $p,q$ be probability distributions over $(\Omega, \mathcal{F})$, and let $p', q'$ be corresponding probability distributions over a ``coarser'' probability space $(\Omega, \mathcal{F}')$ (i.e. $\calF' \subseteq \calF$).  
Then
	\[\dtv(p',q') \leq \dtv(p,q).\]
\end{lemma}

\begin{proof}
We have 
\[
	\dtv(p', q') = \sup_{E\in \calF'} |p'(E) - q'(E)| = \sup_{E\in\calF'} |p(E) - q(E)| \leq \sup_{E\in\calF} |p(E) - q(E)| = \dtv(p, q). \qedhere
\]
\end{proof}

\subsection{Random Walks on Graphs}

We briefly recall some basics of random walks on graphs; see \cite{LevinPeres17} for a detailed treatment.
Let $A$ denote the adjacency matrix of the graph $G = (V, E)$ with $|V| = N$. 
Let $\wt{A}$ denote the normalized adjacency matrix 
\[
	\wt{A} := D^{-1/2} A D^{-1/2}
\]	 
where $D$ is the $N \times N$ diagonal matrix $\mathrm{diag}(\deg(v)_{v\in V})$. 

Let $x\in(\R_{\geq 0})^V$ with $\|x\|_1 = 1$ be a distribution on the vertex set $V$. One step of a standard random walk, starting from $x$, results in the distribution $Wx,$ where $W := AD^{-1}$. 
Consequently, $t$ steps of the random walk results in the distribution
\[
	W^t x = D^{1/2}\pbra{\wt{A}}^t D^{-1/2} x. 
\]

It is easy to check that $W$ has an eigenvector $\pi \in \R^V$ with $\pi_i = \deg(i)\cdot(2|E|)^{-1}$, and more generally $W$ has eigenvalues that are at most $1$ in absolute value. 
Moreover, if $G$ is connected then there is a unique eigenvector with eigenvalue $1$ and it has all positive entries; this eigenvector is the stationary distribution of the random walk, which we denote by $\pi$.

Closely related to the random walk matrix is the \emph{Laplacian} matrix of $G$, defined as $L := D - A$, as well as the \emph{normalized} Laplacian, defined as $\wt{L} := D^{-1/2} L D^{-1/2}$. Note that $(I - \wt{L}) = \wt{A}$, where $I$ is the $N \times N$ identity matrix. 

It is well known that the convergence of random walks is closely related to the second largest eigenvalue of $\wt{A}$, which corresponds to the second smallest eigenvalue of $\wt{L}$, denoted $\lambda_2$. This in turn is related to the combinatorial expansion of the graph by Cheeger's inequality:

\begin{theorem}[Cheeger's Inequality]
$\frac{1}{2} \lambda_2 \leq \Phi(G) \leq \sqrt{2 \lambda_2}.$
\end{theorem}
We say that a graph with $\lambda_2 \geq \lambda$ is a \emph{$\lambda$-spectral expander}.

We recall the standard definition of the \emph{mixing time} of a random walk:
\begin{definition}[Mixing Time]
	The mixing time of a random walk on $G$ is the smallest number of steps $t_{\mix}$ such that for all starting distributions $x \in (\R_{\geq 0})^V$, we have $\|W^{t_{\mix}} x - \pi\|_1 \leq 1/4$.
\end{definition}
(Note that the constant $\frac{1}{4}$ is essentially arbitrary; indeed, to achieve a total variation distance bound of at most $\eps$, it suffices to walk for $t_{\mix} \log(1/\eps)$ steps.) 

The following fundamental result upper bounds the mixing time in terms of the second smallest eigenvalue:
\begin{theorem}[\cite{LevinPeres17}] \label{thm:mixing-time}
Let $\lambda_2$ denote the second smallest eigenvalue of the normalized Laplacian of $G$. The mixing time of a random walk on $G$ is then at most $O \left( \frac{\log N }{\lambda_2} \right)$. 
\end{theorem}

We record the following useful corollary of \Cref{thm:mixing-time}:
\begin{corollary} \label{cor:thetasquared}
If $G$ is a $\theta$-expander then a random walk on $G$ mixes in $\wt{O}(\theta^{-2})$ steps.
\end{corollary}

\subsection{Higher Order Cheeger Inequalities}
Unfortunately, the quadratic loss from Cheeger's inequality will not be sufficient for our purposes. Instead, we will use the following strengthening of Cheeger's inequality:

\begin{theorem}[Kwok et al. \cite{kwok2013improved}]
\label{thm:improved-cheeger}
Let $G = (V,E)$ be a graph and $\lambda_i$ denote the $i$th smallest eigenvalue of its normalized Laplacian. Then
	\[\Phi(G) \leq O(k) \frac{\lambda_2}{\sqrt{\lambda_k}}\]
\end{theorem}

In order to use \Cref{thm:improved-cheeger}, we require a lower bound on $\lambda_k$.  A combinatorial lower bound on $\lambda_k$ is given by the following result of Louis et al.:

\begin{theorem}[Louis et al. \cite{louis2012many}]
\label{thm:many-cuts}
	Let $G = (V,E)$ be a graph and $\lambda_i$ denote the $i$th smallest eigenvalue of its normalized Laplacian. For every $k \leq n$, there exist $ck$ disjoint sets  sets $S_1, ..., S_{ck}$  such that
	\[\max_i \Phi(S_i) \leq C \sqrt{\lambda_k \log k},\]
where $c<1,C>0$ are suitable absolute constants.
\end{theorem}

\subsection{Escape Probabilities}
How hard is it for a random walk in $G=(V,E)$ to escape from a given set $S \subseteq V$? In this subsection we seek to understand this quantity via the one-sided conductance of $S$. Towards this end, for a set $S \subseteq V$ we define the vector $\pi_S$ to be the stationary distribution restricted to the set $S$, i.e. $\pi_i = 0$ if $i \not \in S$ and $\pi_i = \deg(i)/\vol(S)$ otherwise. 

\begin{lemma}[Spielman and Teng \cite{spielman2013local}]
\label{lem:escape-lb}
For any $S \subseteq V$, a random walk starting from a vertex $\bX \sim \pi_S$ remains in $S$ for $t$ steps with probability at least $1 - t \Psi(S)$.
\end{lemma}

We remark that stronger bounds are known for larger values of $t$ \cite{gharan2012approximating}, but we will not need them here. (Spielman and Teng technically prove a variant of \Cref{lem:escape-lb} with $\Phi(S)$ in place of $\Psi(S)$, but the same proof gives \Cref{lem:escape-lb} by simply modifying the last line of the proof.)

A converse to this statement follows from the fact that random walks in expanders mix quickly:

\begin{lemma}
\label{lem:escape-ub}
	Let $G = (V,E)$ and let $S \subseteq V$. Moreover, assume that for every $T \subseteq S$, we have that $\Psi(T) \geq \theta$. Then with probability at least $1/8$, a random walk starting from any distribution over vertices in $S$ leaves $S$ within $\wt{\Omega}(1/\theta^2)$ steps.
\end{lemma}

\begin{proof}
	We build a new graph $G'$ by taking $G[S]$ and adding a new vertex $s$, so the vertex set of $G'$ is $S \cup \{s\}$. For each vertex $v \in S$, we add $|\{e \ni v: e \in E(S,\overline{S})\}|$ edges from $v$ to $s$. We then add $\vol_{G}(S)$ self-loops from $s$ to itself. (Note that the graph $G'$ is not ``nice'' because node $s$ has more self-loops than outgoing edges to other vertices, but all other vertices have the same number of self-loops as outgoing edges.)
	
	We now claim that $G'$ has $\Phi(G') \geq \theta$. Indeed, fix any cut $(T, \overline{T})$ of $V(G')$. Without loss of generality, assume that $s \not \in T$. We then have that
		\[\Phi_{G'}(T) \geq \frac{|E_{G'}(T, \overline{T})|}{\vol_{G'}(T)} = \frac{|E_{G}(T, \overline{T})|}{\vol_{G}(T)} = \Psi_G(T)\geq \theta\] 
	by assumption. 
	
	Now note that the stationary distribution of the random walk over $G'$ satisfies $\pi(s) = \frac{\deg_{G'}(s)}{\vol_{G'}(G')} \geq \frac{1}{4}$.
	It now follows from \Cref{cor:thetasquared} that if we take a random walk over $G'$ starting from any distribution over the vertices in $S$ for $\wt{\Omega}(1/\theta^2)$ steps, so that the total variation distance to the stationary distribution of the random walk over $G'$ is at most $1/8$, then the final distribution places mass at least $\frac{1}{8}$ on $s$. Hence, with probability at least $\frac{1}{8}$, such a walk leaves $S$.
\end{proof}


\section{Graph Covers}
Recall that we are interested in graphs with a local structure. Towards this end, we make the following definition:
\begin{definition}[Cover of a Graph]
	Given a graph $G = (V,E)$, we say $\mathcal{C} := \{A_1, \dots A_s\}$ forms a \emph{cover} of $G$ if $V = \bigcup_i A_i$. We say that $\mathcal{C}$ forms a \emph{$\theta$-cover} if $\Phi(G[A_i]) \geq \theta$ for all $i$. Finally, we say $\mathcal{C}$ is a \emph{disjoint} $\theta$-cover if $A_1, ..., A_s$ are disjoint and form a $\theta$-cover.
\end{definition}

Given a cover and a set $I \subseteq [s]$, it will often be convenient to write $A_I := \bigcup_{i \in I} A_i$

If a graph $G = (V,E)$ has a $\theta$-cover $A_1, \dots, A_s$, then we should expect it to be very well behaved. We start by noting that unions of the sets $A_i$ approximate the worst cut in covered graphs.

\begin{lemma}
	\label{lem:cut-approx}
	Let $G = (V,E)$ be a graph and $\mathcal{C} := \{A_1, ..., A_s\}$ a $\theta$-cover. There exists a set $I \subseteq [s]$ such that
	
	\[\Phi \left( \bigcup_{i \in I} A_i \right) \leq 4 sd^2_{\max} \Phi(G)/\theta \]
	(recall that trivially we have $\Phi(G) \leq \Phi \left( \bigcup_{i \in I} A_i \right)$).
\end{lemma}

\begin{proof}
	Throughout the proof we assume that $\Phi(G) \leq \theta/(4s \dmax ^2)$, as otherwise the statement is trivially true. Let $S$ be the sparsest cut in $G$, so $\Phi(S) = \Phi(G)$, and supposed without loss of generality that $\vol(S) \leq \frac{1}{2} \vol(G)$. 
	
First, note that for each $i \in [s]$, we either have that $(i)$ $|S \cap A_i| \leq \frac{\dmax  \Phi(G) |S|}{\theta}$ or $(ii)$ $|A_i \setminus S| \leq \frac{\dmax  \Phi(G) |S|}{\theta}$. To see this, suppose that $\vol_{A_i}(S \cap A_i) \leq \frac{1}{2} \vol_{A_i}(A_i)$, then note that
	\[\Phi(G) = \Phi(S) = \frac{|E(S,\overline{S})|}{\vol(S)} \geq \frac{|E_{A_i}(S,\overline{S})|}{\vol(S)} \geq \frac{\theta \cdot |S \cap A_i|}{\dmax  |S|}. \] 
	Rearranging then gives the desired result. The case of $\Vol_{A_i}(S \cap A_i) > \frac{1}{2} \Vol_{A_i}(A_i)$ follows by considering $\overline{S}$ in the previous analysis.
	
	We now take $I = \left\{i \in [s]: |A_i \setminus S| \leq \frac{\dmax  \Phi(G) |S|}{\theta} \right\}$. Set $T = \bigcup_{i \in I} A_i$ and compute as desired that (see below for further elaboration on \Cref{eq:peach} and \Cref{eq:plum}):
	\begin{align}
		\Phi \left( T \right) =  \frac{|E(T,\overline{T})|}{\min(\vol(T), \vol(\overline{T}))} &\leq \frac{|E(S,\overline{S})| + |S \Delta T| \dmax }{\vol(S) - |S \Delta T| \dmax } \label{eq:peach}\\
		&\leq \frac{|E(S,\overline{S})| +  (s \dmax ^2 \Phi(G)/\theta) \cdot \vol(S)}{(1 - s \dmax ^2 \Phi(G) /\theta) \cdot \vol(S)} \label{eq:plum} \\
		&\leq 2 \cdot \Phi(S) + 2sd^2_{\mathrm{max}} \Phi(G)/\theta \nonumber \\
		&\leq 4 sd^2_{\mathrm{max}} \Phi(G)/\theta. \nonumber
	\end{align}
\Cref{eq:peach} above uses $\vol(T),\vol(\bar{T}) \geq \vol(S) - |S \Delta T| \dmax $. The first of these inequalities is obvious; for the second, recall first that $\vol(T)+\vol(\bar{T})=\vol(G)$.  If $\vol(T) \leq {\frac 1 2} \vol(G)$ then $\vol(\bar{T}) \geq \vol(T)$ and we are done; so suppose ${\frac 1 2} \vol(G) \leq \vol(T)$, and observe that $\vol(T) \leq \vol(S) + |S \Delta T| \dmax $. Then 
\begin{align*}
		\vol(\bar{T}) &= \vol(G) - \vol(T)\\
		&\geq \vol(G) - \vol(S) - |S \Delta T| \dmax \\
		&\geq {\frac 1 2} \vol(G) - |S \Delta T| \dmax\\  
		&\geq \vol(S) - |S \Delta T|\dmax ,
\end{align*}
		where the last two inequalities used $\vol(S) \leq {\frac 1 2} \vol(G).$
		
		\Cref{eq:plum} above uses $|S \Delta T| \leq (s \dmax  \Phi(G)/\theta)\vol(S).$
		To see this, note first that the definition of $I$ implies $|T \setminus S| \geq (s \dmax \Phi(G)/\theta)\cdot|S|$. To bound $|S \setminus T|$, observe that $|S \cap A_i| \leq \dmax  \Phi(G) |S| /\theta$ for any $i \notin I$ (since either $(i)$ or $(ii)$ above holds); only the sets $A_i$ with $i \notin I$ contribute to $S \setminus T$ (since the $A_i$'s form a cover), and hence the total contribution of all of these to $|S \setminus T|$ is at most $(s-|I|)\dmax \Phi(G)|S|/\theta$. Since $\Vol(S) \geq 2|S|$ (this is an easy consequence of the fact that $G$ has no isolated vertices),
		we get that $|S \Delta T| \leq (s \dmax  \Phi(G)/\theta)\vol(S)$ as claimed.
\end{proof}

The next lemma gives a lower bound on higher-order eigenvalues of any graph with a $\theta$-cover; this will be useful in conjunction with the higher order Cheeger inequality mentioned in the preliminaries.

\begin{lemma} \label{lem:squaring}
	Let $G = (V,E)$ be a graph and $\mathcal{C} := \{A_1, ..., A_s\}$ be a $\theta$-cover. Moreover, let $\lambda_k$ be the the $k$th	smallest eigenvalue of the normalized Laplacian of $G$. For $k = \Omega(s)$, we then have that
	\[\lambda_k \geq \Omega \left ( \frac{\theta^2}{s^2 \dmax ^2  \log k} \right). \]
\end{lemma}

\begin{proof}
	Let $S_1,\dots,S_{3s}$ be any $3s$ disjoint subsets of $V$.
	There must be some $i^\ast \in [3s]$ such that for every $j \in [s],$ we have $\Vol_{A_j}(S_{i^\ast} \cap A_j) \leq {\frac 1 2} \Vol_{A_j}(A_j)$.
(For each $j \in [s]$, there are at most two values of $i \in [3s]$ such that 	
$\Vol_{A_{j}}(S_{i} \cap A_{j}) \geq  {\frac 1 2} \Vol_{A_{j}}(A_{j})$, so there must be some value $i^\ast \in [3s]$ for which no $j$ satisfies $\Vol_{A_{j}}(S_{i} \cap A_{j}) \geq  {\frac 1 2} \Vol_{A_{j}}(A_{j})$.)
Since $A_1,\dots,A_s$ is a cover, some $j \in [s]$ has $|A_j \cap S_{i^\ast}| \geq |S_{i^\ast}|/s.$
	We then note that
	\[|E_{A_j}(S_{i^\ast}, \overline{S_{i^\ast}})| \geq \theta \vol_{A_{i^\ast}}(S_{i^\ast} \cap A_j) \geq \theta \cdot \frac{|S_{i^\ast}|}{s} \geq \theta \frac{\vol(S_{i^\ast})}{s \dmax },\]
	which implies that
	\[\Phi(S_{i^\ast}) \geq \frac{\theta}{s \dmax }\]

Finally, taking $S_1,\dots,S_{3s}$ to be the $3s$ disjoint sets from \Cref{thm:many-cuts},
we get that 
\[
{\frac \theta {s \dmax}} \leq \max_i \Phi(S_i) \leq
C\sqrt{\lambda_k \log k},
\] which rearranges to give the lemma.
\end{proof}

Combining this \Cref{lem:squaring} with \Cref{thm:mixing-time} and \Cref{thm:improved-cheeger}, we then get that

\begin{corollary}
	\label{cor:fast-covered-mixing}
	Let $G = (V,E)$ be a graph and $\mathcal{C} = \{A_1, ..., A_s\}$ a $\theta$-cover. Then the mixing time of a random walk in $G$ is $\wt{O} \left( \frac{s^2 \dmax }{\theta \Phi(G)} \right)$.
\end{corollary}

\section{Local Mixing}
\label{sec:local-mixing-intro}

We now come to a definition that is central to our work:
\begin{definition}[$(p, \eps)$-Local Mixing Time]
	Given a vertex $v \in V$ and values $p,\eps \in [0,1]$, we say that the \emph{$(p, \eps)$-local mixing time of a random walk with respect to the cover $\mathcal{C} = \{A_1, ..., A_s\}$ starting at $v$}, denoted $T_{\mix}^{p,\eps}(v, \mathcal{C})$, is the smallest number $t$ such that there exists an $i \in [s]$ and an event $E$ such that a random walk for $t$ steps $v:= \bX_0, ..., \bX_t$ and independent, uniformly random $\br \sim [0,1]$ together satisfy:

\begin{itemize}
\item [(i)] $\quad \Pr[E(\bX_0, ..., \bX_t, \br)] \geq p$; and
		
\item [(ii)] $\dtv(\bX_t|E, \mathcal{U}_{A_i}) \leq \eps$, where $\mathcal{U}_{A_i}$ denotes the uniform distribution on $A_i$. 
\end{itemize}
	\flushleft For a graph $G$, the \emph{$(p,\eps)$-local mixing time} is $T_{\mix}^{p,\eps}(G, \mathcal{C}) := \max_{v \in V} T_{\mix}^{p, \eps}(v, \mathcal{C})$.\footnote{Throughout the paper, for clarity of notation we use capital letters (i.e., $T_{\mix}$) for local mixing times, and use lower case letters (i.e., $t_{\mix}$) for regular mixing times.} 
We may occasionally drop the $p, \eps$ and $\mathcal{C}$ when they are clear from context.
\end{definition}

Before proceeding, we make two remarks about this definition. First, we aim for $\bX_t$ to be distributed close to uniform on $A_i$, even if the nodes in $A_i$ do not all have the same degrees. (Stationary distributions of graphs would give higher weight to higher-degree nodes.) We will make use of the ancillary random variable $\br$ to achieve this when proving that certain graphs have small local mixing times: the event $E$ will (among other things) use $\br$ to reject higher-degree nodes more often, in order to make $\bX_t|E$ close to uniform. Second, unlike with the normal notion of mixing time, it is no longer clear that the choice of $\eps$ is arbitrary (up to constant factors). That said, our techniques give a bound with a reasonable dependence on $\eps$. 

The main result which we will prove about local mixing is the following:
	\begin{theorem} [Quasipolynomial bound on local mixing time]
	\label{thm:local-mixing}
		Let $G = (V,E)$ be a graph and $\mathcal{C} := \{A_1, ..., A_s\}$ be a $\theta$-cover. For $0 < \eps < 1$, we have that
			\[T_{\mix}^{p,\eps} (G, \mathcal{C}) \leq \left( \frac{s \dmax  \log(|V|) \log(1/\eps)}{\theta} \right)^{O(\log(s))}\]
		for some $p = \left( \frac{s \dmax  \log(|V|) \log(1/\eps)}{\theta} \right)^{-\Omega(\log(s))}$.
	\end{theorem}

 \begin{remark} \label{rem:no-isolated-vertices}
Recall the assumption (cf.~\Cref{ass:nice}) that $G$ has no isolated vertices.  This is without loss of generality vis-a-vis \Cref{thm:local-mixing} for the following reason: If $G$ has an isolated vertex $v$, then since each $A_i$ is a $\theta$-expander, $v$ must be its own part in the $\theta$-cover, i.e., there is some $A_i$ which consists solely of $v$. But then, starting from $v$, a random walk trivially locally mixes, and there is no other way of reaching it starting from a different vertex, so we can remove $v$ and consider the problem on the remainder of the graph.
 \end{remark}

\subsection{Local Mixing in Super Covers}
The goal of the next two sections will be to prove \Cref{thm:local-mixing}. Our proof strategy will crucially rely on the fact, which we prove in \Cref{lem:super-cover-mixing}, that refining a cover can only make local mixing easier (up to a small quantitative loss in parameters).

\begin{definition}[Power Set of a Cover]
	Let $G = (V,E)$ be a graph and $\mathcal{C} = \{A_1, ..., A_s\}$ be a cover of $G$. We define the \emph{power set of the cover} to be 
	 \[ \mathcal{P}(\mathcal{C}) := \left \{\bigcup_{j \in J} A_j: J \subseteq [s] \right \}.\]
\end{definition}

\begin{definition}[Super Cover]
	Let $G = (V,E)$ be a graph and $\mathcal{C} = \{A_1, ..., A_s\}$ be a cover of $G$. We say that a cover $\mathcal{C}' = \{B_1, .., B_r\}$ is a (disjoint $\alpha$) \emph{super cover of $G$} if $\mathcal{C}' \subseteq \mathcal{P}(\mathcal{C})$ and ${\cal C}'$ is a (disjoint $\alpha$) cover of $G$. 
\end{definition}

Note that if $\calC'$ is a super cover of $\calC$, then $\calC$ can be viewed as a refinement of the super cover $\calC'$.

\begin{lemma}
\label{lem:super-cover-mixing}
Let $G = (V,E)$ be a graph and $\mathcal{C} := \{A_1, ..., A_s\}$ be a cover of $G$. If $\mathcal{C}' := \{B_1, ..., B_r\}$ is a super cover of $\mathcal{C}$, then for all $v \in V$ and every $p, \eps \in [0,1]$, we have
\[T_{\mix}^{p, \eps}(v, \mathcal{C}') \geq T_{\mix}^{p/(2s), 2\eps s}(v, \mathcal{C}). \] 	
\end{lemma}

\begin{proof}
	Assume that $\eps \leq 1/2s$, as otherwise the statement is trivial. Let $F_v$ denote the event from the definition of locally mixing with respect to $\mathcal{C}'$. Moreover, let $B_j \in \mathcal{C}'$ denote the set such that $\dtv(\bX_t|F_v, \mathcal{U}_{B_j}) \leq \eps$, where $v := \bX_0, ..., \bX_t$ is a random walk in $G$ starting from $v$. By assumption $B_j = \bigcup_{i \in I} A_i$ for some set $I \subseteq [s]$. Take $i \in I$ to be the index that maximizes $|A_i|$, so $|A_i|/|B_j| \geq 1/s$. Let $E_v$ be the event that $F_v$ occurs and $\bX_t \in A_i$. It now follows from \Cref{lem:tv-condition} that
	\[\dtv(\bX_t|E_v, \mathcal{U}_{A_i}) \leq \frac{2 \dtv(\bX_t|F_v, \mathcal{U}_{B_j})}{|A_i|/|B_j|} \leq 2s \cdot \dtv(\bX_t|F_v, \mathcal{U}_{B_j}) \leq 2 \eps s.\]
	Moreover, since $\Pr_{\bX \sim \mathcal{U}_{B_j}}[\bX \in A_i] - \Pr[\bX_t \in A_i | F_v] \leq \eps \leq \frac{1}{2s}$ we have that 
		\[
		\Pr[E_v] =  \Pr[F_v] \cdot \Pr[\bX_t \in A_i | F_v] \geq \frac{p}{2s},
		\]
		establishing the lemma.
\end{proof}

\section{Reduction to Disjoint Expanders} \label{sec:disjointification}

Intuitively, it should only be easier to achieve local mixing in the case when the expanders which make up the cover overlap with each other. In this section, we formalize this intuition and show that it suffices to prove \Cref{thm:local-mixing} in the case of a disjoint $\theta$-cover. 

More precisely, we consider the following version of \Cref{thm:local-mixing} for disjoint covers:

\begin{theorem}
	\label{thm:disjoint-local-mixing}
	Let $G = (V,E)$ be a graph and $\mathcal{C} := \{A_1, ..., A_s\}$ be a disjoint $\theta$-cover of $G$. We then have that
	\[T_{\mix}^{p,\eps} (G, \mathcal{C}) \leq \left( \frac{s \dmax  \log(|V|) \log(1/\eps)}{\theta} \right)^{O(\log(s))}\]
	where $p = \left( \frac{s \dmax  \log(|V|) \log(1/\eps)}{\theta} \right)^{-\Omega(\log(s))}$
\end{theorem}
\noindent Our goal in this section is to prove that \Cref{thm:disjoint-local-mixing} implies the more general \Cref{thm:local-mixing}. Afterwards, in \Cref{sec:localmixingboundproof} below, we will prove \Cref{thm:disjoint-local-mixing}, and thus conclude the proof of \Cref{thm:local-mixing}.

The key tool in this section is the following transformation.

\begin{definition}[Disjointification]
	Given a graph $G$ and a $\theta$-cover $\mathcal{C} = \{A_1, ..., A_s\}$, we say that the \emph{disjointification} is a graph $H$ defined as follows: For each vertex $v \in V(G)$, we make $s$ copies of $v$ in $H$. Given two distinct vertices $u,v$ in $G$, there is an edge between copies of $u$ and $v$ in $H$ if there is an edge between $u$ and $v$ in $G$.\footnote{Note that we do not have edges in $H$ corresponding to the self-loops from $G$.} We also add $s \cdot \deg_G(v)$ self-loops from each copy of $v$ in $H$ to itself.
	
	The disjointification also has a corresponding disjoint cover of $H$, consisting of $s$ sets, which we denote by $\mathcal{C}' := \{B_1,..., B_s\}$. We define it as follows: for each $v \in V(G)$, evenly partition the copies of $v$ among all the sets $B_i$ for which $v \in A_i$. (If the number of copies of $v$ is not divisible by the number of $i$ for which $v \in A_i$, then partition them as evenly as possible so that the number of copies put into any two such sets $B_i, B_{i'}$ differs by at most 1.) 
\end{definition}

We first note that each part in the disjoint cover of the disjointification is still a good expander:

\begin{lemma}
 	Let $G$ be a graph, $\mathcal{C} = \{A_1, ..., A_s\}$ be a $\theta$-cover of $G$, and $(H,\mathcal{C}')$ be the disjointification of $(G, \mathcal{C})$. Then $\mathcal{C}'$ is a disjoint $\theta/s^2$-cover of $H$.
\end{lemma}

\begin{proof}
Let $\mathcal{C}' = \{B_1, ..., B_s\}$. Consider a cut $(S, \overline{S})$ in $H[B_i]$. Let $T$ denote the vertices in $G$ corresponding to the vertices in $S$. Without loss of generality, we assume that $\vol_{G[A_i]}(T) \leq \frac{1}{2} \vol_{G[A_i]}(A_i)$. We then have that
	\[|E_{H[B_i]}(S, \overline{S})| \geq |E_{G[A_i]}(T, \overline{T})| \geq \theta \vol_{G[A_i]}(T).\]
We now note
	\[\vol_{H[B_i]}(S) \leq s^2 \vol_{G[A_i]}(T), \]
as $H[B_i]$ contains at most $s$ copies of each vertex in $T$ and each satisfies $\deg_{H[B_i]}(v) \leq s \deg_{G[A_i]}(v)$. Combining these two facts then gives us as desired that $\Phi_{H[B_i]}(S) \geq \theta/s^2$.
\end{proof}

We now use the disjointification operation to prove that we may assume our cover is disjoint without loss of generality.

\begin{proof}[Proof of \Cref{thm:local-mixing} assuming \Cref{thm:disjoint-local-mixing}]
Let $(H, \mathcal{C}')$ be the disjointification of $(G,\mathcal{C})$, and fix a vertex $v \in V(G)$. We will argue that $T_{\mix}^{p,\eps}(v, \mathcal{C})$ is small.

Fix a copy of $v$ in $H$, which we denote by $u \in V(H)$. Let 
$T = T_{\mix}^{p, \eps/2s}(u, \mathcal{C}')$ and let $F_u$ be the event in the definition of this local mixing, i.e., for a random walk $u := \bY_0, ..., \bY_t$
	\[\dtv(\bY_t|F_u, \mathcal{U}_{B_i}) \leq \eps/2s\]
for some $B_i \in \mathcal{C}'$. 

We will also take a random walk of length $T$ starting from $v$, which we write as $v := \bX_0, ..., \bX_T$. We note that we can naturally couple this walk and a random walk $u := \bY_0, ..., \bY_T$ in $H$. Namely, for each $j \in \{1, ..., T\}$, we set $\bY_j$ to be a uniformly and independently random copy of the vertex $\bX_j$ in $H$, except that if the random walk in $G$ traverses a self-loop, then we have the random walk in $H$ stay in place at the same copy as the previous step. 

We now define the appropriate event for local mixing of the random walk in $G$. For vertex $w \in V(G)$, let $B(w)$ denote the event that a Bernoulli random variable, which is $1$ with probability $1/(\text{\# copies of $w$ in $B_i$})$, takes value 1.
 Our event is then $E_v:=\bF_u(\bY_0, ..., \bY_T, \br) \land B(\bX_T)$. 
 
Notice that the walk in $H$ can be viewed as a refinement of the walk in $G$. Thus by \Cref{lem:tv-refine},
 	\[\dtv(\bX_T|F_u, \mathcal{U}^\star_{B_i}) \leq {\frac \eps {2s}} \]
where $\mathcal{U}^\star_{B_i}$ denotes the corresponding event in $G$. Now observe that $\mathcal{U}^\star_{B_i}|B = \mathcal{U}_{A_i}$, so using \Cref{lem:tv-condition} yields
	\[\dtv(\bX_T|E_v, \mathcal{U}_{A_i}) \leq \eps. \]
Finally, we note that $\Pr[E_v(\bX_1, ..., \bX_T,\br)] \geq \Pr[F_v(\bY_1, ..., \bY_T, \br)] \cdot \frac{1}{s} = p/s$. So we conclude that
\begin{equation}
\label{eq:see-you-later}
	T_{\mix}^{p/s, \eps} (v, \mathcal{C}) \leq T_{\mix}^{p, \eps/(2s)}(u, \mathcal{C}').
\end{equation}
Using \Cref{thm:disjoint-local-mixing} to bound $T_{\mix}^{p, \eps/(2s)}(u, \mathcal{C}')$ then gives the desired result.
\end{proof}

\section{Proof of Local Mixing Bound (\Cref{thm:disjoint-local-mixing})} \label{sec:localmixingboundproof}
The goal of this section is to prove \Cref{thm:disjoint-local-mixing}. At a high level, our proof will proceed as follows.

Let $G$ be a graph and $\mathcal{C} = \{A_1, ..., A_s\}$ a disjoint $\theta$-cover. We will first observe that if for every $i \in [s]$ there exists a subset $I \subseteq [s]$ that contains $i$ and has $\Phi(G[A_I]) \gg \Psi(A_I)$, then we can bound the local mixing time of $G$. As such, the remainder of the proof attempts to prove that such sets $I$ in fact exist.

To construct the sets $I$, we first show that such a set $I$ exists for the biggest set, which we assume is $A_1$ without loss of generality. Afterwards, we will roughly apply this argument recursively, e.g., we will first apply it to $G[\bigcup_{i > 1} A_i]$ in order to show the same thing for $A_2$, 
and so on.

Before embarking on the proof, we will set several parameters. First, we fix an $\eps \in [0,1/2)$. Next we set 
\begin{equation} \label{eq:lambda-def}
\lambda = (1000 s\dmax \theta^{-1} \log(|G|)\log(1/\eps))^{-1000}.
\end{equation}
 We assume that $\theta$-covers are sorted by size, i.e.~$|A_1| \geq |A_2| \geq \dots \geq |A_s|$ unless otherwise specified. Finally, we recall (cf.~\Cref{rem:no-isolated-vertices}) that we may assume that $G$ has no isolated vertices.

\subsection{Good Sets and Local Mixing}  \label{sec:goodsetslocalmixing}

To start, we show that a family of sets with large expansion and low conductance locally mixes quickly. To capture precisely what we aim to prove about such sets, we make the following definition.

\begin{definition}[$(\Delta, g)$-Good Set]
	Given a graph $G = (V,E)$ and a cover $\mathcal{C} := \{A_1, ..., A_s\}$, we say that a set $S \subseteq V$ is \emph{$(\Delta, g)$-good with respect to $\mathcal{C}$} if 
		\[\left| \{ v \in S: T_{\mix}^{1/8\dmax ,\eps}(v, \mathcal{\mathcal{P}(\mathcal{C}})) \leq \Delta \} \right| \geq g |S|.\]
\end{definition}

To begin, we show that if a set's expansion is much larger than its conductance, then most points in the set locally mix quickly.

\begin{lemma}
\label{lem:low-cond-good}
Suppose that $G = (V,E)$ is a graph with $\theta$-cover $\mathcal{C} := \{A_1, ..., A_s\}$. If $I \subseteq [s]$ is such that 
	\[\Psi(A_I) \leq \wt{O} \left(\frac{\theta \Phi(G[A_I])}{s^2 d^2_{\mathrm{max}} \log(1/\eps)} \right) \]
then $A_I$ is $(t, \frac{1}{2})$-good (with respect to $\mathcal{C}$) for $t = \wt{\Theta}(\frac{s^2 \dmax  \log(1/\eps)}{\theta \Phi(G[A_I])})$.
\end{lemma}

\begin{proof}
For any $v \in A_I$, we consider a random walk $\bX_0, ..., \bX_t$ starting at $v$. The key behind this proof will be showing that there are many $v \in A_I$ for which, with decent probability, this random walk will never leave $A_I$. Let us first show why this suffices to imply local mixing, and then we will return to proving it.

With this in mind, we define the event needed in the definition of local mixing. Let $\br \sim [0,1]$ be drawn uniformly at random. In proving local mixing, we will condition on the event $E_v$ that we never leave the set $A_I$ during the random walk $\bX_0, ..., \bX_t$ and that $\br \leq \frac{1}{\deg_{A_I}(\bX_t)}$.

We now consider the random walk conditioned on never leaving $A_I$ (but not yet on the event $E_v$). Note that this is equivalent to taking a random walk in $G[A_I]$. By \Cref{cor:fast-covered-mixing}, such a random walk has total variation distance at most $\eps/4\dmax $ to the stationary distribution, $\mu^{A_I}$, after $\wt{O}(\frac{s^2 \dmax  \log(4\dmax /\eps)}{\theta \Phi(G[A_I])}) \leq t$ steps.

Observe that the stationary distribution $\mu^{A_I}$ is given by $\mu^{A_I}(v) = \frac{\deg_{A_I}(v)}{\vol_{A_I}(A_I)}$ for $v \in A_I$, and $\mu^{A_I}(v)= 0$ for $v \notin A_I$. Furthermore, by the definition of the event $E_v$, we see that $\mu_{A_I}|E_v$ is the same as $\calU_{A_I}$. (Indeed, the condition on $\br$ in the definition of $E_v$ is designed to cancel out the numerator in $\mu^{A_I}(v)$.)

It now follows by \Cref{lem:tv-condition} that
	\[\dtv(\bX_t|E_v, \mathcal{U}_{A_I}) \leq \frac{2\eps}{4 \dmax } \cdot \left( \Pr_{\bX \sim \mu^{A_I}, \br} \left[\br \leq \frac{1}{\deg_{A_I}(\bX_t)} \right] \right)^{-1} \leq \eps.\]
	
It now suffices to show that 
$\Pr[$the random walk never leaves $A_I$ in $\bX_0,\dots,\bX_t]\geq {\frac 1 4}$ 
for many choices of $v \in A_I$. By \Cref{lem:escape-lb}, we have that a random walk starting from a random vertex $\bX \sim \pi_{A_I}$ leaves $A_I$ after $t$ steps with probability at most
		\[t \cdot \Psi(A_I) \leq \frac{1}{32 \dmax }.\]
On the other hand, any set of $|A_I|/2$ vertices in $A_I$ has mass under $\pi_{A_I}$ at least
		\[\frac{|A_I|/2}{\vol_{A_I}(A_I)} \geq \frac{1}{4 \dmax }.\]
Thus, the number of vertices that escape with probability at least $\frac{1}{4}$ is at most $|A_I|/2$, completing the proof.
\end{proof}

Next, to demonstrate how we will use good sets, we prove in \Cref{lem:good-mixing} that if all the sets in a super cover are good with respect to an original cover, then the local mixing time of the graph with respect to the original cover is small.  Given \Cref{lem:good-mixing}, the rest of our effort towards proving \Cref{thm:disjoint-local-mixing} will be to establish the existence of a super cover with the properties described in the lemma statement.

\begin{lemma}
\label{lem:good-mixing}
Suppose that $G = (V,E)$ is a graph and $\mathcal{C} := \{A_1, ..., A_s\}$ is a $\theta$-cover. Moreover, let $\{B_1, ..., B_s\}$  be an $\alpha$-super cover, such that $B_1, ..., B_s$ are all $(\Delta, g)$-good with respect to $\mathcal{C}$. (Note that the $\{B_i\}$ do not necessarily need to be disjoint.) 
We then have that 
		\[ T_{\mix}^{p, 4 \eps s}(G, \mathcal{C}) \leq \wt{O} \left(\Delta + \frac{s^2 d^4_{\mathrm{max}}}{g^2 \alpha^2} \right) \]
	for some $p = \Omega \left( \frac{\alpha^2 g^2}{s^4 \dmax ^5 \Delta} \right)$.
\end{lemma}

\begin{proof}
We will actually prove that, if we furthermore assume $g \leq 1/2\dmax $, then we get the stronger statement that
		\[ T_{\mix}^{p, 4 \eps s}(G, \mathcal{C}) \leq \wt{O} \left(\Delta + \frac{s^2 d^2_{\max}}{g^2 \alpha^2} \right) \]
for some $p = \Omega \left( \frac{\alpha^2 g^2}{s^4 \dmax ^3 \Delta} \right)$. The desired result will then follow if this additional assumption does not hold, and $1/2\dmax < g \leq 1$, by simply replacing $g$ with the smaller value $g = 1/2\dmax$, and substituting into this stronger result.

Let $S$ denote the set of starting points from which a random walk locally mixes in time at most $\Delta$ and take $R := V \setminus S$. We start by showing that a random walk escapes $R$ with good probability.

Toward this goal, we first prove that every $T \subseteq R$ has one-sided conductance at least $\frac{\alpha g}{\dmax s}$. To see this, first note there must exist some $i$ such that
	\[|T \cap B_i| \geq \frac{|T|}{s}. \]
Now if $\Vol(T \cap B_i) \leq \frac{1}{2}\Vol_{B_i}(B_i)$, then 
	\[|E(T,\overline{T})| \geq \alpha \Vol_{B_i} (T \cap B_i) \geq \alpha \frac{|T|}{s} \geq \alpha \frac{\vol_G(T)}{s \dmax }.\]
So it follows that $\Psi(T) \geq \alpha/(s \dmax )  > \alpha g / (\dmax s)$.

On the other hand, suppose that $\Vol(T \cap B_i) \geq \frac{1}{2} \Vol_{B_i}(B_i)$. Since $B_i$ is $(\Delta, g)$-good, it follows that
	\[|\overline{T} \cap B_i| \geq |\overline{R} \cap B_i| \geq g |B_i|.\]
Thus, we have that
	\[|E(T,\overline{T})| \geq \alpha g |B_i| \geq \alpha g \frac{|T|}{s} \geq \alpha g \cdot \frac{\vol(T)}{\dmax  s}.\]
Thus, $\Psi(T) \geq \frac{\alpha g}{\dmax s}$, as desired.
	
In the remainder of the proof, we will show that a random walk leaves $R$ after at most $\wt{O}(s^2d^2_{\mathrm{max}}/g^2 \alpha^2)$ time steps with good probability by \Cref{lem:escape-ub}, and that once we've left $R$, we locally mix with respect to some set $A_i$ by \Cref{lem:super-cover-mixing}, which will complete the proof.

To formally prove this, note first that for every $v \in S$, we have by \Cref{lem:super-cover-mixing} that
	\[T_{\mix}^{1/16s\dmax , 4 s \eps}(v, \mathcal{C}) \leq T_{\mix}^{1/8\dmax , \eps}(v, \mathcal{P}(\mathcal{C})) \leq \Delta.\] 	

We now conclude that for any vertex $v \in R$, there exist an integer $a = \wt{O}(s^2d^2_{\mathrm{max}}/g^2 \alpha^2)$, an integer $b \leq \Delta$, and $i \in [s]$, such that a random walk for $a$ steps $\bX_0, ..., \bX_a$ leaves $R$, reaches a vertex $u \in S$ with $T_{\mix}^{1/16s\dmax , 4s\eps}(u, \mathcal{C}) = b$, and ``looks like'' $\mathcal{U}_{A_i}$ (in the sense of having total variation distance at most $4s\eps$)  after conditioning on the event $E_u$ (from the local mixing of $u$). We let $E_v$ denote the event that the three of the previous items occur. Since the walk is Markovian, it follows from \Cref{lem:escape-ub} that we can choose $a,b,i$ such that the probability that $E_v$ occurs is at least $\Omega \left( \frac{\alpha^2 g^2}{s^4 \dmax ^3 \Delta}\right)$.

We next aim to show that $\bX_{a+b}|E_v$ is close to $\mathcal{U}_i$ in total variation distance. To prove this, we note that the distribution $\bX_{a+b}|E_v$ is a convex combination of distributions with total variation distance at most $\eps$ to $\mathcal{U}_i$. Indeed, for a suitable normalization constant $\gamma$, we have that 
\begin{align*}
	\Pr[\bX_{a+b} = w|E_v] &= \gamma \sum_{u \in S: T_{\mix}(u) = b} \Pr[\bX_a = u \land \bX_{a+b} = w \land E_u] \\
	&= \gamma \sum_{u \in S: T_{\mix}(u) = b} \Pr[\bX_a = u] \Pr[E_u|\bX_a = u] \cdot \Pr[\bX_{a+b} = w | \bX_a = u \land E_u].
\end{align*}

Noting that
	\[\gamma^{-1} = \sum_{u \in S: T_{\mix}(u) = b} \Pr[\bX_a = u] \Pr[E_u|\bX_a = u]\]
shows that $\bX_{a+b}|E_v$ is indeed the convex combination of such distributions. Thus, by Jensen's inequality, $\bX_{a+b}|E_v$ also has total variation distance at most $4s\eps$ to $\mathcal{U}_i$. We have thus proved $T_{\mix}^{p, 4s\eps}(v) \leq
\wt{O} \left(\Delta + \frac{s^2 d^2_{\max}}{g^2 \alpha^2} \right)$, 
as desired.
\end{proof}

\subsection{The Largest Set is Good}
\label{sec:largestsetgood}

In this section, we prove that $A_1$ (the biggest set in the cover) is contained in a set $A_I$ that is good and such that $G[A_I]$ has large expansion. To prove this, we roughly show that if no such set exists with expansion $\lambda^\ell$, then the largest set $I \subseteq [s]$ that contains $1$ for which $A_I$ has expansion $\lambda^\ell$ satisfies $|I| \geq \exp(\ell)$. For large enough $\ell \geq \Omega(\log s)$, this is impossible, so some such set must exist with expansion $\lambda^{O(\log s)}$.

In order to prove the statement, we need the following two definitions.

\begin{definition}[$\ell$-thick graph]
	Let $G$ be a graph and $\mathcal{C} = \{A_1, ..., A_s\}$ be a disjoint $\theta$-cover. We define the \emph{$\ell$-thick graph}, denoted $H_\ell$, as follows: $H_\ell$ is an edge-weighted undirected graph with $s$ vertices, where each vertex in $V(H_\ell)$ corresponds to a set $A_i$. We will refer to these vertices interchangeably as both $A_i$ and $i$. There is an edge between $A_i$ and $A_j$ if $|E(A_i, A_j)| \geq \lambda^{\ell} |A_1|$. The edge has weight $|E(A_i, A_j)|$.
\end{definition}

\begin{definition}[$\ell$-thick component]
	Let $G$ be a graph and $\mathcal{C} = \{A_1, ..., A_s\}$ be an disjoint $\theta$-cover. The \emph{$\ell$-thick component of $A_i$} is the subset $M_G^\ell(A_i) \subseteq [s]$ corresponding to the connected component of $A_i$ in $H_\ell$.
\end{definition}

These definitions may seem somewhat strange at first. Indeed, the most natural way of arguing that the largest $\lambda^\ell$-expanding set $I$ must be large would be to cluster the $A_i$'s into $\lambda^{\ell-1}$ expanders. Since these are bad, \Cref{lem:low-cond-good} implies that this a high conductance cut between two clusters. We can then merge them and continue. Unfortunately, however, we cannot afford to merge components in this way. Indeed, recall that we improved the mixing time bounds to avoid having to square the conductance each time and this is vital to getting a quasi-polynomial bound. Looking at \Cref{lem:cut-approx}  the expansion of fused sets is only promised to be $\Phi(G[A_i])^2$ if we only use that there are many edges between these two clusters of sets. As such one needs more care to circumvent this. Moreover, one must be very careful when combining sets as they can have very different sizes, which can cause problems. We essentially define the $\ell$-thick graph and component to handle these problems. 

The picture to have in mind with $\ell$-thick graphs is to initially start with the empty graph. Afterwards, we slowly increase $\ell$ and edges start to to appear. As such, a desirable property for a set to have, for us, is that it has expansion much larger than its conductance in $\ell$-thick graph \emph{and} no matter what edges appear in the $(\ell')$-thick graph for $\ell' > \ell$, the expansion of the component will still be bigger than the conductance. We will formally define this notion, which we call a ``redeemable'' set, shortly. But before this, we need to provide a definition describing how the $\ell$-thick graph controls the conductance.

\begin{definition}[Revealed Conductance] \label{def:revealed-conductance}
Let $G = (V,E)$ be a graph and $\mathcal{C} = \{A_1, ..., A_s\}$. The \emph{$\ell$-revealed conductance} of a set $I \subseteq [s]$ is equal to 
	\[
	\nu^{G,{\cal C}}_\ell(I)  := \frac{w(E_{H_\ell}(I, \overline{I})) }{|A_I|} + \frac{\lambda^{\ell} |A_1| \cdot |\{(i,j): i \in I, j \not \in I, (i,j) \not \in E(H_\ell)\}| }{|A_I|}.  \]
	When $G$ and ${\cal C}$ are clear from context, we may omit them and simply write $\nu_\ell(I)$ instead of $\nu^{G,{\cal C}}_\ell(I)$.
\end{definition}

Intuitively, $\nu_\ell(I)$ is an upper bound on the one-sided conductance, $\Psi(A_I)$, of the set $A_I$, conditioned on only knowing the weights of the $\ell$-thick edges:

\begin{lemma}
\label{lem:rev-conductance-upper-bounds}
	Let $G = (V,E)$ be a graph and $\mathcal{C} = \{A_1, ..., A_s\}$. For any $\ell$ and any $I \subseteq [s]$ we have \[ \Psi(A_I) \leq \nu^{G,{\cal C}}_\ell(I). \]
\end{lemma}

\begin{proof}
	Recall that
	\begin{align*}
		\Psi(A_I) = \sum_{i \in I} \sum_{j \in \overline{I}} \frac{|E(A_i, A_j)|}{|A_I|}.
	\end{align*}
	By definition of the $\ell$-thick graph, we know that for $(i,j) \in E(H_\ell)$ we have that $|E(A_i, A_j)| = w(E_{H_\ell}(i,j))$, and for $(i,j) \notin E(H_\ell)$ we have that $|E(A_i, A_j)| < \lambda^\ell \cdot |A_1|$.
\end{proof}

We next show that increasing $\ell$, and thus intuitively revealing more information about edges, can only decrease the revealed conductance:

\begin{lemma} \label{lem:prayer}
For any set $I \subseteq [s]$ and any integer $\ell$, we have $\nu_{\ell + 1}(I) \leq \nu_{\ell}(I)$.
\end{lemma}

\begin{proof}
Let $F := E_{H_{\ell+1}}(I, \overline{I}) \setminus E_{H_{\ell}}(I, \overline{I})$. Since the edges in $F$ are not in $H_\ell$, we have that
	\[
w(F) \leq |F| \lambda^{\ell} |A_1|.\]
Thus,
\begin{align*}
	\nu_{\ell+1}(I) &= \frac{w(E_{H_{\ell+1}}(I, \overline{I})) }{|A_I|} + \frac{\lambda^{\ell+1} |A_1| \cdot |\{(i,j): i \in I, j \not \in I, (i,j) \not \in E(H_{\ell+1})\}| }{|A_I|} \\
	&= \frac{w(E_{H_{\ell}}(I, \overline{I})) + w(F) }{|A_I|} + \frac{\lambda^{\ell+1} |A_1| \cdot \left(  |\{(i,j): i \in I, j \not \in I, (i,j) \not \in E(H_{\ell})\}| - |F| \right) }{|A_I|} \\
	&\leq \frac{w(E_{H_{\ell}}(I, \overline{I})) }{|A_I|} + \frac{\lambda^{\ell+1} |A_1| \cdot \left(  |\{(i,j): i \in I, j \not \in I, (i,j) \not \in E(H_{\ell})\}| - |F|  \right) }{|A_I|} + 
	\frac{ |F| \lambda^\ell |A_1|}{|A_I| } \\
	&\leq \frac{w(E_{H_{\ell}}(I, \overline{I})) }{|A_I|} + \frac{\lambda^{\ell} |A_1| \cdot \left(  |\{(i,j): i \in I, j \not \in I, (i,j) \not \in E(H_{\ell})\}|  \right) }{|A_I|} \\
	&= \nu_\ell(I). \qedhere
\end{align*}
\end{proof}

As hinted at earlier, we won't work directly with good sets, but rather the following surrogate.

\begin{definition}[$\ell$-Redeemable Sets]
	Let $G$ be a graph and $\mathcal{C} = \{A_1, ..., A_s\}$ be a disjoint $\theta$-cover. We say that a set $A_i$ is \emph{$\ell$-redeemable} if there exists a set $I \subseteq [s]$ which satisfies the following three conditions:
	\begin{enumerate}[label=(\roman*)]
		\item there is a path in $H_\ell$ from $i$ to some $j \in I$,
		\item \label{item2} $\Phi(G[A_I]) \geq \lambda^{-1/3} \nu_\ell(A_I)$, and
		\item $H_\ell[I]$ is a connected graph.
	\end{enumerate}
\end{definition}

Intuitively, proving that $A_i$ is redeemable will help to show that a random walk starting from a vertex in $A_i$ will locally mix. Condition (i) says that the random walk has a good chance of getting into the set $A_I$, and conditions (ii) and (iii) say that, once the walk get into $A_I$, it has a good shot at staying in $A_I$ and mixing.

It's not hard to verify that the connected component of $A_i$, i.e., $M_G^\ell(A_i)$, contains a good set (namely $A_I$).
Moreover, since revealed conductance can only decrease as $\ell$ increases, it follows that if a set is $\ell$-redeemable, then it's also $\ell+1$-redeemable. We also observe that the induced subgraph of $G$ corresponding to an $\ell$-thick component in $H_\ell$ is a good expander:

\begin{lemma}
\label{lem:thick-graph-exp}
Let $G = (V,E)$ be a graph and $\mathcal{C} := \{A_1, \dots , A_s\}$ be a $\theta$-cover. If $I$ is an $\ell$-thick component, then $\Phi(G[A_I]) \geq \lambda^{\ell + 0.01}$.
\end{lemma}

\begin{proof}
	Fix the subset $J$ of $I$ that minimizes $\Phi_{G[A_I]}(A_J)$. Since $H_\ell[I]$ is connected, there must be edges between $I \setminus J$ and $J$ in $H_\ell$, and so there must exist an $i \in I \setminus J$ and a $j \in J$ such that \[|E(A_i, A_j)| \geq \lambda^\ell |A_1|.\]
	
	Note furthermore that 
		\[\Vol_{A_I}(A_J) \leq \Vol(G) \leq s \cdot d_{\max} \cdot |A_1|\]
	since there are at most $s\cdot|A_1|$ vertices in $G$, and each vertex has maximum degree at most $d_{\max}$.
	
	It thus follows that
		\[\frac{|E_{A_{I}}(A_J, \overline{A_J})| }{\vol_{A_I}(A_J)} \geq \frac{|E(A_i, A_j)|}{s \dmax |A_1|} \geq \frac{\lambda^{\ell}}{s \dmax }.\]
	The result now follows by \Cref{lem:cut-approx} by virtue of the choice of $\lambda$ as being so small compared to $\theta/(4s\dmax^2)$. 
\end{proof}

We now prove that if $A_1$ is not $(\ell+1)$-redeemable, then the $(\ell+1)$-thick component of $A_1$ must include a new set  $A_j$ that is fairly large.

\begin{lemma}
\label{lem:big-new-elem}
Suppose that $A_1$ is not $(\ell+1)$-redeemable, then there exists a $j \in M_G^{\ell + 1}(A_1) \setminus M_G^{\ell}(A_1)$ such that $|A_j| \geq \lambda |A_1|$. 
\end{lemma}

\begin{proof}
	Towards a contradiction, suppose no such set $A_j$ exists. We will derive a contradiction by examining how many rounds the loop will run in the following procedure:
	
\begin{center}
\fbox{\parbox{0.8\textwidth}{
	\begin{enumerate}
		\item Set $J_0 = M_G^\ell(A_1)$ and $k = 0$. 
	
		\item While there exists a $j \in [s] \setminus J_{k}$ such that $|E(A_{J_k}, A_j)| \geq s^3 \dmax  \lambda^{\ell + 1} |A_1|$:
		\begin{enumerate}
			\item Let $j$ be the choice which maximizes $|E(A_{J_k}, A_j)|$.
	
			\item Update $J_{k+1} \gets \{j\} \cup J_k$ and $k \gets k+1$.
		\end{enumerate}
	\end{enumerate}
}}
\end{center}
	
	We claim that the loop will run for $s$ rounds, which is a contradiction as we would have $|J_{s}| \geq s + 1$ even though $J_{s} \subseteq [s]$. We'll proceed by induction on $k$ to show that in each of the first $s$ rounds the protocol will find a $j$ in the next round such that $|E(A_{J_k}, A_j)| \geq s^3\dmax  \lambda^{\ell+1} |A_1|$.
	
	Suppose that the procedure has run for $k$ steps so far, and let $j_1, ..., j_k$ denote the choices of $j$ during the first $k$ rounds in order. We will first show that \[\Phi(G[A_{J_k}]) \geq \frac{\lambda^\ell \theta}{4s^2 d^3_{\mathrm{max}}}.\] In the case when $|J_k| = 1$, this follows by definition of the cover $\mathcal{C}$, so we will prove this in the case when $|J_k| > 1$.
	
Consider a nonempty set $S \subseteq J_k$ with $\vol_{A_{J_k}}(A_S) \leq \frac{1}{2} \vol_{A_{J_k}}(A_{J_k})$. Such a set must exist since $|J_k| > 1$, by taking any strict subset of $J_k$ or its complement. We now consider two cases, corresponding to whether or not $S \cap M_G^\ell(A_1)$ is empty, in order to show that $A_S$ has large conductance in $G[A_{J_k}]$.
	
	First, suppose that $S \cap M_G^\ell(A_1) \not = \emptyset$. In this case, we have that $\vol_{A_{J_k}}(A_{J_k \setminus M_G^\ell(A_1)}) \leq k \dmax  \lambda |A_1| \leq s \dmax  \lambda |A_1| < |A_1| \leq \vol(A_1)$, and hence we have that $M_G^\ell(A_1) \not \subseteq S$, since if $S$ contained $M^\ell_G(A_1)$ then this would violate the assumption that $\Vol_{A_{J_k}}(A_S) \leq {\frac 1 2} \Vol_{A_{J_k}}(A_{J_k})$. Since $M_G^\ell(A_1)$ is connected in $H_\ell$, it follows that there exists an $i \in M_G^\ell(A_1) \setminus S$ and a $j \in S$ such that
			\[|E(A_i, A_j)| \geq \lambda^\ell |A_1|.\]  
		Thus, we conclude that
			\[\Phi_{G[A_{J_k}]}(A_S) \geq \frac{|E(A_i, A_j)|}{\vol(A_S)} \geq \frac{\lambda^\ell |A_1|}{s \dmax  |A_1|} \geq \frac{\lambda^{\ell}}{s \dmax }.\]
			
	Second, suppose that $S \cap M^\ell_G(A_1) = \emptyset$. In this case, let $m$ be the smallest integer such that $A_{j_m} \in S$. By construction, $|E(A_{j_m}, A_{J_{m-1}})| \geq s^3 \dmax  \lambda^{\ell +1} |A_{1}|$. By the minimality of $m$, these edges all cross the cut $(S, \overline{S})$, so
		\[\Phi_{G[A_{J_k}]}(A_S) \geq \frac{|E(A_{j_m}, A_{J_{m-1}})|}{\vol(A_S)} \geq \frac{s^3 \dmax  \lambda^{\ell+1} |A_1|}{k \lambda |A_1| \dmax } \geq s^2 \lambda^{\ell}. \]  
	
	This concludes the case analysis; either way we have $\Phi_{G[A_{J_k}]}(A_S) \geq \frac{\lambda^{\ell}}{s \dmax }$, so we conclude by \Cref{lem:cut-approx} as desired that 
		\[\Phi(G[A_{J_k}]) \geq \frac{\lambda^\ell \theta}{4s^2 d^3_{\mathrm{max}}}.\]
	
	We now use this to conclude the proof. Since $A_1$ is not $(\ell+1)$-redeemable, 
	but $1 \in J_k$
	and $J_k$ is connected in $H_{\ell + 1}$, it follows that condition \ref{item2} in the definition of a redeemable set does not hold, i.e.,
		\[\nu_{\ell+1}(A_{J_k}) \geq \lambda^{1/3} \Phi(G[A_{J_k}]). \]
	
	The definition of revealed conductance gives that 
	\[\nu_{\ell+1}(A_{J_k})\leq \dmax \cdot \Psi_{G}(A_{J_k}) +  \frac{s^2 \lambda^{\ell + 1} |A_1|}{|A_{J_k}|},
	\]
	and so we can combine the above three inequalities to bound
		\[\Psi_{G}(A_{J_k}) \geq \frac{\lambda^{1/3}}{\dmax }  \Phi(G[A_{J_k}]) - \frac{s^2 \lambda^{\ell + 1} |A_1|}{|A_{J_k}|} \geq \lambda^{1/3} \frac{\lambda^\ell \theta}{4s^2 d^4_{\mathrm{max}}} - s^2 \lambda^{\ell + 1} \geq \lambda^{\ell + 1/2} \geq s^4\dmax  \lambda^{\ell+1}, \]
	where in the last two steps, we used the definition $\lambda = (1000 s\dmax \theta^{-1} \log(|G|)\log(1/\eps))^{-1000}$.
		
	It then follows that there exists a $j$ such that
		\[E(A_{J_k}, A_j) \geq s^3 \dmax  \lambda^{\ell + 1} |A_1|,\]
	which completes the inductive step. Thus the procedure runs would run for at least $s$ steps, a contradiction.
\end{proof}

Finally, we will use this to prove that, if $A_1$ is not $\ell$-redeemable, then $M_G^\ell(A_1)$ cannot be too small, and in fact must have size exponentially large in $\ell$. In other words, we aim to prove an exponential lower bound on the following quantity.

\begin{definition}[Minimum $\ell$ Component Size] \label{def:min-component-size}
We define the minimum $\ell$-component size as
		\[f(\ell) := \min |M_G^\ell(A_1)|, \]
	where the minimum is taken over all possible graphs $G$ and disjoint $\theta$-covers $\mathcal{C}$ of size $s$ of $G$, such that $A_1$ is not $\ell$-redeemable. If no such graph and cover exists, we define $f(\ell) := \infty$.
\end{definition}

\begin{lemma}
\label{lem:big-sets-big-components}
Let $G$ be a graph, $\mathcal{C} = \{A_1, ..., A_s\}$ a disjoint $\theta$-cover of $G$, and $j \in [s]$ be such that 
	$A_j$ is the largest element of its connected component $M_G^\ell(A_j)$. If $|A_j| \geq \lambda |A_1|$ and $A_j$ is not $\ell$-redeemable, then $|M_G^\ell(A_j)| \geq f(\ell - 1)$.
\end{lemma}

\begin{proof}
Define $A_{\geq j}:= \bigcup_{j' \geq j} A_{j'}$. Consider the induced graph $G[A_{\geq j}]$ and the cover of this graph $\{A_j, ..., A_{s},$
$B_{s+1}, ..., B_{s+j-1}\}$, where for $i\geq 1$ we define $B_{s+i} = \emptyset$. 

We first show that it suffices to show that $j$ is not $(\ell-1)$-redeemable in $G[A_{\geq j}]$. To see this we first set up some notation. Let $H_{\ell - 1}^{\geq j}$ denote the $(\ell-1)$-thick graph of $G[A_{\geq j}]$. Since by assumption $j$ is not $(\ell-1)$-redeemable in $G[A_{\geq j}]$, it follows that $|M_{G[A_{\geq j}]}^{\ell-1}(A_j)| \geq f(\ell -1)$. On the other hand, there is an edge between two components $A_a$ and $A_b$ in $H_{\ell-1}^{\geq j}$ if and only if $|E_{G[A_{\geq j}]} (A_a, A_b)| \geq \lambda^{\ell-1} |A_j| \geq \lambda^\ell |A_1|$. Thus $M_{G[A_{\geq j}]}^{\ell-1}(A_j) \subseteq M_G^\ell(A_j)$, which proves the result. 

It now remains to show that $j$ is not $(\ell-1)$-redeemable in $G[A_{\geq j}]$. We will prove this by contradiction, and thus assume to the contrary that it is $(\ell-1)$-redeemable, i.e., that there exists a set $J \subseteq M_{G[A_{\geq j}]}^{\ell-1}(A_j)$ such that 
\begin{enumerate}[label=(\roman*)]
	\item there is a path in $H_{\ell-1}^{\geq {j}}$ from $j$ to a some $i \in J$,
	\item $\Phi(G[A_J]) \geq \lambda^{-1/3} \nu_{\ell-1}^{H_{\ell-1}^{\geq j}}(J)$, and 
	\item $H_{\ell-1}^{\geq j}[J]$ is a connected graph.
\end{enumerate} Since every edge in $H_{\ell-1}^{\geq j}$ corresponds to an edge in $H_\ell$, we have that (i) and (iii) hold in $H_\ell$ as well. 
We will next prove that $\nu_{\ell-1}^{H^{\geq j}_{\ell-1}}(J) \geq \nu_{\ell}^{H_{\ell}}(J)$. From this, it follows that (ii) also holds in $H_\ell$, and hence that $j$ is $\ell$-redeemable in $G$, a contradiction to the initial assumption that $A_j$ is not $\ell$-redeemable, completing the proof.

We finally prove the desired inequality, that $\nu_{\ell-1}^{H^{\geq j}_{\ell-1}}(J) \geq \nu_{\ell}^{H_{\ell}}(J)$. Recall that these quantities are defined as:

$$|A_J| \cdot \nu_{\ell-1}^{H^{\geq j}_{\ell-1}}(J) = w(E_{H^{\geq j}_{\ell-1}}(J,\overline{J})) + \lambda^{\ell-1} |A_j| \cdot |\{ (a,b) : a \in J, b \notin J, (a,b) \notin E(H^{\geq j}_{\ell-1}) \}|\},$$

$$|A_J| \cdot \nu_{\ell}^{H_{\ell}}(J) = + \lambda^{\ell} |A_1| \cdot |\{ (a,b) : a \in J, b \notin J, (a,b) \notin E(H_{\ell}) \}|\}.$$

We introduce notation to slightly rewrite these sums. For $a \in J$ and $b \notin J$, define the quantities
$$q_0(a,b) := \begin{cases}
	w(E_{H^{\geq j}_{\ell-1}}(\{a\},\{b\})), & \text{if $(a,b) \in E(H^{\geq j}_{\ell-1})$}\\
	\lambda^{\ell-1} |A_j|, & \text{otherwise,}
\end{cases}$$

$$q_1(a,b) := \begin{cases}
	w(E_{H_{\ell}}(\{a\},\{b\})), & \text{if $(a,b) \in E(H_{\ell})$}\\
	\lambda^{\ell} |A_1|, & \text{otherwise.}
\end{cases}$$

We can then rewrite the above as

$$|A_J| \cdot \nu_{\ell-1}^{H^{\geq j}_{\ell-1}}(J) = \sum_{a \in J, b \notin J} q_0(a,b),$$

$$|A_J| \cdot \nu_{\ell}^{H_{\ell}}(J) = \sum_{a \in J, b \notin J} q_1(a,b).$$

We will show that, for all $a \in J$ and $b \notin J$, we have $q_0(a,b) \geq q_1(a,b)$, which will complete the proof. We will show this by casework, depending on whether or not $(a,b) \in E(H^{\geq j}_{\ell-1})$.

If $(a,b) \in E(H^{\geq j}_{\ell-1})$, then we have $$w_{H^{\geq j}_{\ell-1}}(\{a\},\{b\}) \geq \lambda^{\ell-1}|A_j| \geq \lambda^\ell |A_1|,$$ where the first inequality is by definition of the $(\ell-1)$-thick graph, and the second is because $|A_j| \geq \lambda |A_1|$. Thus, $(a,b) \in E(H_{\ell})$ as well, and so we have $q_0(a,b) = q_1(a,b) = w(E_{H_{\ell}}(\{a\},\{b\}))$.

If $(a,b) \notin E(H^{\geq j}_{\ell-1})$, then $q_0(a,b) = \lambda^{\ell-1} |A_j|$. If $q_1(a,b) = \lambda^\ell |A_1|$ then we have $q_0(a,b) \geq q_1(a,b)$ because $|A_j| \geq \lambda |A_1|$. Otherwise, we have $$q_1(a,b) = w(E_{H_\ell} (\{a\},\{b\})) \leq \lambda^{\ell-1}|A_j| = q_0(a,b),$$
where the inequality follows from $(a,b) \notin E(H^{\geq j}_{\ell-1})$ by definition of the $(\ell-1)$-thick graph. This completes the proof.
\end{proof}

\begin{lemma}
\label{lem:thick-comp-bound}
$f(\ell) \geq F_\ell$, where $F_\ell$ is the $\ell$th Fibonacci number.
\end{lemma}

\begin{proof}
Let $G = (V,E)$ be a graph and $\mathcal{C} = \{A_1, ..., A_s\}$ a disjoint $\theta$-cover such that $A_1$ is not $\ell$-redeemable and $M_G^\ell(A_1) = f(\ell)$. Note that if no such graph and cover exist, then $f(\ell) = \infty$ by definition and so the statement holds. 

We now argue that $f(\ell) \geq f(\ell-1) + f(\ell - 2)$. First, by \Cref{lem:big-new-elem}, there exists a $j \in M_G^\ell(A_1) \setminus M_G^{\ell-1}(A_1)$ with $|A_j| \geq \lambda |A_1|$. Without loss of generality, let $A_j$ denote the largest element in $M_G^{\ell-1}(A_j)$.

Let us now show that $A_j$ is not $\ell$-redeemable. Assume to the contrary, and let $A_I$ be the set certifying that $A_j$ is $\ell$-redeemable. We use this to prove that $A_1$ is $\ell$-redeemable as well. Since $1$ and $j$ are in the same connected component in $H_\ell$, it follows that condition (i) in the definition of an $\ell$-redeemable set holds for $A_1$. Then, since conditions (ii) and (iii) only depend on $\ell$ and $A_I$, they are also satisfied. Thus $A_1$ is $\ell$-redeemable, a contradiction. 

We may thus assume $A_j$ is not $\ell$-redeemable. By \Cref{lem:big-sets-big-components}, we thus have $|M_G^{\ell-1}(A_j)| \geq f(\ell-2)$. Note also that, since $A_1$ is not $\ell$-redeemable, it follows that it is not $(\ell-1)$-redeemable, and so $|M^{\ell-1}_G(A_1)| \geq f(\ell-1)$ by definition of $f$. Finally, since  $j \in M_G^\ell(A_1) \setminus M_G^{\ell-1}(A_1)$ by definition, we can bound
	\[f(\ell) = |M_G^\ell(A_1)| \geq |M_G^{\ell-1}(A_1)| + |M_G^{\ell-1}(A_j)| \geq f(\ell -1) + f(\ell-2).\]
The result now follows by induction and noting that $f(1) \geq 1$ and $f(2) \geq 1$.
\end{proof}

As an immediate corollary of \Cref{def:min-component-size} and \Cref{lem:thick-comp-bound}, we get:

\begin{corollary}
\label{cor:A1-redeemable}
Let $G = (V,E)$ be a graph and
$\calC = \{A_1,\dots,A_s\}$ be a disjoint $\theta$-cover.
Then $A_1$ (the biggest set in the cover) is $O(\log(s))$-redeemable.
\end{corollary}

Finally, we see that this also implies that $A_1$ is contained in a good set with reasonable expansion as promised:

\begin{corollary}
\label{cor:A1-good}
Let $G$ be a graph and $\calC = \{A_1,\dots,A_s\}$ be a disjoint $\theta$-cover. Then, there exists a set $S \subseteq [s]$  with $1 \in S$ such that $A_S$ is a $\lambda^{O(\log(s))}$ expander and $A_S$ is $(\lambda^{-O(\log(s))},\lambda^{O(\log(s))})$ good.
\end{corollary}

\begin{proof}
	Let $\ell = O(\log(s))$ as per \Cref{cor:A1-redeemable}. We'll show that it suffices to take $S = M^\ell_G(A_1)$. Clearly, we have that $1 \in S$ and $A_S$ has expansion $\lambda^{O(\log(s))}$ by \Cref{lem:thick-graph-exp}.
	
	Finally, to show that $A_S$ is good, note that as $A_1$ is $\ell$-redeemable, there exists a set $I \subseteq [s]$ such that:
	\begin{enumerate}[label=(\roman*)]
		\item there is a path in $H_\ell$ from $1$ to some $j \in I$,
		\item $\Phi(G[A_I]) \geq \lambda^{-1/3} \nu_\ell(A_I)$, and
		\item $H_\ell[I]$ is a connected graph.
	\end{enumerate}
	
	Conditions (i) and (iii) imply that $I \subseteq S$. Using (ii) and \Cref{lem:rev-conductance-upper-bounds}, then yields that
		\[\Psi(A_I) \leq \nu_\ell(A_I) \leq \lambda^{1/3} \Phi(G[A_I]),\]
	and so applying \Cref{lem:low-cond-good} then yields that $A_I$ is $(\lambda^{-O(\log(s))}, \frac{1}{2})$ good. 
	
	Now, if $1 \not \in I$, then there must be some edge incident to $I$ in $H_\ell$ and thus at least $\lambda^{\ell} |A_1|$ edges incident to $A_I$. This gives us that
		\[|A_I| \geq \frac{\lambda^{\ell} |A_1|}{d_{\mathrm{max}}}.\]
	On the other hand, if $1 \in I$, then we trivially have that 
	\[|A_I| \geq |A_1| \geq \frac{\lambda^{\ell} |A_1|}{d_{\mathrm{max}}}.\]
	As $|A_S| \leq s |A_1|$, we immediately have that $A_S$ is $(\lambda^{-O(\log(s))}, \lambda^{O(\log(s))})$-good as desired.
\end{proof}

\subsection{All Sets in the Cover are Contained in Good Expander}
\label{sec:allsetsingood}
The previous section proved that the largest set $A_1$ is always contained a good set with reasonable expansion.
 We will now leverage this to conclude that the smaller sets $A_2, ..., A_s$ are also contained in good, expansive sets.

To give an overview of the argument we briefly describe it for $A_2$. We first consider restricting to $G[A_{\geq 2}]$. In this graph we note that $A_{2}$ is $O(\log(s))$ redeemable and hence good. If there are many edges to $A_1$, then we are done, as we can combine these components and the resulting component will be a reasonable expander and a good set. If there are no edges, then we conclude that $A_2$ itself is good.

Note that unpacking the argument for $A_3$ and so on, we end up potentially forming a path to the good set containing $A_1$. Because of this, we cannot iteratively apply \Cref{lem:cut-approx}, as the expansion lower bound would suffer too much. Fortunately, the sets will roughly decrease in size, so we can do better in the following lemma.

\begin{lemma}
\label{lem:pyramid-expansion}
	Let $G$ be a graph and $\mathcal{C} := \{B_1, ..., B_s\}$ be a disjoint-$\theta$ cover. (Note that we do not assume that the $B_i$ are sorted by size.) Assume that $|B_i| \geq \frac{|B_{j}|}{s}$ for all $i < j$. Moreover, suppose that $|E(B_i, B_{i+1})| \geq \alpha |B_{i+1}|$. It then follows that
		\[\Phi(G) \geq \lambda^{0.01} \alpha \theta. \]
\end{lemma}

\begin{proof}
Consider some set $S \subseteq [s]$. We will lower bound $\Phi(B_S)$. Without loss of generality, assume $1 \not \in S$. Let $i$ denote the smallest index such that $i \not \in S$ and $i+1 \in S$. We thus have
\[|E(B_S, \overline{B_S})| \geq |E(B_{i+1}, B_i)| \geq \alpha |B_{i+1}|.\]
We then conclude that
	\[\Phi(B_S) \geq \frac{\alpha |B_{i+1}|}{\vol(B_S)} \geq \frac{\alpha |B_{i+1}|}{s^2 \dmax  |B_{i+1}|} \geq \frac{\alpha}{s^2 \dmax },\]
where the bound $\vol(B_S) \leq s^2 \dmax  |B_{i+1}|$ used in the second step follows because, by assumption, $|B_j| \leq s \cdot |B_{i+1}|$ for all $j \in S$.
By \Cref{lem:cut-approx}, it then follows that
	\[\Phi(G) \geq \frac{\alpha \theta}{4s^3\dmax ^3} \geq \lambda^{0.01} \alpha \theta,\]
by the definition $\lambda = (1000 s\dmax \theta^{-1} \log(|G|)\log(1/\eps))^{-1000}$, as desired.
\end{proof}

We now turn to formalizing the proof sketch presented for this subsection.

\begin{lemma}
\label{lem:ladder-rungs}
Let $G$ be a graph and $\calC = \{A_1,\dots,A_s\}$ be a disjoint $\theta$-cover. Then there exist disjoint sets $A_{S_1}, \dots, A_{S_m}$ (i.e.,~$S_1,\dots,S_m$ are disjoint subsets of $[s]$) that form a $\lambda^{O(\log(s))}$-cover of $G$ and satisfy $|A_{S_i}| \geq \frac{1}{s} |A_{S_j}|$ for any $i < j$. Moreover, for every $i$ either $A_{S_i}$ is $(\lambda^{-O(\log(s))}, \lambda^{O(\log(s))})$-good or $|E(A_{S_i}, A_{S_j})| \geq \lambda^{O(\log(s))} \cdot |A_{S_i}|$ for some $j < i$.
\end{lemma}

\begin{proof}
Let $\ell$ be such that the $\ell$th Fibonacci number is strictly greater than $s$.

We build the sets via the following procedure: 

\begin{center}
\fbox{\parbox{0.95\textwidth}{
~~~~Let $S = \emptyset$ and $k = 1$. 
~~~~While $S \not= [s]$ do the following: 
		\begin{itemize}
\item Let $i$ be the smallest index such that $i \not \in S$.
 Consider the graph $G[A_{\geq i}]$ and cover $\{A_i, A_{i+1},\dots,A_s,B_{s+1}, ..., B_{s+i-1}\}$ where $B_{s+j} = \emptyset$ for all $j$. 
		
\item Set $S_k = M_{G[A_{\geq i}]}^\ell(A_i)$. 
		
\item If $S_k \cap S \not = \emptyset$, then update $S_k$ to be the connected component of $A_i$ in $H_\ell^{G[A_{\geq i}]}[[s] \setminus S]$.
		
\item Finally, update $S \gets S \cup S_k$ and $k \gets k +1$. 
\end{itemize}
}}
\end{center}

It is clear that the sets $A_{S_1},\dots,A_{S_m}$ are disjoint by construction and that the process will terminate. Moreover, by \Cref{lem:thick-graph-exp}, it follows that the graphs $S_j$ satisfy $\Phi(G[S_j]) \geq \lambda^{\ell + 0.01} \geq \lambda^{O(\log(s))}$. To see that $|A_{S_i}| \geq \frac{1}{s} |A_{S_j}|$ for any $i < j$, let $A_a$ be the largest uncovered element at the time of adding $S_i$. Since $|A_a| \geq |A_b|$ for all $b \geq a$ and $S_j$ only contains elements $A_b$ with $b > a$ when $i < j$, it follows that $|A_{S_i}| \geq |A_a| \geq \frac{1}{s} |A_{S_j}|$ for all $i < j$.

It now remains to show that each $S_k$ is good or has many edges to an earlier 
set. Fix some $k$, and let $A_i$ be the largest uncovered set before $S_k$ was added. We consider two cases, depending on whether $S_k = M_{G[A_{\geq i}]}^\ell(A_i)$.

Assume first that $S_k \not = M_{G[A_{\geq i}]}^\ell(A_i)$. In this case, it follows that $S_k$ is the connected component containing $i$ in $H_\ell^{G[A_{\geq i}]}[[s] \setminus S]$. However, in this case, there must exist some element $a \in S_k$ that had an edge to a vertex $b \in S$. But then by definition of the $\ell$-thick graph, $|E(A_{a}, A_{b})| \geq \lambda^\ell |A_i| \geq \lambda^{O(\log(s))} |A_{S_k}|$, so the edge condition is met.

For the remainder of the proof, we thus assume that $S_k = M_{G[A_{\geq i}]}^\ell(A_i)$. Now applying \Cref{lem:thick-comp-bound}, it follows that $A_i$ is $\ell$-redeemable in $G[A_{\geq i}]$. Let $J$ be the set certifying that $A_i$ is $\ell$-redeemable in $G[A_{\geq i}]$. Moreover, note that $|A_J| \geq \frac{1}{\dmax } \lambda^{\ell} |A_i| \geq \lambda^{\ell+1} |A_i|$ since it is connected to $A_i$ in the $\ell$-thick graph.
	
	If $A_J$ is $(\lambda^{-O(\log(s))}, \frac{1}{2})$ good in $G$, then $M_{G[A_{\geq i}]}^\ell(A_i)$ is $(\lambda^{-O(\log(s))}, \lambda^{O(\log(s))})$-good, satisfying the goodness condition. So, we assume $A_J$ is not $(\lambda^{-O(\log(s))}, \frac{1}{2})$ good in $G$. Applying the contrapositive of \Cref{lem:low-cond-good} then gives us that
		\[\Psi_G(A_J) \geq  \lambda^{0.01} \Phi(G[A_J]). \]

	On the other hand, since $J$ certifies $A_i$ being $\ell$-redeemable in $G[A_{\geq i}]$, we have that
		\[\Phi(G[A_J]) \geq \lambda^{-1/3} \nu_{\ell}^{H_{\ell}^{\geq i}}(J) \geq \lambda^{-1/3} \Psi_{G[A_{\geq i}]}(A_J),\]
	where $H_\ell^{\geq i}$ is the $\ell$-thick graph of $G[A_{\geq i}]$. This implies that
		\[\Psi_G(A_J) - \Psi_{G[A_{\geq i}]}(A_j) \geq \left(\lambda^{0.01} - \lambda^{1/3} \right)\Phi(G[A_J]) \geq \frac{1}{2} \lambda^{0.01} \Phi(G[A_J]) . \]
	This then implies that there exists an $a < i$ and a $b \in A_J \subseteq S_k$ such that 
		\[|E(A_a, A_b)| \geq \frac{1}{2 s^2} \lambda^{0.01} \Phi(G[A_J]) |A_J| \geq \frac{1}{2} \lambda^{2\ell + 1.03} |A_i|, \]
	where the final inequality follows by \Cref{lem:thick-graph-exp}. Since $\ell = O(\log(s))$ the edge condition is satisfied, and this completes the proof. 
\end{proof}

Finally, we show that the above structural lemma implies that every set $A_i$ is contained in a good, expansive set.

\begin{lemma}
\label{lem:good-cover-exists}
Let $G = (V,E)$ and $\mathcal{C} = \{A_1, ..., A_s\}$ be a disjoint $\theta$-cover, then for every $i \in [s]$, there exists a set $I \subseteq [s]$ with $i \in I$ such that $\Phi(G[A_I]) \geq \lambda^{O(\log(s))}$
 and $A_I$ is $(\lambda^{-O(\log(s))}, \lambda^{O(\log(s))})$-good.
\end{lemma}

\begin{proof}
By \Cref{lem:ladder-rungs}, given an $i$ there exist some disjoint $S_{j_1}, ..., S_{j_k}$ with $j_1 < j_2 < \dots < j_k$ such that 

\begin{enumerate}[label=(\roman*)]
	\item $i \in S_{j_k}$, 
	
	\item $S_{j_1}$ is $(\lambda^{-O(\log(s))}, \lambda^{O(\log(s))})$-good,
	
	\item $|E(A_{S_{j_m}}, A_{S_{j_{m+1}}})| \geq \lambda^{O(\log(s))} \cdot |A_{S_{j_{m+1}}}|$ for all $m \in [k-1]$,
	
	\item $|A_{S_{j_a}}| \geq \frac{1}{s} |A_{S_{j_b}}|$ for all $a \leq b$, and
	
	\item For all $m \in [k]$, we have $\Phi(G[A_{S_{j_m}}]) \geq \lambda^{O(\log(s))}$.
\end{enumerate}

It now follow from \Cref{lem:pyramid-expansion} that $G \left[ \bigcup_{m = 1}^k S_{j_m} \right]$ is a $\lambda^{O(\log(s))}$-expander. Furthermore, since $S_{j_1}$ is $(\lambda^{-O(\log(s))}, \lambda^{O(\log(s))})$-good, it follows that $\bigcup_{m = 1}^k S_{j_m}$ is $(\lambda^{-O(\log(s))}, \frac{1}{s^2} \lambda^{O(\log(s))})$-good by property (iv). This concludes the proof.
\end{proof}

\subsection{Putting it All Together}

We can now reap the glorious benefits of our labor and prove \Cref{thm:disjoint-local-mixing}.

\begin{proof}[Proof of \Cref{thm:disjoint-local-mixing}]
	By \Cref{lem:good-cover-exists}, there exists a super cover of $\mathcal{C}$, which we write as $\mathcal{C}' = \{B_1, ..., B_s\}$, such that $\Phi(G[B_i]) \geq \lambda^{O(\log(s))}$ and each $B_i$ is $(\lambda^{-O(\log(s))}, \lambda^{O(\log(s))})$-good. It then follows by \Cref{lem:good-mixing} that
		\[T_{\mix}^{p,4\eps s}(G, \mathcal{C}) \leq \wt{O}(\lambda^{-O(\log(s))})\]
	for some $p = \lambda^{O(\log(s))}$. 
\end{proof}

\part{Applications to Learning DNF}
\label{part:dnf}


\newcommand{\UTH}{{\large \red{\bf UP TO HERE}}}
\newcommand{\dterm}{\mathrm{d}_{\mathrm{term}}}
\newcommand{\dass}{\mathrm{d}_{\mathrm{sat}}}
\newcommand{\ol}[1]{\overline{#1}}
\newcommand{\simple}{\textsc{Simple-Learning}}

\section{Preliminaries}
\label{sec:notation-exact-k}

We start by setting up some useful notation and terminology. Given a string $y \in \zo^n$, we define the \emph{literal corresponding to $y_i$} as $x_i$ if $y_i = 1$ and $\overline{x}_i$ otherwise. We will frequently identify a term with the set of literals it contains; the size of a term $T$, written $|T|$, will thus be the size of the set of its literals. Given two terms $T_1, T_2$, we define the \emph{(term) distance between $T_1, T_2$} as 
\[
	\dterm(T_1, T_2) := \min\cbra{|T_1\setminus T_2|, |T_2 \setminus T_1|}
\]
where we identified the terms $T_1$ and $T_2$ with their set of literals. 

\begin{remark} \label{rem:term-dist-for-same-size}
	Note that when $|T_1| = |T_2|$ (which will certainly be the case for us in~\Cref{sec:learning-exact-dnf} when we learn exact DNF), we have $\dterm(T_1, T_2) = |T_1 \setminus T_2| = |T_2 \setminus T_1|$. 
\end{remark}

Given a string $y\in\zo^n$ and a set $S \sse [n]$, we will write $T_y(S)$ for the term induced by $y$ on the set $S$, i.e.
\[
	T_y(S) := \bigwedge_{\substack{i \in S \\ y_i = 1}} x_i \wedge \bigwedge_{\substack{j \in S \\ y_j = 0}} \overline{x}_j.
\]
Given a term $T$ and a set $S \sse [n]$, we write $T|_S$ for the term $T$ restricted to the set $S$, i.e. 
\[
	T|_S := \bigwedge_{\substack{i \in S \\ \ell_i \in T}} \ell_i
\]
where $\ell_i$ is a literal corresponding to the $i^{\text{th}}$ variable, i.e. $\ell_i = x_i$ or $\ol{x}_i$. 

We will identify a collection of terms with a DNF in the natural way, i.e. given a set of terms $\calL$, we view it as the DNF  
\[
	\calL(x) = \bigvee_{T\in\calL} T(x).
\]
Finally, given a point $x \in \zo^n$ and a collection of terms $\calL$, we define the \emph{(satisfying) distance between $x$ and $\calL$} as the Hamming distance between $x$ and the closest satisfying assignment of $\calL$, i.e. 
\[
	\dass(x, \calL) = \min_{\substack{ y \in \zo^n \\  \calL(y) = 1}} \|x-y\|_1.
\]

\subsection{Weak Term Learners}
\label{subsec:weak-term-learners}

\begin{definition}[Weak terms and weak term learners] \label{def:weak-term-learner}
Let $f=T_1 \vee \cdots \vee T_s$ be an unknown target $s$-term DNF over $\zo^n$ and let ${\cal D}$ be an arbitrary and unknown distribution.  
A term $T$ is said to be a \emph{$\gamma$-weak term} (with respect to $f$ and ${\cal D}$) if
\begin{itemize}
\item
[$(i)$] 
$\Pr_{\bx \sim {\cal D}}[T(\bx) = 1 \land f(\bx) = 0] \leq {\frac{\gamma}{s \log(1/\gamma)}}$,
and
\item
[$(ii)$] $\Pr_{\bx \sim {\cal D}}[T(\bx) = 1  \ | \ f(\bx) = 1] \geq \frac{1}{2s}$.
\end{itemize}

An algorithm which outputs a single term is said to be a \emph{$\gamma$-weak term learner} if with probability at least $1 - {\frac 1 {\omega(s  \log(1/\gamma))}}$ it outputs a $\gamma$-weak term.
\end{definition}

Intuitively, a weak term learner is an algorithm which can find a single term $T$ which simultaneously ``covers'' a non-trivial portion of the correct satisfying assignments of $f$ and only a very small portion of the non-satisfying assignments of $f$ with respect to ${\cal D}.$
We now show that given an efficient weak term learner, it is possible to efficiently construct a high-accuracy DNF hypothesis for the target DNF via a simple boosting-type procedure:

\begin{lemma}
\label{lem:term-booster}
Suppose that $\alg$ is a $\gamma$-weak term learner running in time $T$.
Then the procedure \DNFLearn~makes $O(s \log(1/\gamma))$ calls to \alg~
and with probability at least 99/100 outputs a DNF $h$, with $O(s \log(1/\gamma))$ terms, such that $\Pr_{\bx \sim {\cal D}}[h(\bx) \neq f(\bx)] \leq O(\gamma).$
The running time of \DNFLearn~is at most $\poly(s,1/\gamma,n)\cdot T.$
\end{lemma}

\begin{proof}
To begin with some simplifying assumptions, we assume that $\textsf{ALG}$ succeeds in outputting a $\gamma$-weak term in each of the calls that \DNFLearn~makes to it; by a union bound this incurs a total failure probability of only $o(1)$.
We also assume that $2 \gamma \leq \Pr_{\bx \sim {\cal D}}[f(\bx) = 1]  \leq 1-\gamma$, as otherwise it is trivially easy to output a DNF $h$ (the constant-0 function or the constant-1 function) satisfying the required error bound.

The \DNFLearn~algorithm is given as \Cref{alg:DNFLearn}.
As another simplifying assumption, we assume that each of the empirical estimates of $\Pr_{\bx \sim \mathcal{D}}[\bigvee_{0 \leq j < i} T_j(\bx) =0 \land f(\bx) = 1]$ that are computed in Step~2 is accurate to within $\pm \gamma/4$ (this adds at most another $o(1)$ to the failure probability, as discussed below), so each time Step~2(a) is reached we have that the true value of $\Pr_{\bx \sim \mathcal{D}}[\bigvee_{0 \leq j < i} T_j(\bx) =0 \land f(\bx) = 1]$ is at least $\gamma/4.$

\begin{figure}
\begin{algorithm}[H]
\addtolength\linewidth{-2em}

\vspace{0.5em}

\textbf{Input:} Query access to $f: \zo^{n} \rightarrow \zo$, sample access to $\calD$, access to a $\gamma$-weak term learner $\alg$,  parameter  $\gamma \in (0,1/2]$ \\[0.25em]
\textbf{Output:} An $O(s \log(1/\gamma))$-term DNF hypothesis $h$

\

\DNFLearn($f, \calD, \alg, \gamma$):

\begin{enumerate}
	\item Set ${\cal D}_0 = {\cal D}$ and $i = 0$. 
	\item 
	While an empirical estimate of $\Pr_{\bx \sim \mathcal{D}}[\bigvee_{0 \leq j < i} T_j(\bx) =0 \land f(\bx) = 1]$ is at least $\gamma/2$:
		\begin{enumerate}
			\item Run $\alg$ on distribution ${\cal D}_i$ and let $T_i$ be the term it returns.
			\item Set $\mathcal{D}_{i+1}$ to be the distribution that is an equal mixture of $\mathcal{D} | f(\bx) = 0$ and $\mathcal{D} | ( \bigvee_{j \leq i} T_j(\bx) =0 \land f(\bx) = 1)$.
			\item $i \gets i+1$.  
		\end{enumerate}
	\item Output the DNF hypothesis $h(x) := \bigvee_i T_i(x)$.
\end{enumerate}
\caption{The \DNFLearn~algorithm for learning DNFs, given access to a weak term learner.}
\label{alg:DNFLearn}
\end{algorithm}
\end{figure}

Let us analyze the efficiency of \DNFLearn. It is clear that each empirical estimate referred to in Step~2 can be computed, with failure probability at most $\delta$, in time $\poly(1/\gamma,\log(1/\delta)).$ Moreover, recalling that $\Pr_{\bx \sim {\cal D}}[f(\bx) = 1] \leq 1-\gamma$, we see that at each execution of the while-loop it is possible to efficiently simulate draws from ${\cal D}_i$ given access to draws from ${\cal D}$ with at most an $O(1/\gamma)$ running time overhead.
Thus the claimed running time of \DNFLearn~follows from the fact that \DNFLearn~makes $O(s \log(1/\gamma))$ calls to \alg, which we now establish.

Let us show that under the assumptions given above, the \DNFLearn~algorithm terminates after at most $O(s \log(1/\gamma))$ iterations of the while-loop in Step~2.
Towards this end, we define
\[
h_i(x) := \bigvee_{j<i} T_j(x).
\]
The following claim shows that the fraction of $f$'s satisfying assignments which are incorrectly labeled 0 by the hypothesis decreases at each iteration of the while-loop:

	\begin{claim} \label{claim:lollipop}
		\[\Prx_{\bx \sim \mathcal{D}}[h_{i+1}(\bx) = 0 \ | \ f(\bx) = 1] \leq \left( 1 - \frac{1}{2s} \right) \cdot \Prx_{\bx \sim \mathcal{D}}[h_{i}(\bx) = 0 \ | \ f(\bx) = 1]. \]	
	\end{claim}
	
	\begin{proof}
		We have that
			\begin{align*}
				\frac{1}{2s} &\leq \Prx_{\bx \sim \mathcal{D}_{i+1}}[T_{i+1} = 1 \  | \ f(\bx) = 1] = \Prx_{\bx \sim \mathcal{D}}[T_{i+1} = 1 \ | \ h_i(\bx) = 0 \land f(\bx) = 1] 
			\end{align*}
where the first inequality is by \Cref{def:weak-term-learner} and the assumption that \alg~always outputs a $\gamma$-weak term, and the second is by construction of ${\cal D}_{i+1}$.
	So it follows that 
	\begin{align*}
		& \Prx_{\bx \sim \mathcal{D}}[h_{i+1}(\bx) = 0 \ | \ f(\bx) = 1]\\ 
		&= \Prx_{\bx \sim \mathcal{D}}[h_{i}(\bx) = 0 \ | \ f(\bx) = 1] - \Prx_{\bx \sim \mathcal{D}}[T_{i+1}(\bx) = 1 \land h_i(\bx) = 0 \ | \ f(\bx) = 1] \\		
		&= \left( 1  - \Prx_{\bx \sim \mathcal{D}}[T_{i+1}(\bx) = 1 \ | \ h_i(\bx) = 0 \land f(\bx) = 1] \right) \Prx_{\bx \sim \mathcal{D}}[h_i(\bx)=0 \ | \ f(\bx) =1] \\	
		&\leq \left( 1  - \frac{1}{2s} \right) \Prx_{\bx \sim \mathcal{D}}[h_i(\bx)=0 \ | \ f(\bx) =1]
	\end{align*}
			as desired, 
where the last inequality is by part ($ii$) of \Cref{def:weak-term-learner} and the construction of ${\cal D}_{i+1}$.
			\end{proof}

Recalling the condition that is checked in the while-loop in Step~2, an immediate consequence of \Cref{claim:lollipop} is that the \DNFLearn~algorithm halts after at most $m := O(s \log(1/\gamma))$ iterations of the while-loop. It remains to bound the error of the final hypothesis $h$. 
We first bound the probability of false negatives, which is easy: when the algorithm exits the while-loop in Step~2, we have 
		\[\Prx_{\bx \sim \mathcal{D}}[h_m(\bx)=0 \land f(\bx) = 1] \leq \gamma.\]
	To bound the probability of false positives, we use the following simple claim:
	
	\begin{claim} \label{claim:easy}
	For all $i$, we have
		\[\Prx_{\bx \sim \mathcal{D}} [T_i(\bx) = 1 \land f(\bx) =0] \leq O \left( \frac{\gamma}{s \log(1/\gamma)} \right).
		 \]
	\end{claim}
		\begin{proof}
We have		\begin{align*}
		\frac{\gamma}{s \log(1/\gamma)} &\geq \Prx_{\bx \sim \mathcal{D}_i} [T_i(\bx) = 1 \land f(\bx) =0] \\
		&= \frac{1}{2} \cdot \Prx_{\bx \sim \mathcal{D}} [T_{i}(\bx) = 1  \ | \  f(\bx) =0] \\
		&\geq \frac{1}{2} \cdot \Prx_{\bx \sim \mathcal{D}} [T_{i}(\bx) = 1 \land f(\bx) =0]
		\end{align*}
where the first line is by the assumption on $\alg$ and the second is by the definition of ${\cal D}_i$.
	\end{proof}

	Using the above claim and a union bound over $T_1,T_2,\dots,T_m$, we now have that
		\[\Prx_{\bx \sim \mathcal{D}} [h_m(\bx) = 1 \land f(\bx) = 0] \leq O(\gamma).\]
	Thus, it follows that $h$ has error $O(\gamma)$ as desired. 
This concludes the proof of \Cref{lem:term-booster}.
\end{proof}

\section{List-Decoding a Term in $(ns)^{O(\log (ns))}$ Time}
\label{sec:generating-new-terms}

Recall that  \Cref{sec:list-decoding-DNFs} gave a proof of \Cref{thm:list-decoding} with a slightly weaker quantitative runtime bound of
 ${\frac 1 p} \cdot (ns)^{O(\log(s) \log(n))}$. 
We briefly sketch how to achieve the ${\frac 1 p} \cdot n^{O(\log(ns))}$ running time claimed in \Cref{thm:list-decoding}. 
This is done using the \genlistofterms~algorithm which is described below. 

\begin{algorithm}
\addtolength\linewidth{-2em}

\vspace{0.5em}

\textbf{Input:} Query access to $f: \zo^n \rightarrow \zo$, $y \in f^{-1}(1)$ \\[0.25em]
\textbf{Output:} A list $\calL$ of terms

\

\genlistofterms($f, y$):
\begin{enumerate}
	\item $\calL \gets \emptyset$. 
	\item Let $M := (ns)^{O(\log s)}$ and let $y=\bY_0,\bY_1, ..., \bY_M$ denote a random walk starting from $y$.
	\item For $t \in [M]$ and $\ell \in [(ns)^{O(\log s)}]$:
		\begin{enumerate}
			\item Run $2 \log(n)$ independent random walks of length $\ell$ starting from $\bY_t$ to get $\bZ^1_t, \dots , \bZ^{2 \log(n)}_t$.
			\item Add the largest term $T$ that satisfies all of $\bY^{1}_t, \dots, \bY^{2 \log n}_t$ to $\calL$.
		\end{enumerate}
	\item Return $\calL$.
\end{enumerate}

\caption{Generating a list of terms using locally-mixing random walks.}
\label{alg:speedy-generate-term}
\end{algorithm}

We will rely on~\Cref{alg:speedy-generate-term} (as well as its performance guarantee thanks to~\Cref{thm:list-decoding}) in~\Cref{sec:learning-exact-dnf} to learn exact DNF; we defer the proof of~\Cref{thm:list-decoding} to~\Cref{sec:speedy-mixing}. 

\section{Learning Exact DNFs}
\label{sec:learning-exact-dnf}

We now turn to our second main result, which we recall below. 
Recall that
an \emph{exact-$k$ DNF} is a DNF formula in which every term has exactly $k$ distinct literals, and that an \emph{$s$-term exact-DNF} is an $s$-term DNF formula which is an exact-$k$ DNF for some value of $k$.

\begin{theorem} [Restatement of \Cref{thm:exact-k-learn-intro}]
\label{thm:exact-k-learn}
There is a distribution-free PAC+MQ algorithm which, given any $\eps > 0$, runs in time 
\[
\exp \left( \log^{O(1)}(ns) \right) \cdot \poly(\eps^{-1})
\] 
and with probability at least 9/10 learns any unknown $s$-term exact-DNF to accuracy $\eps$, using a hypothesis which is a DNF of size $O(s \log(1/\eps))$. 
\end{theorem}

As mentioned in the introduction, the $9/10$ success probability in~\Cref{thm:exact-k-learn} can be amplified to $1-\delta$ with an additional runtime factor of $O(\log(1/\delta))$ using standard techniques, see e.g.~\cite{HKL+:91}.

\paragraph{Preliminary Simplifications and Remarks.}  We assume that the value $k$ for which the target function is an exact-$k$ DNF is provided to the algorithm, as well as the number of terms $s$.  This assumption is without loss of generality since we can run interleaved executions of the algorithm with guessed values of $k=1,2,\dots$ and guessed values of $s=1,2,\dots$ (see \cite{HKL+:91} for the now-standard argument establishing this).

We remark that we make no attempt at optimizing the exponents in the various quasipolynomial factors that arise in our proof. 
Moreover, we assume throughout the rest of this section that 
\begin{equation}
\log^{O(1)}(ns) \leq k  \leq n - \log^{O(1)}(ns) \label{eq:covered}
\end{equation}
since, as is shown in \Cref{sec:small-or-large-k}, the case where $k \leq \log^{O(1)}(ns)$ is handled by a simple alternate algorithm, and the case where $k \geq n - \log^{O(1)}(ns)$ is handled by a simple reduction to the case where $k < n - \log^{O(1)}(ns)$.

Throughout this section, let $f\isazofunc$ be the target $s$-term exact-$k$ DNF that we are trying to learn under an unknown distribution $\calD$.
We will assume that there is an arbitrary but fixed realization of $f$ as an exact-$k$ DNF throughout this section; we will sometimes refer to the terms $T$ appearing in this realization as \emph{true} terms.

Finally, we also assume that both $n$ and $s$ are at least sufficiently large absolute constants, which is clearly without loss of generality for the purpose of proving \Cref{thm:exact-k-learn}.

\subsection{Expanding and Pruning a List of Terms}
\label{subsec:expand-and-prune}

All the terms we consider throughout this section (regardless of whether they are true terms appearing in $f$) will have size exactly $k$. 
We start by describing two simple procedures that we will rely on in our learning algorithm:
\begin{itemize}
	\item \prune: Given a candidate list of terms $\calL$, \prune$(f, \calL)$ removes terms in $\calL$ for which there is definitive evidence that the term is not in $f$.
	\item \expand: Given a candidate list of terms $\calL$, \expand$(\calL)$ outputs a list of terms containing all terms of length $k$ that are close to some term in $\calL$.
\end{itemize}

A precise description of $\prune$ is given in \Cref{alg:prune}.

\begin{algorithm}[h]
\addtolength\linewidth{-2em}

\vspace{0.5em}

\textbf{Input:} MQ access to $f: \zo^n \rightarrow \zo$, a list of terms $\calL$ \\[0.25em]

\textbf{Output:} A list of terms $\calL' \subseteq \calL$

\

\prune($f, \calL$):
\begin{enumerate}
	\item $\calL' \gets \calL$
		\item For each $T \in \calL$, if $|T| \neq k$ then update $\calL' \leftarrow \calL' \setminus \{T\}$.  If $|T| = k$, then repeat the following $100n$ times:
	\begin{enumerate}
		\item Sample a uniformly random $\bz$ satisfying $T(\bz) = 1$.  If $f(\bz) = 0$, then update $\calL' \leftarrow \calL' \setminus \{T\}$ and move on to the next $T \in {\cal L}.$

	\end{enumerate}
	\item Return $\calL'$.
\end{enumerate}

\caption{The~\prune~subroutine to remove implausible terms.} 
\label{alg:prune}
\end{algorithm}

It is easy to see that (w.h.p.) for every term $T \in \calL' := \prune(f, \calL)$, there exists a \emph{true} term~$T'\in f$ such that $\dterm(T, T') \leq O(\log(s))$:

\begin{lemma}
\label{lem:prune}
Fix a term $T \in \calL$ such that $\dterm(T, T') > 2\log(s)$
for all $T' \in f$. Then 
\[
	\Prx\sbra{T \in \prune(f, \calL)} \leq \frac{1}{10^n}
\]
where the probability is over the randomness of the algorithm~\prune. 
\end{lemma}

In other words, if a term $T \in \calL$ is not removed by $\prune(f, \calL)$, then $T$ is $O(\log(s))$-close to a true term in $f$ with very high probability. 

\begin{proof}
	Fix a true term $T$ and let
	$\calU_T$ denote the uniform distribution over the satisfying assignments of $T$. We will establish the following:
	\begin{equation} \label{eq:term-checker}
		\text{If}~\Prx_{\bz \sim \calU_T} [f(\bz) = 0] \leq \frac{1}{4},~\text{then for some true $T' \in f$,}~|T'\setminus T| \leq 2\log(s). 
	\end{equation}
	Note that this immediately implies the lemma, since~\prune~samples $100n$ points $\bz\sim\calU_T$, and $|T| = |T'| = k$ (thanks to Step 2(b) of~\Cref{alg:prune}), because of which we have $\dterm(T, T') = |T'\setminus T|$. 
	
	We now turn to proving~\Cref{eq:term-checker}; we will do so by establishing the contrapositive. Suppose that $|T'\setminus T| > 2\log(s)$ for each true term $T'\in f$. Note that 
	\[
		\Prx_{\bz \sim \calU_T} [T'(\bz) = 1] \leq 2^{-|T' \setminus T|} \leq \frac{1}{s^2}. 
	\]
	Taking a union bound over all $s$ true terms in $f$ then implies that
	\[
		\Prx_{\bz \sim \calU_T} [f(\bz) = 1] \leq \frac{1}{s} < \frac{3}{4},
	\]
	which proves~\Cref{eq:term-checker}.
\end{proof}	

As a consequence of~\Cref{lem:prune}, note that we have that $|\prune(f,{\cal L})| \leq s \cdot n^{2 \log s}$ except with inverse exponential failure probability. 

We now turn to the~\expand~subroutine, which is formally presented in \Cref{alg:expand}. By inspection of Step~2 of the~\expand~subroutine, we immediately have the following lemma:

\begin{algorithm}[h]
	\addtolength\linewidth{-2em}
	
	\vspace{0.5em}
	
	\textbf{Input:} A list of width-$k$ terms $\calL$ \\[0.25em]
	
	\textbf{Output:} A list of terms $\calL' \supseteq \calL$.
	
	\
	
	\expand($\calL$):
	\begin{enumerate}
		\item Set $\calL' \gets \emptyset$.
		\item For each $T \in \calL$, add every width-$k$ term $T'$ such that $\dterm(T, T') \leq 2\log^{1000}(ns)$ to $\calL'$.
		\item Return $\calL'$.
	\end{enumerate}
	
	\caption{The~\expand~subroutine to add ``nearby'' terms.}
	\label{alg:expand}
\end{algorithm}	

\begin{lemma}
\label{lem:expand-clusters}
	Let $\calL$ be a collection of width-$k$ terms. Suppose there exists a term $T \in \calL$ and a true term $T'\in f\setminus \calL$ such that $\dterm(T, T') \leq 2\log^{1000}(ns)$. Then $T' \in \expand(\calL) \setminus \calL$. 
\end{lemma}

\subsection{The Principal Challenge: Finding Far Points}
\label{subsec:far-points}

The~\prune~and~\expand~subroutines as well as the~\genlistofterms~procedure (from~\Cref{sec:generating-new-terms}) motivate a simple learning procedure which will later form the basis of our full learning algorithm. 
For this and future sections, it will be convenient to work with a slight modification of the~\genlistofterms~procedure: 

\begin{definition} \label{def:genterm}
	The procedure~\genterm~is identical to~\genlistofterms~(\Cref{alg:speedy-generate-term}) but with ``$\ell \in [(ns)^{O(\log (n) \log (s))}]$'' replacing ``$\ell \in [(ns)^{O(\log s)}]$'' in Step~3 of~\genlistofterms.\footnote{In particular, using~\genterm~instead of~\genlistofterms~makes the proof of~\Cref{lem:far-point-learn} cleaner. 
	} 
\end{definition}

\begin{algorithm}[h]
	\addtolength\linewidth{-2em}
	
	\vspace{0.5em}
	
	\textbf{Input:} MQ access to $f\isazofunc$, sample access to $\calD$, $\eps \in (0, 0.5]$ \\[0.25em]
	
	\textbf{Output:} An $\eps$-weak term.
	
	\
	
	\simple$(f, \calD, \eps)$:
	\begin{enumerate}
		\item Estimate $\Pr_{\by \sim \calD} [f(\by) = 1]$ to additive error $\eps/2$ and return $\emptyset$ if our estimate is $\leq \eps/2$.

		\item Set $\calL = \emptyset$ and $\ell = -1$. 
		\item While $|\calL| > \ell$:
			\begin{enumerate}
				\item Let $\ell = |\calL|$. 
				\item Sample $\by \sim \calD$ until a $\by$ such that $f(\by) = 1$ is obtained.
				
				\item $\calL \gets \calL \cup \prune(f,\genterm(f,\by))$.
				
				\item Set $\calL \leftarrow \calL \cup~\prune(f,\expand(\calL))$.
				
				\item If some $T\in\calL$ is an $\eps$-weak term, halt and return $T$. 
			\end{enumerate}
		\item Return FAIL. 
	\end{enumerate}
	
	\caption{A simple weak-term learning procedure.}
	\label{alg:simple-learning}
\end{algorithm}	

The simple learning procedure \simple~is given in \Cref{alg:simple-learning}.
The algorithm repeatedly (a) generates a long list of terms via random walks, and then (b) expands and prunes the list of terms $\calL$, for as long as the size of the list $\calL$ is increasing. 
Recall, however, that~\expand~only adds ``nearby'' terms to $\calL$. 
The difficult case therefore is that the algorithm reaches the following situation:
\[
	\text{The size of $\calL$ is no longer increasing, as every remaining term in $f$ is ``too far'' from every term in $\calL$}.	
\]
More formally, we make the following definition:

\begin{definition}[Far from $\calL$]
\label{def:far-from-L}
	Let $\calL$ be a set of width-$k$ terms and supposer that $\calR := f\setminus \calL$ is nonempty.  In other words, $\calR$ is the (non-empty) set of true terms that do not appear in $\calL$. 
	We say that $f$ is \emph{far from $\calL$} if for each term $T' \in \calR$, 
	\[
		\min_{T\in\calL} \dterm(T, T') \geq \log^{1000}(ns).
	\]
	We say that $f$ is \emph{close} to $\calL$ otherwise.
\end{definition}

If $\calL$ is close to $f$, then \simple~will add a new term to $\calL$ thanks to $\expand$.  So let us suppose that $\calL$ is far from $f$.   
In this case, one might hope for the next iteration of~\genterm~to output a list of terms, one of which is a true term that did not appear in $\calL$ in the previous iteration of the while loop.   
This brings us to our main technical challenge:
\[
	\text{When will \genterm{} output a \emph{new} true term of $f$?}
\]
We claim that it will suffice to run \genterm{} from a \emph{far point}, i.e. a satisfying assignment of $f$ that is sufficiently far from every satisfying assignment of $\calL$. Formally, we have the following:

\begin{definition}[Far points]
\label{def:far-point}
	Let $\calL$ be a list of terms viewed as a DNF. We say that a point $y\in\zo^n$ is a \emph{far point} w.r.t. $\calL$) if 
	\[
		f(y) = 1 
		\qquad\text{and}\qquad 
		\dass(y, \calL) \geq \log^3(ns).\footnote{See the discussion following~\Cref{rem:term-dist-for-same-size} for the definition of $\dass$.}
	\]
\end{definition} 

\medskip

The following example should provide some intuition for this definition vis-a-vis the definition of $f$ being far from a set of terms (\Cref{def:far-from-L}):  

\begin{example}
Even if $\calL$ contains only a single term $T$, it is possible for $f$ to be far from $\calL$ in the sense of~\Cref{def:far-from-L} but for a satisfying assignment of $f$ to be non-far from $\calL$ in the sense of~\Cref{def:far-point}. As an example, consider 
	\[
		T = \bigwedge_{i=1}^{\ceil{\frac{n}{2}}} x_i 
		\qquad\text{and}\qquad 
		f = \bigwedge_{i= \ceil{\frac{n}{2}} + 1}^n x_i.
	\]
We have $\dterm(T, f) = \Theta(n)$, but $T(1^n) = f(1^n) = 1$, so $\dass(1^n,\calL)=0.$ 
	(On the other hand, note that most satisfying assignments of $f$ are far points w.r.t. $\calL := \{T\}$; we will use this in~\Cref{sec:finding-far-points}.)
\end{example}

We now show that if $\calL$ is far from $f$ \emph{and} a satisfying assignment $y\in f^{-1}(1)$ is a far point w.r.t. $\calL$, then a random walk starting at $y$ will yield a new term with high probability:

\begin{lemma}
\label{lem:far-point-learn}
Let $\calL$ be a set of width-$k$ terms with $|\calL| \leq n^{O(\log(s))}$. 
	Suppose that $y \in \zo^n$ is a far point w.r.t.~$\calL$ and furthermore $f$ is far from $\mathcal{L}$. Thenwith probability $1 - o(1)$ \genterm($f,y$) contains a true term that does not appear in $\calL$.
\end{lemma}

\Cref{lem:far-point-learn} is an easy consequence of the following:

\begin{lemma}
	\label{lem:far-point-walk}
	Suppose 
	$\calL$ is a set of terms
	such that $f$ is far from $\calL$. If $y \in f^{-1}(1)$ is a far point w.r.t.~$\calL$, and $y := \bY_0, \ldots , \bY_t$ is a lazy random walk over the satisfying assignments of $f$, then for any true term $T$ in $\calL$, 
	\[
		\Prx_{}\sbra{T(\bY_t) = 1} \leq t\cdot\exp\pbra{-\Omega\pbra{\log^3(ns)}}. 
	\]
	As an immediate consequence, 
	\[
		\Prx_{}\sbra{T(\bY_t) = 0~\text{for all}~T\in\calL} \geq 1 - t|\calL|\cdot\exp\pbra{-\Omega\pbra{\log^3(ns)}}. 
	\]	
\end{lemma}
		
Note that since $|{\cal L}| \leq s \cdot n^{2 \log(s)}$ w.h.v.p. after running~\prune, this is a nontrivial bound. 
We first see how this yields~\Cref{lem:far-point-learn}:

\begin{proof}[Proof of~\Cref{lem:far-point-learn}]
	For convenience, set 
	\[
		\eps = \frac{1}{n} 
		\qquad\text{and}\qquad 
		p = (ns)^{-O(\log(s))}. 
	\]
	Consider the iteration of \genterm{} where $t$ is equal to the $(p, \eps)$-local mixing time of a random walk starting at $y$ in the graph defined by the satisfying assignments of $f$ with a cover corresponding to the set of terms. Note that such a local mixing time $t$ exists by \Cref{thm:local-mixing}, and moreover by \Cref{thm:local-mixing} we have that $t = (ns)^{O(\log s)}$. Moreover, let $E$ denote the event we condition on for local mixing and $T^\star$ denote the term of $f$ over which the walk is $\eps$-close to uniformly distributed.
	
We first note that $T^\star \notin \calL$. Indeed, if $T^\star \in \calL$, then as $\Pr[ E] \geq p = (ns)^{-O(\log(s)}$, we get that 
	\[
		\Pr\sbra{T^\star(\bY_t) = 1} \geq (1-\eps)p \geq (ns)^{-O(\log(s))}, 
	\]
	which would contradict~\Cref{lem:far-point-walk}.
	So we have $T^\star \notin \calL$. 

	Now, fix an index $\ell$ in Step 3 of the~\genterm~algorithm.\footnote{Recall from~\Cref{def:genterm} that $\ell \in [(ns)^{O(\log(s)\log(s))}]$.} 
	Let
	$\bZ_t^1, \ldots , \bZ_t^{2\log(n)}$ denote the end-points of the $2\log(n)$-many $t$-step random walks, and $E_j$ denote the event that $E$ holds for the $j^\text{th}$ walk. Note that 
	\[
		\Pr\sbra{T^\star(\bZ_t^j) = 1 \,|\, E_j } \geq 1 - \frac{1}{n}.  
	\]
	On the other hand, note that for any $i \in [n]$ without a corresponding literal in $T^\star$,
	\[
	\Pr\sbra{(\bZ_t^j)_i = 0~\text{for all}~j \,|\, E_1, \dots, E_{2 \log(n)}} \leq \left( \frac{1}{2} + {\frac 1 n}\right)^{2 \log(n)} < \frac{1}{n^{3/2}}, \]
		and an analogous statement holds with $(\bZ_t^j)_i = 1$. 
		
	Let $\bT$ be the term output in the $\ell^{\text{th}}$-iteration of Step~3 of~\genterm. Taking a union bound over all indices $i\in[n]$, we get 
	\[
	\Pr\sbra{\bT \setminus T^\star \not = \emptyset \,|\, E_1, \dots, E_{2 \log(n)}} \leq \frac{1}{\sqrt{n}}. 
	\]
		We now conclude that
		\begin{align*}
			\Pr\sbra{\bT \setminus T^\star = \emptyset \land T^\star \subseteq \bT \,|\, E_1, ..., E_{2 \log(n)}} &\geq 1 - \Pr\sbra{\bT \setminus T^\star \not = \emptyset\,|\, E_1, \dots, E_{2 \log(n)}} \\& \qquad\qquad -  \Pr\sbra{\exists~j~\text{s.t.}~T^\star(\bZ_t^j) = 0 \,|\, E_1, \dots , E_{2 \log(n)} } \\
			&\geq 1 - o(1).
		\end{align*}
		Thus the probability we output $T^\star$ in the $\ell^{\text{th}}$ iteration is
			\[\Pr[\bT = T^\star] \geq (1 - o(1)) (ns)^{-O(\log(s) \log(n))}.\]
		The lemma now follows as \genlistofterms{} iterates $(ns)^{-O(\log(s) \log(n))}$ times for each value of $t$.
\end{proof}	

We now turn to the proof of~\Cref{lem:far-point-walk}:

\begin{proof}[Proof of~\Cref{lem:far-point-walk}]
	Let $\mathcal{R}$ denote the set of terms in $f$ that do not appear in $\calL$.
	We will work in a slightly more adversarial setting: At each time step, an adversary observes $\mathbf{Y}_i$ and chooses a term $T^\star \in \mathcal{R}$ that satisfies $\bY_i$. We then generate $\bY_{i+1}$ by choosing a random coordinate $\bj$ such that neither $x_{\bj}$ nor $\ol{x}_{\bj}$ appear in $T^\star$ and randomly flipping the $j$th coordinate of $\bY_{i}$ with probability $1/2$. (This corresponds to taking a step in a lazy random walk.) 
	
	We will now show that under any strategy of the adversary, for any fixed true term $T \in \calL$,
	\[
	\Pr\sbra{T(\bY_t) = 1} \leq O(t) \exp\pbra{-\Omega\pbra{\log^3(ns)}}.\]
	The result will then follow by a union bound over the (at most $s$) true terms in $\calL$. 
	Fix $j \in [t]$ and let $E_j$ denote the event that $j$ is the last time at which $\dist(\bY_j, T) \geq \log^3(ns)$. We will bound $\Pr[E_j \land T(\bY_t) = 1]$ and then do a union bound over all possibilities for $j$.
	
	For convenience, let $A_i$ denote the $(j + i)^\text{th}$ term that the adversary plays, and define the random variables 
	\[
		\bZ_i = \dass(\bY_{j+i}, T).
	\]
	We will proceed by a coupling argument. Let $\bX_i$ denote independent random variables defined as\footnote{Note that here we are using the assumption that $k \leq n - n - \log^{O(1)}(ns)$, cf.~\Cref{eq:covered}.}
	\[
		\bX_i = 
		\begin{cases}
			-1 & \text{w.p.}~\frac{\log^{3}(ns)}{n-k}, \\
			+1 & \text{w.p.}~\frac{0.5 \log^{1000}(ns)}{n-k}, \\
			0 & \text{otherwise}.
		\end{cases}
	\]

		Since $f$ is far from $\calL$, it follows that $|T \setminus A_i| \geq \log^{1000}(ns)$. Note that any variable in $|T \setminus A_i|$ that is fixed by $A_i$, must be set to the wrong value in $\bY_{i+j}$. Thus, there can be at most $\log^{3}(ns)$ such variables. As such, we can safely conclude that there are at least $\frac{3}{4} \log^{1000}(ns)$ variables in $T$ that are not fixed by $A_i$. It now follows that
			\[\Pr[\bZ_{i+1} - \bZ_i = -1] \leq \Pr[\bX_{i+1} = -1] \qquad\text{and}\qquad 
			\Pr[\bZ_{i+1} - \bZ_i = 1] \geq \Pr[\bX_{i+1} = 1].\]
		Thus, we can produce a coupling of $\bZ_{i+1} - \bZ_i$ and $\bX_i$ where $\bZ_{i+1} - \bZ_i \geq \bX_{i+1}$. Thus,
			\begin{align}
				\Pr[\bZ_{t-j} - \bZ_0 \leq -\log^{3}(ns) ] &\leq \Pr \left [\sum_{i = 0}^{t-j-1} \bX_i \leq -\log^{3}(ns) \right ].  \label{eq:longshot}
			\end{align}
			
To upper bound the RHS of \Cref{eq:longshot}, condition on the number of $i$'s for which $\bX_i$ is either $-1$ or $+1$; call this value $m$.  We must have $m \geq \log^3(ns)$ or else it is impossible to have  $\sum_{i=0}^{t-j-1} \bX_i \leq -\log^3(ns)$ so suppose $m \geq \log^3(ns)$.  Under this conditioning each of the $m$ nonzero $\bX_i$'s is $-1$ with probability ${\frac {\log^3 (ns)}{0.5 \log^{1000}(ns) + \log^3(ns)}} < {\frac 2 {\log^{997}(ns)}}$, so the probability that at least $m/2$ of them are $-1$ (which needs to happen for $\sum_i \bX_i$ to be negative) is at most 
\[
{m \choose {m/2}} \cdot \pbra{ {\frac 2 {\log^{997}(ns)}}}^{m} \leq 
\pbra{ {\frac 4 {\log^{997}(ns)}}}^{m} \leq
\exp \left( - \Omega \left( \log^3(ns) \right) \right);
\]
so we have
\[
\Pr[\bZ_{t-j} - \bZ_0 \leq -\log^{3}(ns) ] \leq 
\exp \left( - \Omega \left( \log^3(ns) \right) \right).
\]			
Taking a union bound over all indices $j$ then gives the result for $T$. A union bound over all $T \in \calL$ completes the proof.
\end{proof}

To summarize, we have shown that assuming $f$ is far from $\calL$, running~\genterm~from $y\inf^{-1}(1)$ where $y$ is a far point w.r.t $\calL$ yields a new true term $T^\star$ that does not belong to $\calL$. 
The bulk of the remainder of this section will give an algorithm to find far points; we describe this in~\Cref{sec:finding-far-points} but state the main technical guarantee below:

\begin{lemma}
\label{lem:far-point-found}
	There exists an algorithm, \findfar~(\Cref{alg:gen-far}), with the following guarantee: Suppose $\calL$ is a collection of width-$k$ terms with $|\calL| \leq n^{O(\log(s))}$. Suppose $f$ is far from $\calL$ and $y\in\zo^n$ satisfies a term $T^\star \in f\setminus\calL$. Then with probability $1 - o(1)$ $\findfar(f, y, \calL)$ outputs a far point w.r.t. $\calL$.
	Furthermore, $\findfar(f, y, \calL)$ runs in $\exp\pbra{\log^{O(1)}(ns)}$-time.
\end{lemma}

With~\Cref{lem:far-point-found} in hand, we can now prove~\Cref{thm:exact-k-learn}. 

\subsection{Learning Exact DNF: Proof of~\Cref{thm:exact-k-learn}}
\label{subsec:proof-of-exact-k}

We now give a simple modification of~\simple~(\Cref{alg:simple-learning}) which incorporates the~\findfar~procedure and will be used to prove~\Cref{thm:exact-k-learn}. The complete algorithm,~\exactlearn,~appears below as~\Cref{alg:exact-k-learn}.

\begin{algorithm}
\addtolength\linewidth{-2em}

\vspace{0.5em}

\textbf{Input:} MQ access to $f: \zo^{n} \rightarrow \zo$, sample access to $\calD$,  $\eps \in (0,1/2]$ \\[0.25em]
\textbf{Output:} An $\eps$-weak term

\

\exactlearn($f, \calD, \eps$):
\begin{enumerate}
	\item Estimate $\Pr_{\by \sim \calD} [f(\by) = 1]$ to additive error $\eps/2$ and return $\emptyset$ if the estimate is $\leq \eps/2$.
	\item Set $\calL \gets \emptyset$ and $i \gets 0$.
	\item For $i \in \sbra{\frac{\Theta(1)}{\eps}\cdot n^{O(\log(s))}}$:
		\begin{enumerate}
				\item Sample $\by \sim \calD$ until a $\by$ such that $f(\by) = 1$ is obtained.
			\item $\bz \gets \findfar(f, \by, \calL)$.
			\item $\calL \gets \calL \cup \prune(f,\genterm(f,\bz))$.
			\item $\calL \gets \calL \cup \prune(f,\expand(\calL))$.
			\item $i \gets i+1$. 
			\item If some $T\in\calL$ is an $\eps$-weak term, halt and return $T$. 
			\end{enumerate}
	\item Return FAIL. 
\end{enumerate}

\caption{A weak-term learner for exact-$k$ DNF.}
\label{alg:exact-k-learn}
\end{algorithm}

\begin{proof}[Proof of \Cref{thm:exact-k-learn}]

First, note that each iteration of the for-loop runs in quasi-polynomial time, and consequently so does~\Cref{alg:exact-k-learn}. 

Thanks to \Cref{lem:term-booster}, it suffices to show that 
\[
	\Pr\sbra{\exactlearn{}~\text{outputs FAIL}} \leq 
{\frac {c}{s \log(1/\epsilon)}}
\]
for a sufficiently small absolute constant $c>0$.
Thanks to~\Cref{lem:prune}, we can assume (at the additive expense of an inverse-exponentially small amount of failure probability) that $\prune{}$ ``succeeds,'' i.e. $\prune{}$ never accepts any terms with distance $\Omega(\log(s))$. We thus know that $|\calL| \leq n^{O(\log(s))}$ throughout the run of the algorithm (see the remark following the proof of~\Cref{lem:prune}).

Consider some iteration of the for-loop.  We consider two cases:
\begin{itemize}
	\item If $f$ is close to $\calL$, then Step~3(d) will grow the size of $\calL$ thanks to \Cref{lem:expand-clusters}. 
	\item If $f$ is far from $\calL$, then with probability at least $0.01\eps$, $\bz$ satisfies a term $T^\star \in f\notin \calL$ and so Step~3(c) will grow the size of $\calL$ thanks to~\Cref{lem:far-point-learn,lem:far-point-found}. 
\end{itemize}
Thus, each iteration increases the size of $\calL$ with probability at least $0.01\epsilon$. (Note that this is independent for each iteration.)

For notational convenience, let $\bI_j$ denote the indicator that $\calL$ increased in size during the $j^\text{th}$ iteration of the for-loop. 
From the above discussion, we have that for each $j$, regardless of the outcome of iterations $1,\dots,j-1$ of the for-loop, we have $\Pr[\bI_j=1] \geq 0.01\eps.$ 
Moreover, since $|\calL| \leq n^{O(\log(s))}$, we have that \exactlearn{}~can output FAIL only if $\sum_{j=1}^{\frac{\Theta(1)}{\eps} n^{\Omega(\log(s))}} \bI_j \leq n^{O(\log(s))}.$
Consequently, we have
 have
\begin{align*}
&\Pr\sbra{\exactlearn{}~\text{outputs FAIL} \,\bigg|\, \prune{} \text{ never fails} }\\
&=	\Pr \sbra{\text{the for-loop repeats for all~}\frac{\Theta(1)}{\eps} n^{\Omega(\log(s))} \text{~possible iterations} \,\bigg|\, \prune{} \text{ never fails} }\\
 &\leq 
	\Pr \left[\sum_{j=1}^{\frac{\Theta(1)}{\eps} n^{\Omega(\log(s))}} \bI_j \leq n^{O(\log(s))} \,\bigg|\,\prune{} \text{ never fails} \right] \\
	& \ll  \frac{1}{\omega(s  \log(1/\eps))}
\end{align*}
where the second inequality is by Azuma's inequality.
This concludes the proof of \Cref{thm:exact-k-learn}.
\end{proof}

\begin{remark}
We quickly describe why we believe that our approach should yield an exact learning algorithm. Throughout the algorithm, we maintain a list of candidate terms $\calL$ and label points according to the hypothesis $g = \bigvee_{T \in \calL} T$. When given a false positive (i.e. $g(y) = 1$ but $f(y) = 0$), we can prune our list $\calL$. If we are given a false negative (i.e. $g(y) = 0$ and $f(y) =1$), then we attempt to expand our list with this positive example as in the PAC+MQ model in lines $3(b), 3(c),$ and $3(d)$ of \exactlearn{}.
\end{remark}


\subsection{Finding Far Points: Proof of~\Cref{lem:far-point-found}}
\label{sec:finding-far-points}

Suppose $\calL$ is a list that contains some but not all of the true terms of the DNF $f$, and suppose $f$ is far from $\calL$. (Recall from the discussion in~\Cref{subsec:far-points} that if $f$ is close to $\calL$ then~\expand~will output a new true term.) 
Then note that there must exist far points w.r.t.~$\calL$: we can simply choose a random assignment of any ``far'' term $T^\star \in f\setminus \calL$. However, we clearly cannot implement this since the term $T^\star$ is unknown. 

In this section, we will prove~\Cref{lem:far-point-found} by giving a brute-force procedure which will find far points. In more detail, we describe an algorithm~\findfar~(\Cref{alg:gen-far}) which, given a point $y \in \zo^n$ satisfying an unknown true term $T^\star \in \calR := f\setminus \calL$, outputs a point $z$ which is (a) continues to satisfy $T^\star$, and (b) is a far point w.r.t. $\calL$. 

We give a brief overview of how~\findfar~allows us to achieve this guarantee. Let $y \in \zo^n$ be a satisfying assignment of $f$ which satisfies some term $T^\star \in \calR$. If $y$ is a far point w.r.t.~$\calL$ (note that this is trivial to check), then we simply halt and return $y$. If not, then there must exist terms in $T\in \calL$ such that $\dass(y, T) \leq \log^{3}(ns)$; call such terms \emph{bad} terms. 
In this case we can then make the following ``win-win'' argument: Let $ i \in [n]$ be a coordinate such that the literal corresponding to $y_i$ appears in at least $1\%$ of the bad terms. (We will get to the issue of the existence of such ``popular'' coordinates shortly.)
\begin{enumerate}
	\item If either $x_i$ or $\ol{x}_i \in T^\star$, then we have made progress towards learning the term $T^\star$. 	(Note that if we explicitly knew the term $T^\star$ then we could efficiently find a far point $z$ that satisfies $T^\star$.)
	\item If not, then we can jump from $y$ to  $y^{\oplus i}$ while increasing the distance (measured by $\dass$) to at least $1\%$ of the bad terms. 
\end{enumerate}
 Of course, it is unclear at the moment how to obtain ``popular'' coordinates; this will be an easy consequence of a natural noising procedure which we next describe. 

\subsubsection{Noise and Popular Coordinates}

We first define the notion of a ``popular'' coordinate from the overview above: 

\begin{definition}[Popular coordinates]
\label{def:popular-coords}
	Let $\mathcal{L}$ be a collection of width-$k$ terms and $y \in \zo^n$. 
	We say that a coordinate $i \in [n]$ is \emph{$(y, \calL)$-popular} if the literal corresponding to $y_i$ appears in at least $0.005|\calL|$ terms of $\calL$. We say that it is \emph{$(y, \calL)$-super popular} if it appears in at least $0.01|\calL|$ terms of $\calL$. 
\end{definition}

It is easy to see that it may be the case that there are \emph{no} $(y, \calL)$-popular coordinates: 

\begin{example}
	Suppose $\calL = \{x_1\wedge x_2 \wedge x_3, x_4 \wedge x_5 \wedge x_6  \}$, and let $y \in \zo^n$ be such that $y_1 = y_2 = \dots = y_6 = 0$. Then there are no $(y, \calL)$-popular coordinates. 
\end{example}

In this example, though, applying a small amount of noise to $y$ will yield a new point $z$ such that w.h.p.~(a) $z$ will continue to satisfy the same term as $y$, and (b) we will (essentially) have many $(z,\calL)$-\emph{super} popular coordinates. Below we will show that applying noise is also useful in the general case.

We now turn to describe the~\leap~procedure; it will in fact be useful for us to pass a set of restricted coordinates $F\sse[n]$ as an argument to the procedure. Looking ahead, this set $F$ will be the set of coordinates we track in our recursive procedure towards either learning $T^\star$ or increasing the distance to $\calL$ as alluded to in the algorithm overview above.

\begin{algorithm}[h]
\addtolength\linewidth{-2em}

\vspace{0.5em}

\textbf{Input:} Query access to $f: \zo^n \rightarrow \zo$, $y \in f^{-1}(1)$, a set of ``frozen'' variables $F \subseteq [n]$ with $|F| \leq k - 1000 \log^{300}(ns)$ \\[0.25em]
\textbf{Output:} A point $\bz \in f^{-1}(1)$

\

\leap($f, y, F$):
\begin{enumerate}
	\item Set $\bz \gets y$. 
	\item For $i \in \sbra{\exp\pbra{(\log(ns))^{O(1)}}}$:
		\begin{enumerate}
			\item For each $i \in [n] \setminus F$: Independently set $\bz_i$ to be equal to $\ol{y}_i$ with probability $\frac{\log^{300}(ns)}{k - |F|}$ and $y_i$ otherwise.
			\item For each $i \in F$, set $\bz_i = y_i$.
			\item If $f(\bz) = 1$, return $\bz$.
		\end{enumerate}
	\item Return FAIL. 

\end{enumerate}

\caption{Applying noise to ensure popularity.}
\label{alg:warmup-estimator}
\end{algorithm}

We formally have the following guarantee on~\leap:

\begin{lemma}
\label{lem:leap-frog-guarantee}
Suppose $\calL$ is a list of width-$k$ terms with $|\calL| = n^{O(\log(s))}$ and let $F \sse [n]$ be a set of coordinates with $|F| \leq k - 1000\log^{300}(ns)$. 
Suppose $y$ is \emph{not} a far point and satisfies a term $T^\star \in f$ with 
\[
	|T^\star|_{\ol{F}}| \leq 2(k-|F|). 
\]
Then with probability at least $\exp\pbra{-(\log(ns))^{O(1)}}$, \leap($f,y,F$) outputs a point $\bz$ such that: 
\begin{enumerate}
	\item[(i)] $T^\star(\bz) = 1$, 
	\item[(ii)] $\dass(\bz, T) \geq \log^{150}(s)$ for all $T \in \calL$ such that $| T|_{\overline{F}} \setminus T^\star|_{\overline{F}} | \geq \frac{k - |F|}{\log(ns)}$, and 
	\item[(iii)] $\bz_F=y_F$. 
\end{enumerate}
\end{lemma}

\begin{proof}
Item~(iii) of the lemma is clear from the description of the algorithm. Note that we do not flip any variable in $T^\star$ (during Steps~2 (a) and (b) of~\leap) with probability 
	\[\left (1 - \frac{\log^{300}(ns)}{k - |F|} \right)^{2(k- |F|)} \geq \exp\pbra{-(\log(ns))^{O(1)}},\]
so item (i) holds with the claimed probability.  We now condition on (i) holding, i.e.~we condition on the random assignment $\bz$ satisfying $T^\star$.
	Fix a $T \in \calL$ such that 
	\[
		\dterm(T|_{\overline{F}}, T^\star|_{\overline{F}}) \geq \frac{k - |F|}{\log(ns)}. 
	\]
	Note that for each index $i$ such that either $x_i$ or $\ol{x}_i$ appears in $(T \setminus T^\star)|_{\overline{F}}$, there is at least a $(k-|F|)^{-1}\log^{300}(ns)$ probability that it is set so as to falsify the corresponding literal in $T$. A Chernoff bound then implies that the probability of there being at most $\log^{150}(s)$ coordinates of disagreement between $\bz$ and $T$ is at most $n^{-\log^2(s)}$. A union bound then yields item (ii).
\end{proof}

To summarize, after applying~\leap~to obtain $\bz$ from $y$, the satisfying assignment $\bz$ will be close to only those terms $T \in \calL$ that were ``close'' to $T^\star|_{\overline{F}}$ (in the sense that $|T|_{\overline{F}} \setminus T^\star|_{\overline{F}}|$ is small). 
Note that such ``close'' terms will look similar to $T^\star|_{\ol{F}}$ (which $\bz$ satisfies), and so we can guarantee the existence of many super-popular coordinates. More formally, we have the following: 

\begin{lemma}
\label{lem:super-pop}
Let $z \in \zo^n$ satisfy a term $T^\star$ of size $a$. Let $\calL$ be a set of terms such that  every $T\in \calL$ satisfies (a) $|T|\geq 0.99a$, and (b) $|T \setminus T^\star| \leq 2a/\log(ns)$. 
Then there exist at least $a/100$ many $(z,\calL)$-super popular coordinates corresponding to indices of literals in $T^\star$.
\end{lemma}

\begin{proof}
The proof is by double counting $(\ell,T)$ pairs where $T \in \calL$ and $\ell$ is a literal in $T^\star \cap T$. Indeed suppose there are at most $a/100$ literals that appear in at least $|\calL|/100$ terms, then the number of pairs is at most $a|\calL|/50$. 

	On the other hand, any term $T \in \calL$ by assumption satisfies $|T \setminus T^\star| \leq \frac{2a}{\log(ns)}$. So it must have at least $0.99 a - {2a}{\log^{-1}(ns)}$ literals in common with $T^\star$. It then follow that the number of $(\ell, T)$ pairs is at least 
		\[ 
			\left(0.99 a - \frac{2a}{\log(ns)} \right) |\calL| > \frac{1}{2} a|\calL| > a|\calL|/50,
		\]
		a contradiction. 
\end{proof}

\subsubsection{A Brute-Force Procedure for Finding Far Points}
\label{subsec:genfar}

We now describe a recursive procedure~\findfar~(\Cref{alg:gen-far}) for finding a far point w.r.t.~$\calL$. 
The algorithm~\findfar~takes as input an MQ oracle to $f$, a satisfying assignment $y \in f^{-1}(1)$, a list of width-$k$ terms $\calL$, and three sets $\calW \sse \calL$, $S \sse[n]$, and $A \sse [n]$ s.t. $S\cap A = \emptyset$. Before giving a description of these sets, we note that when~\findfar~is called in Step~3(b) of~\exactlearn~(\Cref{alg:exact-k-learn}), we mean 
\[
	\findfar(f, \by, \calL) := \findfar(f, \by, \calL, \calL, \emptyset, \emptyset), 
\]
i.e.~we call $\findfar(f, \by, \calL, \calW, S, A)$ with $\calW = \calL$ and $S = A = \emptyset$. 

Next, we describe the semantics of the sets $\calW$, $S$, and $A$. Suppose $y$ satisfies a term $T^\star \in f \setminus \calL$. 
\begin{itemize}
	\item {The set $\calW \sse \calL$} is the set of \emph{bad} terms w.r.t.~$y$, i.e. terms in $\calL$ that the point $y$ is close to (in the sense of $\dass$). Note that if $\calW = \emptyset$, then $y$ will be a far point w.r.t.~$\calL$. (It may be helpful to think of $\calW$ as the set of terms which we will be \underline{\textbf{w}}orried about; we will simply ignore the terms in $\calL\setminus \calW$ for much of the arguments below.)
	\item {The set $S \sse [n]$}, together with the point $y$, describes a subset of the literals of the unknown term $T^\star$ that $y$ \underline{\textbf{s}}atisfies. Note that $|S| \leq k$ since $|T^\star| = k$; thus, if $|S|$ is sufficiently large, we can brute force enumerate and guess $T^\star$. (This is Step~1 of the algorithm). 
	\item {The set $A$} will be a set of $(y, \calW)$-popular coordinates that \emph{do not} appear in $T^\star$, i.e.~are \underline{\textbf{a}}bsent from $T^\star.$ Because of this, we can flip these co-ordinates in $y$ to (a) increase the $\dass$-distance to $\calW$, while (b) ensuring that we continue to satisfy the term $T^\star$. 
\end{itemize}
Note that if there are no $(y, \calW)$-popular coordinates, then in Step~5 we call~\leap~to jump to a new point which will w.h.p.~have many popular (in fact, super popular) coordinates. 

\medskip
 
\noindent \textbf{The Recursion Tree.}~We will view the algorithm as recursively building a tree (\Cref{fig:cowbell}); in particular, the first call $\findfar(f, y, \calL, \calL, \emptyset, \emptyset)$ will be the root of the tree with sub-calls defining its children. 
It may be the case that the satisfying assignment $y \in f^{-1}(1)$ satisfies multiple terms $T^\star \in f\setminus \calL$; in this case, we will show that there exists some root-to-leaf path in the recursion tree that either outputs FAIL or a far point satisfying some term $T^\star$ that was also satisfied by $y$. Looking ahead, we will call this path the ``main branch'' of the recursion tree. 

The nodes of the tree will carry a label $(y, \calW)$; the leaf nodes of the tree either return FAIL or are labelled by far points.\footnote{Note that the only possible behaviours of the algorithm are that it outputs FAIL, or it outputs a far point; see Steps~1(a), 4, and 7 of~\Cref{alg:gen-far}.}
The tree will have two flavors of nodes: 
\begin{itemize}
	\item \emph{Unhappy} nodes, i.e. nodes where there exist \emph{no} $(y,\calW)$-popular coordinates. In this case, we run~\leap~up to a quasi-polynomial number of times to obtain a new point $\bz$ that will (w.h.p.) have many $(\bz,\calW)$-super popular coordinates. We then recurse and call~$\findfar(f,\bz,\calL,\calW',S,A)$ where $\calW'$ is the set of terms that may be close to $\bz$. This is~Step 5 of the algorithm. 
	\item \emph{Happy} nodes, i.e. nodes where there exist $(y, \calW)$-popular coordinates. Such nodes have two possible outgoing branches:
	\begin{enumerate}
		\item $S$-branches, when~\findfar~is called in Step~7 of the algorithm. This happens when the $(y, \calW)$-popular coordinate is added to the set $S$, i.e.~when we believe that the literal $\ell$ corresponding to $y_i$ appears in $T^\star$. Note that Step~8 of the algorithm ensures that if this recursive call of~\findfar~fails, then indeed $\ell \in T^\star$. 
		\item $A$-branches, when~\findfar~is called in Step~9 of the algorithm. This happens when $i \in [n]$ is $(y,\calW)$-popular \emph{and} the literal corresponding to $y_i$ does not appear in $T^{\star}$. 
	\end{enumerate}
\end{itemize}


\begin{figure}[t]

\centering

\begin{tikzpicture}

	\draw (0,0) to node[midway,fill=white]{$i \in A$} (-2,-3);
	\draw (0,0) to node[midway,fill=white]{$i \in S$} (2,-3);
	\draw (-2, -3) to node[midway,fill=white]{$j \in S$} (0, -6);
	\draw (-2, -3) to node[midway,fill=white]{$j \in A$} (-4, -6);
	\draw (-2,-10.1) -- (0,-7.1);
	\draw (-1,-10.1) -- (0,-7.1);
	\draw (0,-10.1) -- (0,-7.1);
	\draw (1,-10.1) -- (0,-7.1);
	\draw (2,-10.1) -- (0,-7.1);
	
	\node (1) at (-4, -6.25) {$\vdots$};
	\node (2) at (-0, -6.25) {$\vdots$};
	\node (2) at (2, -3.6) {$\vdots$};

	\filldraw[color=white, draw=black] (0,0) circle (0.4);
	\node (y) at (0,0) {$\tiny y$};
	
	\filldraw[color=white, draw=black] (2,-3) circle (0.4);
	\node (yS) at (2,-3) {$\tiny y$};
	
	\filldraw[color=white, draw=black] (-2,-3) circle (0.4);
	\node (yA) at (-2,-3) {$y^{\oplus i}$};
	
	\filldraw[color=white, draw=black] (0,-7.1) circle (0.4);
	\node (y') at (0,-7.1) {$y'$};
	
	\filldraw[color=white, draw=black] (-2,-10.1) circle (0.4);
	\node (z) at (-2,-10.1) {$\bz$};
	
	\node (i-pop) at (2.5, 0) { $i$ is $(y,\calW)$-popular~\raisebox{-0.25em}{\LARGE\smiley{}}};

	\node (j-pop) at (-4.75, -3) { $j$ is $(y^{\oplus i},\calW)$-popular~\raisebox{-0.25em}{\LARGE\smiley{}}};

	\node (j-pop) at (4.5, -3) {\phantom{$j$ is $(y,\calW)$-popular~\raisebox{-0.25em}{\LARGE\smiley{}}}}; 
	
	\node (no-pop) at (3.75, -7.1) { No $(y', \calW')$-popular coordinates~\raisebox{-0.25em}{\LARGE\frownie{}}};
	
	\draw[-latex] (-1.25, -8.2)  to[out=-45,in=225] (1.25, -8.2);
	
	\node (count) at (2.75, -8.2) {\small $\exp\pbra{\log^{O(1)}(ns)}$};

	\node (noise) at (-2, -8.2) {\small $\leap{}$};

	\node (aaa) at (-2, -10.75) {$\vdots$};
 		
\end{tikzpicture}
\caption{An illustration of the recursion tree of~\findfar~(\Cref{alg:gen-far}). For simplicity, the node labels only indicate the satisfying assignment $y$ and omit the set $\calW$.}	
\label{fig:cowbell}

\end{figure}
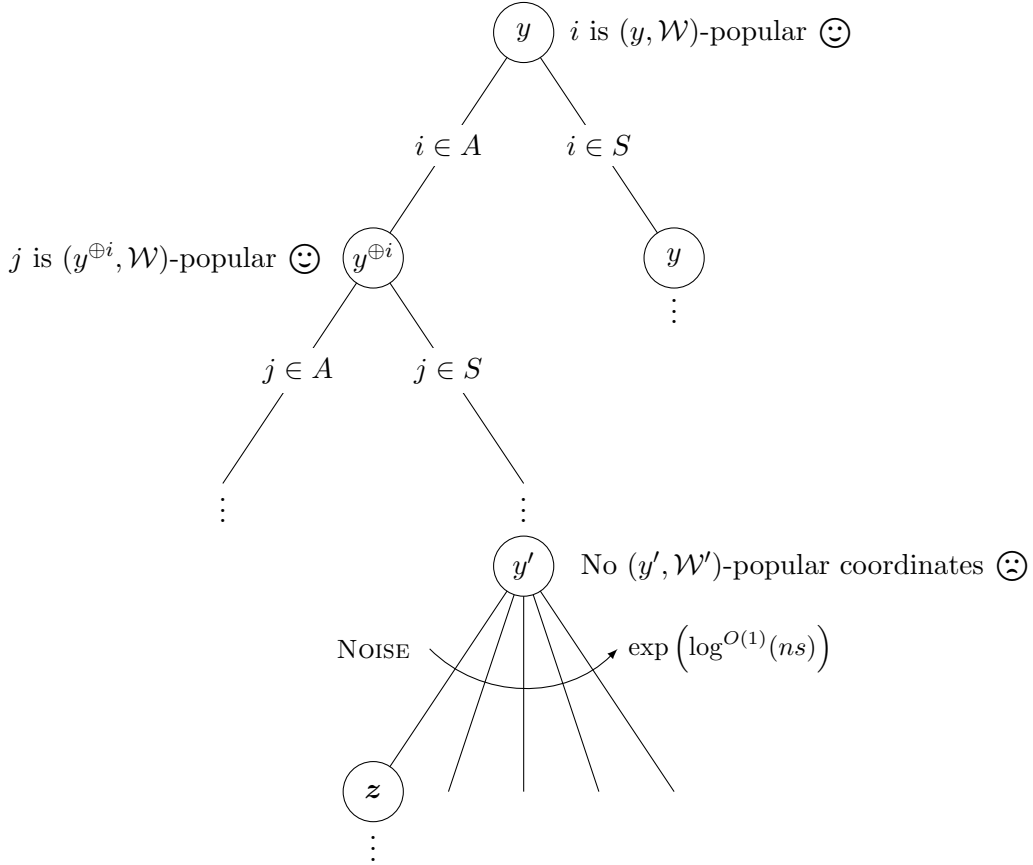

\begin{algorithm}[h!]
\addtolength\linewidth{-2em}

\vspace{0.5em}

\textbf{Input:} MQ access to $f: \zo^n \rightarrow \zo$, a point $y \in f^{-1}(1)$, a list of terms $\calL$, a set $\calW \subseteq \calL$, a set $S \subseteq [n]$, a set $A \subseteq [n]$ disjoint from $S$ \\[0.25em]
\textbf{Output:} A far point $\by$~w.r.t.~$\calL$

\

\findfar($f, y, \calL, \calW, S, A$):
\begin{enumerate}
	\item If $|S| \geq k - \log^{1000}(s)$: 
	\begin{enumerate}
		\item For each width-$k$ term $T$ s.t. $\dterm(T, T_{y}(S)) \leq \log^{1000}(s)$: Starting from $y$, randomly flip the coordinates not fixed by $T$ to get $\bz$. If $\bz$ is a far point with respect to $\calL$, return $\bz$.
		\item Return FAIL otherwise.
	\end{enumerate}
	\item Remove any term $T$ from $\calW$ such that $\dass(y, T) \geq  2\log^{100}(ns)$. 
	\item Remove any term $T$ from $\calW$ such that there exist {more than} $\log^3(s)$ indices $i \in A$ such that the literal corresponding to $\overline{y}_i$ appears in $T$. 
	\item If $\calW$ is empty, return $y$ if it is a far point with respect to $\calL$. Otherwise, return FAIL.
	\item Set $m = 0$. If there does not exist a popular coordinate in $[n] \setminus A \sqcup S$ with respect to $y$ for $\calW$:
		\begin{enumerate}
			\item  While $m \leq \exp\pbra{(\log(ns))^{O(1)}}$:
				\begin{enumerate}
					\item Set $\bz \gets \leap(f,y,S \sqcup A)$. 
					\item If $\bz =$  FAIL, set $m \leftarrow m + 1$ and go to Step~5(a).
					\item $\calW' \gets \calW$. 
					\item Remove any term $T$ from $\calW'$ for which $\dass(\bz, T) \geq 2\log^{100}(ns)$.
					\item If there exists fewer than $0.01 \cdot (k-|S|)$ coordinates that are $(\bz,\calW')$-super popular in $[n] \setminus (A \sqcup S)$, set $m \leftarrow m + 1$ and go to Step~5(a).
					\item Return \findfar$(f,\bz,\calL, \calW', S, A)$ if it doesn't output FAIL; if it outputs FAIL, set $m \leftarrow m + 1$ and go to Step~5(a).
				\end{enumerate}
				\item Return FAIL. 
		\end{enumerate}
	\item Let $i \in [n] \setminus (A \sqcup S)$ 
	be a popular coordinate for $\calW$.
	\item Run \findfar$(f,y,\calL, \calW, S \cup \{i\}, A)$. If it returns a far point w.r.t. $\calL$, halt and return that point. 
	\item If $f(y^{\oplus i}) = 0$, return FAIL. 
	\item Return \findfar$(f,y^{\oplus i},\calL, \calW, S, A \cup \{i\})$. 
\end{enumerate}

\caption{Procedure to compute a far point.}
\label{alg:gen-far}
\end{algorithm}

We will understand the algorithm by analyzing the underlying recursion tree it induces. In particular the first call, i.e. $\findfar(f,y,\calL, \calL, \emptyset, \emptyset)$, is the root. This will make subcalls defining its children.  

\subsubsection*{Runtime Analysis}

We begin by bounding the runtime of the algorithm. Towards this, we will require some simple lemmas. The first and simplest lemma bounds the size of $S$:

\begin{lemma} \label{lem:S-bound}
Any branch of the recursion tree of \Cref{alg:gen-far} has $|S| \leq k.$
\end{lemma}
\begin{proof}
Since $|S|$ only increases in Step~7, this is immediate from inspection of Step~1.
\end{proof}

Next we bound the size of $A$:

\begin{lemma}
\label{lem:cover-bound}
Any branch of the recursion tree of \Cref{alg:gen-far} has $|A| \leq O(\log^3(ns) \log(|\calL|))$.
\end{lemma}

\begin{proof}
Fix a branch $B$ in the tree of recursive executions of the procedure. Consider the node in that tree that adds the $i^\text{th}$ element of $A$ to $A$, and let $y^{(i)}$, $A^i$ and $\calW^i$ be the string $y$, set $A$, and set $\calW$ of bad terms that the node was created with. It is clear from inspecting~\Cref{alg:gen-far} that $|\calW^i|$ never increases with $i$.

	\begin{claim} \label{claim:chobani-flip}
	We have
	\[
		|\calW^{i + 10000\log^3(ns)}| \leq \frac{1}{2} |\calW^i|.
	\]	
	\end{claim}

Since $|{\cal W}|$ can only be halved $\log(|\calL|) + 1$ times before $\calW$ becomes empty (and once $\calW$ becomes empty we are at a leaf in the recursion tree), \Cref{lem:cover-bound} follows. 

\begin{proof}[Proof of~\Cref{claim:chobani-flip}]
Fix an $i$. Towards a contradiction, suppose that $|\calW^{i + 10000\log^3(ns)}| > \frac{1}{2} |\calW^i|.$	
Define the set $P := A^{i + 10000\log^3(ns)} \setminus A^i$. Moreover, consider the potential function
	\[
		\Phi_i(\ell) =\sum_{T \in \calW^i } \min(\log^3(ns), \dass(y^{(i+\ell)}|_{A^{i+\ell}}, T|_{A^{i+\ell}})). 
	\]
Note that for all $\ell$, we have that $A^{i+\ell} \subseteq A^{i+\ell+1}$ since we never remove coordinates from $A$. Moreover, we have that for all $j$, $y^{i+\ell+j}|_{A^{i+\ell}} = y^{i+\ell}|_{A^{i+\ell}}$ as coordinates in $A$ are never toggled/flipped after they are added to $A$. It then follows that $\Phi_i$ is an increasing function in $\ell$.

On the other hand, since each index is popular when added to $A$, it must be the case that for each $j \in P$, when $j \in P$ is added to $A^{i+\ell}$, the literal corresponding to $y^{i+\ell}$ appears in at least a $\frac{1}{200}$ fraction of terms in $\calW^{i+ \ell}$. Thus, the potential $\Phi$ must increase by at least $\frac{1}{200} |\calW^{i + \ell}|$.
By assumption, $|\calW^{i+\ell}|$ is at least $\frac{1}{2} |\calW^i|$ for all $\ell \leq 10000 \log^3(ns)$. It then follows that 
	\[\Phi_i(10000\log^3(ns)) \geq 10000\log^3(ns) \cdot {\frac{1}{400}} |\calW^i| > \log^3(ns) |\calW^i|. \]
But by construction, the potential is at most $\log^3(ns) |\calW^i|$, so we have reached the desired contradiction. 
\end{proof}

This completes the proof of \Cref{lem:cover-bound}.
\end{proof}

\begin{lemma}
\label{lem:branch-bound}
	Suppose that $|\calL| \leq n^{O(\log(s))}$, then any branch of \Cref{alg:gen-far} reaches Step 5(a).vi of the algorithm at most $O(\log(|\calL|) + \log(k))$ times. 
\end{lemma}

\begin{proof}
	Fix a branch $B$ in the recursion tree. Suppose $(y^i, \calW^i)$ is the label of the node that reaches Step 5(a).vi of the algorithm for the $i^\text{th}$-time in $B$. Similarly, let $A^i$ and $S^i$ denote sets $A$ and $S$ at the child of the node that calls \leap{} for the $i^\text{th}$ time in $B$. We will show the following claim:
	
	\begin{claim} \label{claim:Bonn}
		Either $|\calW^{i+1}| \leq \left (1 - \frac{1}{1000} \right) |\calW^i|$ or $(k - |S^{i+1}|) \leq \left( 1 - \frac{1}{1000} \right) (k-|S^{i}|)$
	\end{claim}
	
Since this can only happen $O(\log(|\calL|) + \log(k))$ times before either $|S^i| \geq k - \log^{1000}(s)$ or $|\calW^i| < 1$ (either of which yields a leaf node in the recursion tree), the lemma follows. 
	
	\begin{proof}
	Let $P$ be the set of $(y^i, \calW^i)$-super popular coordinates after reaching Step 5(a).vi of the algorithm for the $i^\text{th}$ time on branch $B$. Since we reached Step 5(a).vi of the algorithm, we have $|P| \geq 0.01 (k - |S^i|)$ from Step 5(a).v.
		
	 Now there must be no $(y^{i+1}, \calW^{i+1})$-popular coordinates when we reach Step 5(a).vi for the $(i+1)^\text{th}$ time. So either (a) every coordinate in $P$ is no longer $(y^{i+1}, \calW^{i+1})$-popular, or (b) $P \subseteq S^{i+1} \sqcup A^{i+1}$. In the former case, as each coordinate in $P$ was $(y^i, \calW^i)$-super popular, by definition of super-popularity we must have that $|\calW^{i+1}| \leq (1 - \frac{1}{200}) |\calW^i|$.
	 In the latter case, we have  
	 	\[
	 		k - |S^{i+1}| \leq~ k - (|S^i| + |P| - \log^{50}(ns)) {\leq} 0.99 \cdot (k - |S^{i}|) + \log^{50}(s) \leq \left (1 - \frac{1}{1000} \right) \cdot (k - |S^{i}|)
	 	\]
	 where we used \Cref{lem:cover-bound} for the first inequality. In more detail, note that $P \cap S^i = \emptyset$, and since (i) $S^{i}\sse S^{i+1}$ and (ii) $|P| \leq |S^{i+1}| + |A^{i+1}|$, we get 
	 \[
	 	|S^{i+1}| \geq |S| + |P| - |A^{i+1}|
	 \]
	 yielding the first inequality together with~\Cref{lem:cover-bound}. 
	 This completes the proof of \Cref{claim:Bonn}.
	\end{proof}
	This completes the proof of~\Cref{lem:branch-bound}. 
\end{proof}

With this, we can now bound the runtime of the algorithm:

\begin{lemma}
Suppose that $f$ is an exact $k$-DNF and $\calL$ is a list of $n^{O(\log(s))}$ terms of size $k$, then \Cref{alg:gen-far} runs in quasi-polynomial time.
\end{lemma}

\begin{proof}
Note that since each node of the recursion tree runs in quasi-polynomial time, it suffices to bound the size of the tree. Let $r(a, d,\ell)$ denote the maximum number of nodes in a tree with two types of nodes (``happy'' and ``unhappy'') with the following properties: 
\begin{enumerate}
	\item Any root-to-leaf path has at most $\ell$ unhappy nodes. (Recall that an unhappy node corresponds to a call to~\leap, i.e. to executing Step~5 of the algorithm.)
	\item Each unhappy node has at most $E := \exp\pbra{(\log(ns))^{O(1)}}$ children. 
	\item Each happy node has at most $2$ children. (Recall that such nodes correspond to calls of~\findfar~that reach Step~6 of the algorithm, and the child nodes corresponds to calls of~\findfar~in either Step~7 or Step~9 of the algorithm.)
	\item Any root-to-leaf path has at most $a$ happy {left} children. (Left branches will correspond to $A$-branches of the recursion tree, cf.~Step 9 of the algorithm.)
	\item Any root-to-leaf path has at most $d$ happy right children.  (Right branches will correspond to $S$-branches of the recursion tree, cf.~Step 7 of the algorithm.)
\end{enumerate}
We will prove the following claim:
	
	\begin{claim} \label{claim:master-theorem}
	We have 	
	\[r(a, d,\ell) \leq 10 (a + d + 1)^{a + 1} E^{\ell}.\]
	\end{claim}
	
	\begin{proof}
		We prove the statement by induction on $a+d+\ell$. For the base case, note the the statement is true when $a + d + \ell = 1$. Moreover, it is easy to see that it is true when $a = 0$ or $d = 0$. We now assume that $a,d > 0$ and $\ell \geq 0$ are such that statement holds for all $a', d', \ell'$ such that $a' + d' + \ell' < a + d + \ell$.
		
		If the root is happy, then we have that 
		\begin{align*}
			r(a,d,\ell) &\leq r(a-1,d,\ell) + r(a,d-1,\ell) \\
			&\leq 10 (a + d)^a E^{\ell} + 10 (a + d)^{a+1} E^{\ell} \\
			&\leq 10 E^\ell \left(1 + (a+d) + \dots + (a + d)^a + (a + d)^{a+1} \right ) 
			\\
			&\leq 10 E^\ell \frac{(a+d)^{a+2}}{a+d-1} \\
			&\leq 10 E^\ell (a+d +1)^{a+1} \frac{(a+d)^{3}}{(a+d+1)^2 (a+d-1)} \\
			&\leq 10 E^\ell (a+d +1)^{a+1} 
		\end{align*}
	
		where in the final line we use the fact that $x^3/((x+1)^2(x-1)) \leq 1$ when $x \geq 2$.
		
		If the root is unhappy, then 
			\[r(a,d,\ell) \leq E \cdot r(a,d,\ell - 1) \leq 10 (a+d+1)^{a+1} E^\ell \]
		as desired. This completes the inductive step.
	\end{proof}

	\sloppy By \Cref{lem:S-bound}, \Cref{lem:cover-bound} and \Cref{lem:branch-bound}, it follows that the size of the recursion tree is bounded by $r(O(\log(|\calL|)\log^3(ns)), k, O(\log(|\calL|) + \log(k)))$. By \Cref{claim:master-theorem}, this is quasi-polynomial as desired.
\end{proof}

\subsubsection*{Correctness}

We now turn to prove that the algorithm actually outputs a far point. To do so, we will show that a particular branch will succeed in doing so with high probability. Throughout this section, let $T^\star$ be a term satisfied by $y \in \zo^n$ where $y$ is the argument to the initial call of \findfar. 
 
\begin{definition}[Main branch]
The \emph{main branch} of the tree is the branch which correctly places each index $i$ into the sets $A$ and $S$, i.e. it adds $i$ to $A$ if the literal corresponding to $y_i$ is not in $T^\star$, and into $S$ if the literal corresponding to $y_i$ is in $T^\star$. If $\leap$ is called, the main branch follows along the first recursive call that satisfies $T^\star$. If no such recursive call is made, then we say the main branch {has been ``cut''}.
\end{definition}

Note that the main branch being ``uncut'' implies there exists a branch with a leaf labelled by a far point that satisfies $T^\star$. 

\begin{lemma}
\label{lem:uncut}
Suppose $|\calL| = n^{O(\log(s))}$. Then with probability at least $1-1/n^9$, the main branch is uncut.	
\end{lemma}

\begin{proof}
Note that in order for the main branch to be cut, it must be the case that some $\leap$ call along it must fail to generate a point $z$ that satisfies $T^\star$ and has many $(z,\calW')$-super popular coordinates. Fixing a particular node on the main branch, we note that 
	\[ \bigg| T^\star|_{\overline{A \sqcup S}} \bigg| = k - |S|  \leq 2 (k - |S| - |A|)\]
since $|S| \geq k - \log^{1000}(ns)$ and $|A| \leq \log^{50}(ns)$ by \Cref{lem:cover-bound}. It then follows by \Cref{lem:leap-frog-guarantee} that with probability at least $1 - \frac{1}{n^{11}}$, one of the \leap{} calls at this node outputs a $\bz$ with 
\begin{enumerate}
	\item[(i)] $T^\star(\bz) = 1$, and 
	\item[(ii)] $\dass(\bz, T) \geq \log^{150}(ns)$ for any $T \in \calL$ with
	\[
		|T|_{\overline{A \sqcup S}} \setminus T^\star|_{\overline{A \sqcup S}}| \geq \frac{|T^\star|_{\overline{A \sqcup S}} |}{\log(ns)}.
	\]
\end{enumerate}

Now note that every term $T \in \calW'$ must satisfy
	\[| T_{\overline{A \sqcup S}} | \geq k - |S| - |A| \geq 0.99 (k - |S|) = 0.99 | T^\star|_{\overline{A \sqcup S}} |. \]
It then follows by \Cref{lem:super-pop} that there are at least $0.01 (k - |S|)$ super popular coordinates with respect to $\calW'$ and $\bz$ in $[n] \setminus (A \sqcup S)$.

Thus, any particular node that calls \leap{} does not cut the main branch w.h.p. By \Cref{lem:branch-bound}, the main branch has length at most $\wt{O}(k)$. Performing a union bound over every node in the branch then gives the result.
\end{proof}

\begin{lemma}
\label{lem:far-forever}
Suppose the points given by the labels of the nodes on the main branch are $\by^1, \dots , \by^m$ and that $|\calL| \leq n^{O(\log(s))}$. With probability $1 - o(1)$,
for each term $T\in\calL$, if there exists an $i$ for which $\dass(\by^i, T) \geq \log^{100}(ns)$, then $\dass(\by^m, T)\geq \log^3(ns)$. 
\end{lemma}

Note that~\Cref{lem:far-forever} justifies Step~2 of the algorithm, ensuring that if you are sufficiently far from a term \emph{not} in $\calW$ at some node, then you remain far at child nodes. 

\begin{proof}
Fix a term $T \in \calL$ and index $i \in [m]$. Suppose that $\dass(\by^i, T) \geq \log^{100}(ns)$. Denote by $\bD^i \subseteq [n]$ the set of bits that must be flipped for $\by^i$ to satisfy $T$, and define $\bD^m$ analogously for $\by^m$. Moreover, we define $A^i, S^i, A^m$, and $S^m$ 
to be the sets $A$ and $S$ at the $i^{\text{th}}$ and $m^{\text{th}}$ nodes respectively. 

\begin{claim} \label{claim:coconut-lentil-curry}
	We have 
	\[\Pr[|\bD^i \setminus \bD^m| \geq |\bD^i|/2
	] \leq \exp(-\log^3(ns))\]\
\end{claim}

\begin{proof}
To see this, consider how any $j \in \bD^i$ is removed. On the one hand, it could be toggled and added to the set $A$ at some point. By \Cref{lem:cover-bound}, this only accounts for at most $O(\log^5(ns))$ indices. Thus the remainder of the bits must be flipped during one of the \leap{} calls. By \Cref{lem:branch-bound}, there are at most $O(\log^2(ns))$ such calls. Since $|S| \leq k - \log^{1000}(ns)$ and $|A| \leq O(\log^5(ns))$ (by \Cref{lem:cover-bound}) for each call, it follows that any fixed coordinate $j \in \bD^i$ is flipped by some \leap{} call after $\by^i$ with probability at most $\smash{O \left( \log^{-600}(ns) \right)}$. The claim now follows, with room to spare, by a Chernoff bound since $|\bD^i| \geq \log^{100}(ns).$
\end{proof}

Recalling that $|{\cal L}| \leq n^{O(\log s)}$ and $m \leq \tilde{O}(k)$, applying \Cref{claim:coconut-lentil-curry} together with a union bound over each $T \in \calL$ and $i \in [m]$ proves the lemma.
\end{proof}

We now have everything we need to prove the correctness of \findfar~(\Cref{alg:gen-far}).

\begin{proof}[Proof of~\Cref{lem:far-point-found}]
Note that it suffies to show that the main branch outputs a far point with high probability, so we restrict our attention to the main branch. If we return because $|S| \geq k - \log^{1000}(ns)$, then we correctly enumerate term $T^\star$ in Step 1(a) of the algorithm. It then follows that since $f$ is far from $\calL$, randomizing the coordinates outside of $T^\star$ will yield a far point with extremely high probability.

So now assume this is not the case. Since the branch is uncut with high probability by \Cref{lem:uncut}, Step~5(b) never returns FAIL. Moreover, since we always follow the branch according to $T^\star$, we note that line $8$ also never returns FAIL since we only add $i$ to $A$ when $i \not \in T^{\star}$. 

Thanks to the above, note that a leaf node in our recursion tree can only arise from Step~4 of the algorithm. It remains to show that Step $4$ does not return FAIL w.h.p. Let $\by^t$ denote the string in the $t^\text{th}$ node of the main branch, $A^t$ the set $A$ at this node, and $\by^m$ the final string in the main branch. Fix a term $T \in \calL$. 
\begin{itemize}
	\item If $T$ was removed from $\calW$ because it was covered by $A^t$, then $\dass(\by^m, T) \geq \log^{3}(ns)$ as the coordinates in $A^t$ are fixed.
	\item On the other hand, if $T$ was removed from $\calW$ because for some $t$, $\dass(\by^t, T) \geq 2\log^{100}(ns)$ (as in Step~2 of the algorithm), then $\by^m$ is $\log^3(ns)$ far from satisfying $T$ with high probability by \Cref{lem:far-forever}.
\end{itemize}
Thus, $\by^m$ is far from these terms as well. We conclude that $\by^m$ will indeed be a far point as desired.
\end{proof}

\section*{Acknowledgements}
Josh Alman is supported in part by NSF Grant CCF-2238221. Shivam Nadimpalli is supported by NSF grants CCF-2106429, CCF-2211238, CCF-1763970, and CCF-2107187. Shyamal Patel is supported by NSF grants CCF-2106429, CCF-2107187, CCF-2218677, ONR grant ONR-13533312, and an NSF Graduate Student Fellowship. Rocco Servedio is supported by NSF grants CCF-2106429 and CCF-2211238. This work was partially completed while a subset of the authors was visiting the Simons Institute for the Theory of Computing.

\begin{flushleft}
\bibliographystyle{alpha}
\bibliography{allrefs}
\end{flushleft}

\appendix

\section{Corner Cases for \Cref{thm:exact-k-learn}:  Learning Size-$s$ Exact-$k$ DNFs When $k$ is Very Small or Very Large} \label{sec:small-or-large-k}

It is easy to learn size-$s$ exact-$k$ DNFs efficiently, using DNFs of size $O(s \log(1/\eps))$ as hypotheses, if $k$ is very small (even without using membership queries). The following lemma covers the $k \leq (\log^{O(1)}(ns))$ case of \Cref{thm:exact-k-learn}:

\begin{lemma} [Small $k$] \label{lem:small-k}
The class of all size-$s$ exact-$k$ DNFs can be distribution-free PAC learned to error $\eps$ in time $\poly(n^k,1/\eps)$, using only random examples and no membership queries, using the hypothesis class of size-$(s \log(2/\eps))$ exact-$k$ DNF formulas.
\end{lemma}

\begin{proof}
We recall the following result due to Haussler (see also Section~2.3 of \cite{KearnsVazirani:94}):

\begin{theorem} [Attribute-efficiently learning disjunctions using disjunction hypotheses \cite{Haussler88}] \label{thm:haussler88}
There is a distribution-free PAC learning algorithm that learns the class of size-$s$ disjunctions over $\zo^N$ in time $\poly(N,1/\eps,\log(1/\delta))$, using hypotheses which are disjunctions of size $s \cdot \log(2/\eps)$.
\end{theorem}

In our context, we can view a size-$s$ exact-$k$ DNF as a size-$s$ disjunction over an expanded feature space of $N={n \choose k} \cdot 2^k =O(n^k)$ ``meta-variables'' corresponding to all possible conjunctions of exactly $k$ literals over $x_1,\dots,x_n$.  Applying \Cref{thm:haussler88}, we get \Cref{lem:small-k}.
\end{proof}

For the case of \Cref{thm:exact-k-learn} when $k$ is very large (at least $n-(\log(ns))^{O(1)}$), 
a simple approach in our context is to introduce $n$ new ``dummy'' variables $x_{n+1},\dots,x_{2n}$ and view the unknown DNF formula $f(x_1,\dots,x_n)$ as a DNF formula over the $2n$ variables $x_1,\dots,x_{2n}$ which happens to only depend on the $n$ variables $x_1,\dots,x_n$.  
Define the distribution ${\cal D}'$ over $\zo^{2n}$ by taking ${\cal D}'(x_1,\dots,x_{2n})={\frac 1 {2^n}} {\cal D}(x_1,\dots,x_n)$ for every $x \in \{0,1\}^n$.  
We can simulate drawing a $2n$-bit random example from ${\cal D}'$ by simply drawing a random $n$-bit example from ${\cal D}$ and appending a uniform random $n$-bit string for the final $n$ coordinates.
And we can trivially simulate membership queries to $f$ (viewed as a $2n$-bit function) using the membership oracle to $f$, since the last $n$ bits are guaranteed to be irrelevant for the target DNF.
Since $k \leq n$, and since we may assume that $n \leq 2n - \log^{O(1)}(2ns)$ (as otherwise we can run in time $2^{\poly(n)}$ within the allowed $\exp((\log(ns))^{O(1)})$ time bound and solve the DNF learning problem in a trivial way),
this reduces our exact-DNF learning problem to the setting which is covered by \Cref{eq:covered}.

Finally, let $h(x_1,\dots,x_{2n})$ be the exact-$k$ DNF hypothesis which is provided by \Cref{thm:exact-k-learn}; taking the accuracy parameter in \Cref{thm:exact-k-learn} to be $\eps/100$, we have $\Pr_{\bx \sim {\cal D}'}[h(\bx) \neq f(\bx)] \leq \eps/100$.
By Markov's inequality, if we choose a uniform random assignment $(\bb_{n+1},\dots,\bb_{2n})$ fixing the last $n$ coordinates and we define $h': \zo^n \to \zo$ to be $h'(x_1,\dots,x_n)=h(x_1,\dots,x_n,\bb_{n+1},\dots,\bb_n)$, with probability 99/100 we have that $\Pr_{\bx \sim {\cal D}}[h'(\bx) \neq f(\bx)] \leq \eps$, thus with high probability achieving an $\eps$-accurate hypothesis under the original distribution ${\cal D}$ over the original domain, as required.


\section{List-Decoding a Term in $(ns)^{O(\log (ns))}$ Time: Proof of \Cref{thm:list-decoding}}
\label{sec:speedy-mixing}

Recall that  \Cref{sec:list-decoding-DNFs} gave a proof of \Cref{thm:list-decoding} with a slightly weaker quantitative runtime bound of
 ${\frac 1 p} \cdot (ns)^{O(\log(s) \log(n))}$.
 We now show how to achieve the ${\frac 1 p} \cdot n^{O(\log(ns))}$ running time claimed in \Cref{thm:list-decoding}.
 This is done using the \genlistofterms~algorithm which we recall below for convenience.

\begin{algorithm}
\addtolength\linewidth{-2em}

\vspace{0.5em}

\textbf{Input:} Query access to $f: \zo^n \rightarrow \zo$, $y \in f^{-1}(1)$ \\[0.25em]
\textbf{Output:} A list $\calL$ of terms

\

\genlistofterms($f, y$):
\begin{enumerate}
	\item $\calL \gets \emptyset$. 
	\item Let $M := (ns)^{O(\log s)}$ and let $y=\bY_0,\bY_1, ..., \bY_M$ denote a random walk starting from $y$.
	\item For $t \in [M]$ and $\ell \in [(ns)^{O(\log s)}]$:
		\begin{enumerate}
			\item Run $2 \log(n)$ independent random walks of length $\ell$ starting from $\bY_t$ to get $\bZ^1_t, \dots , \bZ^{2 \log(n)}_t$.
			\item Add the largest term $T$ that satisfies all of $\bY^{1}_t, \dots, \bY^{2 \log n}_t$ to $\calL$.
		\end{enumerate}
	\item Return $\calL$.
\end{enumerate}

\caption{(Restatement of \Cref{alg:speedy-generate-term}) Generating a list of terms using Locally Mixing Random Walks}
\label{alg:recall-speedy-generate-term}
\end{algorithm}

It is clear that \genlistofterms~runs in time $(ns)^{O(\log s )}$.
Note that \genlistofterms~differs from the algorithm described in \Cref{sec:list-decoding-DNFs} in that it runs an initial random walk and uses the points along this walk as candidate starting points for the ``second-stage'' random walks. Using this modification, we will prove the following:

\begin{lemma}
\label{lem:speedy-generate-term-analysis}
Given a DNF $f$ and a point $y \in \{0,1\}^n$, with probability at least $n^{-O(\log(ns))}$, the procedure $\genlistofterms(f,y)$ outputs a list of terms $\calL$ such that $T \in \calL$ for some term $T \in f$.
\end{lemma}

\Cref{lem:speedy-generate-term-analysis} is easily seen to yield \Cref{thm:list-decoding}, by simply performing $n^{O(\log(ns))}$ runs of \genlistofterms; the total runtime is 
\[(ns)^{O(\log s)} \cdot n^{O(\log(ns))}
=(ns)^{O(\log s)} \cdot (ns)^{O(\log n)}
=(ns)^{O(\log(ns))}.
\] 

To give some intuition for the proof of \Cref{lem:speedy-generate-term-analysis}, recall that since the $p$-parameter of \Cref{thm:local-mixing} is $p = \left( \frac{s \dmax  \log(|V|) \log(1/\eps)}{\theta} \right)^{-\Omega(\log(s))}$, \Cref{thm:local-mixing} only guarantees a near-uniform sample from some term $T_i$ with probability $(ns)^{-O(\log(s))}$. Thus, to get $2\log(n)$ samples we need to repeat a total of $(ns)^{-O(\log(s) \log(n))}$ times. If we could improve the success probability of a walk to $(ns)^{-O(1)}$, then we would immediately get the desired speed-up. While we do not know how to do this, it turns out that the following weaker result will suffice:

\begin{lemma}
\label{lem:good-point-reached}
Let $G = (V,E)$ be a graph and $\mathcal{C} := X_1, ..., X_s$ be a $\theta$-cover of $G$. For any $v \in V$, with probability at least $1/8$ a random walk starting from $v$ will reach a point $w \in V$ with $T_{\mix}^{(s\dmax)^{-O(1)}, \eps}(w, \mathcal{C}) \leq (\frac{s \dmax \log(|V|) \log(1/\eps)}{\theta})^{O(\log s )}$ in at most $(\frac{s \dmax  \log(|V|) \log(1/\eps)}{\theta})^{O(\log s )}$ steps.
\end{lemma}

This was essentially proved along the way to \Cref{thm:local-mixing}, but we will need to rephrase a few of the lemmata to put them in the form that we need. We begin by noting that it suffices to prove a disjoint-cover version of
\Cref{lem:good-point-reached}, specifically the following:

\begin{lemma}
\label{lem:good-point-reached-disjoint}
Let $G = (V,E)$ be a graph and $\mathcal{C} := X_1, ..., X_s$ be a disjoint $\theta$-cover of $G$. For any $v \in V$, with probability at least $1/8$, a random walk starting from $v$ will reach a point $u \in V$ with $T_\mix^{(s\dmax )^{-O(1)}, \eps}(u, \mathcal{C}) \leq (\frac{s \dmax  \log(|V|) \log(1/\eps)}{\theta})^{O(\log(s))}$ in at most $(\frac{s \dmax  \log(|V|) \log(1/\eps)}{\theta})^{O(\log(s))}$ steps.
\end{lemma}

\begin{proof}[Proof of \Cref{lem:good-point-reached} Assuming \Cref{lem:good-point-reached-disjoint}]
We will closely follow the proof of \Cref{thm:local-mixing} (assuming \Cref{thm:disjoint-local-mixing}	).

Let $(H, \mathcal{C}')$ be the disjointification of $(G,\mathcal{C})$. Fix a vertex $v \in V(G)$ and a copy $u \in V(H)$ of $v$. We start by recalling that 
\begin{equation} \label{eq:see-me-now}
T_{\mix}^{p/s, \eps} (v, \mathcal{C}) \leq T_{\mix}^{p, \eps/2s}(u, \mathcal{C}'),
\end{equation}
	which was shown in the proof of \Cref{thm:local-mixing} (\Cref{eq:see-you-later}). 
	
	Now fix an integer $t$ and let $v := \bY_0, ..., \bY_t$ denote a random walk starting from $v$ in $G$. Similarly, let $u := \bZ_0, \dots, \bZ_t$ denote a random walk starting from $u$ in $H$. 
	Recall that there is a natural coupling between walks in $G$ starting from $v$ and walks in $H$ starting from $u$, namely, for each $j \in \{1, ..., t\}$, choose a random number from $[s]$, which corresponds to the copy of $G$ that is moved to in $H$ (if a self-loop is traversed, we ignore the coin flip and stay at the same copy.) As described in the proof of \Cref{thm:local-mixing}, this gives us a coupling between the walk $\bY_1, \dots, \bY_t$ in $G$ and the walk $\bZ_1, \dots, \bZ_t$ in $H$. 
	
	The vital property of the above coupling is that for each $i$, the vertex $\bZ_i$ in $H$ corresponds to (a copy of) the vertex $\bY_i$ in $G$. In particular, combining this coupling with \Cref{lem:good-point-reached-disjoint}, it follows that  with probability at least $1/8$, a random walk from $v \in V(G)$ reaches a vertex $a \in V(G)$ with a corresponding vertex $b \in V(H)$ satisfying
	\[ T_\mix^{(s\dmax (H))^{-O(1)}, \eps}(b, \mathcal{C}) \leq \left(\frac{s \dmax (H) \log(|V(H)|) \log(1/\eps)}{\theta} \right)^{O(\log(s))} \]
	in at most $(\frac{s \dmax  \log(|V(H)|) \log(1/\eps)}{\theta})^{O(\log s)}$ steps. 
Using the relation between local mixing times in $H$ and $G$ that is given by \Cref{eq:see-me-now}, it then follows that
\begin{align*}
T_\mix^{(s\dmax (G))^{-O(1)}, \eps}(a, \mathcal{C}) 
&\leq \left(\frac{s^2 \dmax (G) \log(|V(H)|) \log(2s/\eps)}{\theta} \right)^{O(\log s)} \\
&= \left(\frac{s \dmax (G) \log(|V|) \log(1/\eps)}{\theta} \right)^{O(\log s)}
\end{align*}
	as desired. 
\end{proof}

It remains to prove \Cref{lem:good-point-reached-disjoint}. First we will need the following result.

\begin{lemma}
\label{lem:good-mixing-rephrased}
	Suppose that $G = (V,E)$ is a graph and $\mathcal{C} := \{A_1, \dots, A_s\}$ is a disjoint $\theta$-cover. Moreover, let $\{B_1, ..,. B_s\}$ be an $\alpha$-super cover, such that $B_1, \dots, B_s$ are all $(\Delta, g)$-good with respect to $\calC$. We then have that for each $v \in V$, a random walk starting at $v$ and going for at most $\wt{O}(s^2 \dmax ^2/ g^2 \alpha^2)$ steps with probability at least $1/8$
	reaches a vertex $u$  such that 
		\[T_{\mix}^{1/O(s\dmax ), 4 \eps s}(u,\calC) \leq \Delta. \]
\end{lemma}

\begin{proof}
This was shown in the proof of \Cref{lem:good-mixing}.
\end{proof}

We can now establish \Cref{lem:good-point-reached-disjoint}:

\begin{proof}[Proof of \Cref{lem:good-point-reached-disjoint}]
Applying \Cref{lem:good-cover-exists} to $(G, \calC)$, we get an $\alpha$-super cover $\{B_1,\dots,B_s\}$ with $\alpha = \lambda^{O(\log(s))}$,  where each $B_i$ is the $A_I$ given by \Cref{lem:good-cover-exists}, such that $B_1, \dots, B_s$ are all $(\lambda^{-O(\log s)}, \lambda^{O(\log s)})$-good with respect to $\calC$.  Combining this with \Cref{lem:good-mixing-rephrased}, we get that the random walk in $G$ from $v$ reaches a point $w$ with
		\[T_{\mix}^{1/O(s \dmax ), 4 \eps s}(w,\calC') \leq \left( \frac{\dmax  s \log(|V|) \log(1/\eps)}{\theta} \right)^{O(\log(s))} \]
	in at most $\left(\frac{\dmax  s \log(|V|) \log(1/\eps)}{\theta} \right)^{O(\log(s))}$ steps with probability at least $1/8$.
\end{proof}

Finally, it remains to prove \Cref{lem:speedy-generate-term-analysis} using \Cref{lem:good-point-reached}:

\begin{proof}[Proof of \Cref{lem:speedy-generate-term-analysis} using \Cref{lem:good-point-reached}]
As described in \Cref{sec:list-decoding-DNFs}, we create a graph on the satisfying assignments of the DNF with a $\theta$-cover corresponding to the satisfying assignments of the $s$ different terms of $f$, so $\dmax=O(n).$ and $\theta=1/n$.  Set $\eps = 1/n$ and condition on the event (which happens with probability at least 1/8) that the random walk $\bY_1, ..., \bY_M$ satisfies \Cref{lem:good-point-reached}. Now consider the iteration $t$ for which $\bY_t$ satisfies $T_\mix^{(ns)^{-O(1)}, \eps}(\bY_t,\mathcal{C}) \leq (ns)^{O(\log(s))}$ and $\ell = T_\mix^{(ns)^{-O(1)}, \eps}(\bY_t,\mathcal{C})$. If in this iteration, the event for local mixing, $E_{\bY_t}$, holds over all $2 \log(n)$ random walks, then we will correctly output the term $T$ with probability $1-o(1)$. (This is shown in the proof of \Cref{lem:far-point-learn}.) We can thus lower bound the probability of the compound event that \Cref{lem:good-point-reached} holds; all $2 \log(n)$ walks satisfy $E_{\bY_t}$; and we correctly add a term in $f$ to $\calL$, as occurring with probability at least
		\[
\frac{1}{8} \cdot \left ( \frac{1}{ns} \right)^{O(\log(n))} \cdot 		(1 - o(1)) = n^{-O(\log(ns))},
		\]
		and the proof is complete.
\end{proof}

\end{document}